\def \VersionLong {}
\def \VersionFinal {}
\ifdefined\VersionLong%
	\newcommand{\LongVersion}[1]{#1}
	\newcommand{\ShortVersion}[1]{}
\else
	\newcommand{\LongVersion}[1]{}
	\newcommand{\ShortVersion}[1]{#1}
\fi

\documentclass[a4paper,10pt,envcountsame]{llncs}
\usepackage{comment}
\ifdefined\VersionLong%
        \includecomment{LongVersionBlock}
        \excludecomment{ShortVersionBlock}
\else
        \excludecomment{LongVersionBlock}
        \includecomment{ShortVersionBlock}
\fi

\usepackage[ruled,lined,linesnumbered,noend]{algorithm2e}
\SetKwInOut{Input}{Input}
\SetKwInOut{Output}{Output}
\SetKw{KwRemove}{remove}
\SetKw{KwPush}{add}
\SetKw{KwPop}{pop}
\SetKw{KwFrom}{from}
\SetKw{KwTo}{to}
\SetKw{KwCompute}{compute}
\SetKw{KwReturn}{return}
\SetKw{KwBreak}{break}
\SetKw{KwPick}{pick}
\SetKw{KwLet}{let}
\SetKw{KwBe}{be}
\SetKw{KwSplit}{split}
\SetKw{KwInto}{into}
\SetKw{KwFind}{find}
\SetKwProg{Fn}{Function}{:}{}
\SetAlgorithmName{Algorithm}{algorithm}{List of Algorithms}

\usepackage{subcaption}
\usepackage{paralist} % inline lists

\usepackage{xspace}

\newenvironment{myitemize}
	{\ifdefined\VersionLong\begin{itemize}\else\begin{inparaitem}[]\fi}
	{\ifdefined\VersionLong\end{itemize}\else\end{inparaitem}\fi}

\usepackage{amsmath} % for \text et al., for ``cases'' environment…
\usepackage{amssymb} % for \mathbb et al.
\usepackage{stmaryrd} % for \llbracket and \rrbracket
\usepackage{mathdots}
\usepackage{centernot}
\usepackage{multirow}
\usepackage{mathtools}

\usepackage[misc,geometry]{ifsym} % For \Letter

\ifdefined\VersionWithComments%
	\usepackage{draftwatermark}
	\SetWatermarkText{draft}
    \SetWatermarkScale{15}
	\SetWatermarkColor[gray]{0.9}
\fi
\usepackage[svgnames,table]{xcolor}
\usepackage{colortbl}
\definecolor{darkblue}{rgb}{0.0,0.0,0.6}
\definecolor{darkgreen}{rgb}{0, 0.5, 0}
\definecolor{darkpurple}{rgb}{0.7, 0, 0.7}
\definecolor{darkblue}{rgb}{0, 0, 0.7}

\usepackage[
		colorlinks=true,
		citecolor=darkgreen,
		linkcolor=darkblue,
		urlcolor=darkpurple,
	]{hyperref}
\usepackage[capitalise,english,nameinlink]{cleveref} % load after algorithm2e, amsthm, and hyperref
\crefalias{AlgoLine}{line}
\crefname{line}{\text{line}}{\text{lines}} % to remove the capital
\crefname{item}{\text{item}}{\text{items}} % to remove the capital
\crefname{example}{\text{Example}}{\text{Examples}} % to remove the capital
\crefname{assumption}{\text{Assumption}}{\text{Assumptions}} % added Assumption
\crefname{algorithm}{\text{Algorithm}}{\text{Algorithms}}

\usepackage{tikz}
\usetikzlibrary{arrows,automata,positioning,math}
\tikzstyle{every node}=[initial text=]
\tikzstyle{location}=[rectangle, rounded corners, minimum size=12pt, draw=black, fill=blue!10, inner sep=2pt]
\tikzstyle{invariant}=[draw=black, dotted, inner sep=1pt] % xshift=1em,
\tikzstyle{final}=[double]
\tikzstyle{accepting}=[final]
\tikzstyle{PTPMOPT}=[,dashed,color=red,semithick]

\newcommand{\styleact}[1]{\ensuremath{\textcolor{coloract}{\mathrm{#1}}}}

\ifdefined\VersionAnonymous%
\definecolor{coloract}{rgb}{0.0, 0.0, 0.0}
\else
\definecolor{coloract}{rgb}{0.50, 0.70, 0.30}
\fi
\definecolor{colorclock}{rgb}{0.4, 0.4, 1}
\definecolor{colorconst}{rgb}{0.50, 0.20, 0.00}
\definecolor{colordisc}{rgb}{1, 0, 1}
\definecolor{colorloc}{rgb}{0.4, 0.4, 0.65}
\definecolor{colorparam}{rgb}{1, 0.6, 0.0}

\usepackage{gnuplot-lua-tikz}
\usepackage{pgfplotstable}
\pgfplotsset{compat=1.12}
\usepackage{booktabs}

\usetikzlibrary{automata,positioning,matrix,shapes.callouts}
\tikzset{
accepting/.style={double distance=1pt}
}

\newcommand{\N}{{\mathbb{N}}}

\newcommand{\R}{{\mathbb{R}}}

\newcommand{\Rp}{{\mathbb{R}_{>0}}}

\newcommand{\ttrue}{\mathrm{t{\kern-1.5pt}t}}
\newcommand{\ffalse}{\mathrm{f{\kern-1.5pt}f}}

\newcommand{\Rnn}{\R_{\ge 0}}

\newcommand{\powerset}[1]{\mathcal{P} ({#1})}

\newcommand{\imply}{\Rightarrow}

\makeatletter
\newcommand{\figcaption}[1]{\def\@captype{figure}\caption{#1}}
\newcommand{\tblcaption}[1]{\def\@captype{table}\caption{#1}}
\makeatother

\usepackage{arydshln}
\newcommand{\recallResult}[2]
{%
	\smallskip

	\noindent\fcolorbox{black}{green!15}{
		\begin{minipage}{.95\columnwidth}
			\noindent\textbf{\cref{#1} (recalled).}
			{\em{}#2}
		\end{minipage}
	}

	\smallskip
}

\ifdefined\VersionWithComments%
	\usepackage[colorinlistoftodos,textsize=footnotesize]{todonotes}
\else
	\usepackage[disable]{todonotes}
\fi
\newcommand{\gennote}[3]{\todo[linecolor=#2,backgroundcolor=#2!25,bordercolor=#2]{#3: #1}}

\newcommand{\mw}[1]{\gennote{#1}{orange}{MW}}
\newcommand{\instructions}[1]{{\gennote{\bfseries #1}{red}{Instructions}}}

\ifdefined\VersionFinal%
\else
	\usepackage[pagewise]{lineno} % switch, modulo
	\linenumbers%
	
\fi

\newcommand{\ourTool}{\textsc{LearnTA}}
\newcommand{\DOTA}{\textsc{OneSMT}}
\newcommand{\Random}{\textsf{Random}}
\newcommand{\unbalanced}{\textsf{Unbalanced}}
\newcommand{\AKM}{\textsf{AKM}}
\newcommand{\CAS}{\textsf{CAS}}
\newcommand{\light}{\textsf{Light}}
\newcommand{\PC}{\textsf{PC}}
\newcommand{\TCP}{\textsf{TCP}}
\newcommand{\Train}{\textsf{Train}}
\newcommand{\FDDI}{\textsf{FDDI}}
\newcommand{\setdiff}{\triangle}
\newcommand{\integer}{\mathrm{int}}
\newcommand{\fractional}{\mathrm{frac}}

\newcommand{\project}[2]{\ensuremath{#1{\downarrow_{#2}}}}
\newcommand{\timedomain}{\Rnn}
\newcommand{\Alphabet}{\Sigma}
\newcommand{\TimedWords}{\mathcal{T}(\Alphabet)}
\newcommand{\word}[1][]{w#1}
\newcommand{\action}{a}
\newcommand{\actionSequence}{\action_1,\action_2,\dots,\action_n}
\newcommand{\timestamp}{\tau}
\newcommand{\timestampSequence}[1][]{\timestamp#1_0,\timestamp#1_1,\dots,\timestamp#1_{n#1}}
\newcommand{\wordInside}[1][]{\timestamp#1_0 \action#1_1 \timestamp#1_1 \action#1_2 \dots \action#1_{n#1} \timestamp#1_{n#1}}
\newcommand{\wordWithInside}[1][]{\word[#1]=\wordInside[#1]}
\newcommand{\Lg}{\mathcal{L}}
\newcommand{\untimed}[1]{{\mu({#1})}}
\newcommand{\duration}[1]{\lambda(#1)}
\newcommand{\emptyword}{\varepsilon}
\newcommand{\timeValuation}[1]{\kappa(#1)}
\newcommand{\Clock}{C}
\newcommand{\clock}{c}

\newcommand{\cval}{\nu}
\newcommand{\CVal}{\clockvaluations}
\newcommand{\dConstant}{d}
\newcommand{\clockvaluations}[1][\Clock]{(\timedomain)^{#1}}
\newcommand{\zerovalue}[1][C]{\mathbf{0}_{#1}}

\newcommand{\loc}{l}
\newcommand{\Loc}{L}
\newcommand{\initLoc}{\loc_0}

\newcommand{\Final}{F}
\newcommand{\A}{\mathcal{A}}
\newcommand{\TA}{\A}
\newcommand{\TAInside}{(\Alphabet,\Loc,\initLoc,\Clock,\Invariant,\Edge,\Final)}
\newcommand{\Invariant}{I}
\newcommand{\InvariantAt}[1]{I ({#1})}
\newcommand{\TAWithInside}{\TA=\TAInside}

\newcommand{\edgeInside}{(\loc,\guard,\action,\resets,\loc')}
\newcommand{\Edge}{\Delta}

\newcommand{\Constraint}{\mathcal{C}_{\Clock}}
\newcommand{\ConstraintWithClock}{\Constraint}
\newcommand{\constraint}{\varphi}
\newcommand{\Guard}{\mathcal{G}_{\Clock}}
\newcommand{\GuardWithClock}{\Guard}
\newcommand{\guard}{g}
\newcommand{\Resets}{\powerset{\Clock}}
\newcommand{\resets}{\rho}

\newcommand{\reset}[2]{#1[#2 \coloneqq 0]}

\newcommand{\TTS}{\mathcal{S}}
\newcommand{\TTSstate}{q}
\newcommand{\TTSState}{Q}
\newcommand{\initTTSState}{\TTSstate_{0}}

\newcommand{\AccTTSState}{\TTSState_{F}}
\newcommand{\TTSTransitionRel}{\to}
\newcommand{\TTSTransition}{{\TTSTransitionRel}}
\newcommand{\TTStransitionRel}{\to}
\newcommand{\TTStransition}{{\TTStransitionRel}}
\newcommand{\TTSInside}{(\TTSState,\initTTSState,\AccTTSState,\TTSTransition)}
\newcommand{\TTSWithInside}{\TTS{} = \TTSInside}
\newcommand{\TTStransitionWithLabel}[1]{\stackrel{#1}{\TTStransition}}
\newcommand{\TTStransitionRelWithLabel}[1]{\stackrel{#1}{\TTStransitionRel}}
\newcommand{\runInside}{\TTSstate_0, \TTStransition_1,\TTSstate_1,\dots, \TTStransition_n,\TTSstate_n}

\newcommand{\region}[1]{{#1}^{\mathrm{r}}}
\newcommand{\regionAutom}{\region{\A}}
\newcommand{\RegionState}{\region{\TTSState}}
\newcommand{\regionState}[1][]{\region{\TTSstate#1}}

\newcommand{\initRegionState}{\region{\TTSstate}_{0}}
\newcommand{\AccRegionState}{\region{\TTSState}_{F}}
\newcommand{\RegionTransition}{E}
\newcommand{\regionAutomInside}{(\Alphabet, \RegionState, \initRegionState, \AccRegionState, \RegionTransition)}
\newcommand{\regionAutomWithInside}{\regionAutom{} = \regionAutomInside}
\newcommand{\regionEquiv}{\approx}
\newcommand{\maxConstant}{K}
\newcommand{\timedCondition}[1][]{\Lambda#1}
\newcommand{\elementaryInside}[1][]{u#1, \timedCondition[#1]}

\newcommand{\elementary}[1][]{(\elementaryInside[#1])}
\newcommand{\Elementary}[1][\Alphabet]{\mathcal{E} (#1)}
\newcommand{\SimpleElementary}[1][\Alphabet]{\mathcal{SE} (#1)}
\newcommand{\CRM}{\Phi}
\newcommand{\CRMTuple}{(\elementaryInside, \elementaryInside['], R)}
\newcommand{\RecognizableFinal}{F}
\newcommand{\timeVariables}{\ensuremath{\mathbb{T}}}
\newcommand{\sumTimestamp}[3][]{\ensuremath{\timeVariables#1_{#2,#3}}}

\newcommand{\sem}[1]{\llbracket{ #1 }\rrbracket}

\newcommand{\fractionalCondition}[1][]{\Theta#1}
\newcommand{\fractionalElementaryInside}[1][]{\elementaryInside[#1]}
\newcommand{\fractionalElementary}[1][]{(\fractionalElementaryInside[#1])}

\newcommand{\sumTimestampSequence}[1][]{\sumTimestamp#1{0}{n#1},\sumTimestamp#1{1}{n#1},\dots,\sumTimestamp#1{n#1}{n#1}}
\newcommand{\fractionalSequence}[1][]{\fractional(\sumTimestamp#1{0}{n#1}),\fractional(\sumTimestamp#1{1}{n#1}),\dots,\fractional(\sumTimestamp#1{n#1}{n#1})}

\newcommand{\memQKey}[1][\Lg]{\mathtt{mem}_{#1}}
\newcommand{\memQ}[2][\Lg]{\memQKey[#1](#2)}
\newcommand{\eqQKey}[1][\Lg]{\mathtt{eq}_{#1}}
\newcommand{\eqQ}[2][\Lg]{\eqQKey[#1](#2)}
\newcommand{\nerode}[1]{\equiv_{#1}}
\newcommand{\symbolicMemQ}[1]{\mathtt{mem}^{\mathtt{sym}}_{#1}}
\newcommand{\Rename}{R}
\newcommand{\rename}[2]{#1 \land #2}
\newcommand{\prefix}{p}
\newcommand{\PrefixSet}{P}
\newcommand{\suffix}{s}
\newcommand{\SuffixSet}{S}

\newcommand{\exteriorSymbol}[1][]{\mathrm{ext}^{#1}}
\newcommand{\exterior}[2][]{\exteriorSymbol[#1] (#2)}
\newcommand{\successorSymbol}[1][]{\mathrm{succ}^{#1}}
\newcommand{\successor}[2][]{\successorSymbol[#1] (#2)}

\newcommand{\predecessorSymbol}[1][]{\mathrm{pred}^{#1}}
\newcommand{\predecessor}[2][]{\predecessorSymbol[#1] (#2)}
\newcommand{\ObsTable}[1][]{O#1}
\newcommand{\ObsTableInside}[1][]{(\PrefixSet#1, \SuffixSet#1, \Table#1)}
\newcommand{\ObsTableWithInside}[1][]{\ObsTable[#1] = \ObsTableInside[#1]}

\newcommand{\TimedObsTableInsideInside}[1][]{\PrefixSet#1, \SuffixSet#1, \Table#1}
\newcommand{\TimedObsTableInside}[1][]{(\TimedObsTableInsideInside[#1])}

\newcommand{\targetLg}{\Lg_{\mathrm{tgt}}}
\newcommand{\targetA}{\A_{\mathrm{tgt}}}
\newcommand{\hypothesisA}{\A_{\mathrm{hyp}}}
\newcommand{\Table}{T}
\newcommand{\TableCell}[3][]{\Table#1(#2, #3)}

\newcommand{\CellRel}{\Rename}
\newcommand{\CellSub}{\sqsubseteq}

\newcommand{\CellSim}{\sim}
\newcommand{\NCellSim}{{\not\sim}}

\newcommand{\cex}{\mathit{cex}}

\newcommand{\shiftConstraint}[1]{\mathrm{shift} (#1)}

\usepackage{colortbl}
\newcommand{\tbcolor}{\cellcolor{green!25}\bf}

\ifdefined\VersionWithComments%
 	\definecolor{colorok}{RGB}{80,80,150}
\else
	\definecolor{colorok}{RGB}{0,0,0}
\fi

\newcommand{\eg}{\textcolor{colorok}{e.\,g.,}\xspace}
\newcommand{\ie}{\textcolor{colorok}{i.\,e.,}\xspace}

\newcommand{\resp}{\textcolor{colorok}{resp.}\xspace}

\makeatletter
\def\orcidID#1{\smash{\href{https://orcid.org/#1}{\protect\raisebox{-1.25pt}{\protect\includegraphics{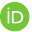}}}}}
\makeatother

\makeatletter
\AtBeginDocument{%
  \@ifpackageloaded{hyperref}
  {\def\@doi#1{\href{https://doi.org/#1}
      {\ttfamily https://doi.org/#1}\egroup}}
  {\def\@doi#1{\ttfamily https://doi.org/#1\egroup}}
  \def\doi{\bgroup\catcode`\_=12\relax\@doi}}
\makeatother

\pdfoutput=1 %uncomment to ensure pdflatex processing (mandatatory e.g. to submit to arXiv)

\title{Active Learning of Deterministic Timed Automata with Myhill-Nerode Style Characterization} %TODO Please add

\titlerunning{Active Learning of DTAs with Myhill-Nerode Style Characterization} %TODO optional, please use if title is longer than one line

\author{
\ifdefined\VersionAnonymous%
\else
Masaki Waga\orcidID{0000-0001-9360-7490}\LongVersion{\thanks{%
    This is the author (and extended) version of the manuscript of the same name published in the proceedings of the 35th International Conference on Computer Aided Verification (CAV 2023).
    The final version is available at \url{www.springer.com}.
    }%
}
\fi
}
\institute{%
\ifdefined\VersionAnonymous%
\else
Graduate School of Informatics, Kyoto University, Kyoto, Japan
\fi
}

\usepackage{graphicx}	% required for `\scalebox' (yatex added)

\ifdefined\VersionFinal%
\ifdefined\VersionLong%
\else
\usepackage[firstpage]{draftwatermark}  % free badge placement
\SetWatermarkAngle{0}
\SetWatermarkText{\raisebox{17.0cm}{%
  \hspace{0.75cm}%
  \href{https://doi.org/10.5281/zenodo.7875383}{\includegraphics{./badges/1-available}}% 
  \hspace{8.5cm}%
  \includegraphics{./badges/3-reusable}%
}}
\fi
\fi
\begin{document}

\maketitle

%TODO mandatory: add short abstract of the document
\begin{abstract}
We present an algorithm to learn a deterministic timed automaton (DTA) via membership and equivalence queries.
Our algorithm is an extension of the L* algorithm with a Myhill-Nerode style characterization of recognizable timed languages,
which is the class of timed languages recognizable by DTAs.
We first characterize the recognizable timed languages with a Nerode-style congruence.
Using it, we give an algorithm with a smart teacher answering \emph{symbolic} membership queries in addition to membership and equivalence queries.
With a symbolic membership query, one can ask the membership of a certain set of timed words at one time.
We prove that for any recognizable timed language, our learning algorithm returns a DTA recognizing it.
We show how to answer a symbolic membership query with finitely many membership queries.
We also show that our learning algorithm requires a polynomial number of queries with a smart teacher and an exponential number of queries with a normal teacher.
We applied our algorithm to various benchmarks and confirmed its effectiveness with a normal teacher.
\end{abstract}

\instructions{(CAV'23) Regular papers: 20 pages, excluding bibliography (i.e., 2 extra pages)}
\mw{hello}

%%%%%%%%%%%%%%%%%%%%%%%%%%%%%%%%%%%%%%%%%%%%%%%%%%%%%%%%%%%%
%%%%%%%%%%%%%%%%%%%%%%%%%%%%%%%%%%%%%%%%%%%%%%%%%%%%%%%%%%%%
\section{Introduction}\label{section:introduction}
%%%%%%%%%%%%%%%%%%%%%%%%%%%%%%%%%%%%%%%%%%%%%%%%%%%%%%%%%%%%
%%%%%%%%%%%%%%%%%%%%%%%%%%%%%%%%%%%%%%%%%%%%%%%%%%%%%%%%%%%%

% **** Introduction & history of active automata learning in general
\emph{Active automata learning} is a class of methods to infer an automaton recognizing an unknown target language $\targetLg \subseteq \Alphabet^*$ through finitely many queries to a teacher.
The L* algorithm~\cite{Angluin87}, the best-known active DFA learning algorithm, infers the minimum DFA recognizing $\targetLg$ using \emph{membership} and \emph{equivalence} queries.
In a membership query, the learner asks if a word $\word \in \Alphabet^*$ is in the target language $\targetLg$, which is used to obtain enough information to construct a hypothesis DFA $\hypothesisA$.
Using an equivalence query, the learner checks if the hypothesis $\hypothesisA$ recognizes the target language $\targetLg$.
If $\Lg(\hypothesisA) \neq \targetLg$, the teacher returns a counterexample  $\cex \in \targetLg \setdiff \Lg(\hypothesisA)$ differentiating the target language and the current hypothesis.
The learner uses $\cex$ to update $\hypothesisA$ to classify $\cex$ correctly.
Such a learning algorithm has been combined with formal verification, \eg{} for testing~\cite{PVY99,MP19,Waga20,SWS21} and controller synthesis~\cite{DBLP:journals/tase/ZhangFL20}.

% **** Introduce Myhill/Nerode congruence as a theoretical background of the L* algorithm
Most of the DFA learning algorithms rely on the characterization of regular languages by \emph{Nerode's congruence}.
For a language $\Lg$, words $\prefix$ and $\prefix'$ are equivalent if for any extension $\suffix$,
$\prefix \cdot \suffix \in \Lg$ if and only if $\prefix' \cdot \suffix \in \Lg$.
It is well known that if $\Lg$ is regular, such an equivalence relation has finite classes, corresponding to the locations of the minimum DFA recognizing $\Lg$ (known as \emph{Myhill-Nerode theorem}; see, \eg{}~\cite{HMU07}).
Moreover, for any regular language $\Lg$, there are finite extensions $\SuffixSet$ such that $\prefix$ and $\prefix'$ are equivalent if and only if
for any $\suffix \in \SuffixSet$,
$\prefix \cdot \suffix \in \Lg$ if and only if $\prefix' \cdot \suffix \in \Lg$.
Therefore, one can learn the minimum DFA by learning such finite extensions $\SuffixSet$ and the finite classes induced by Nerode's congruence.

%-----------------------------------------------------------
\begin{figure}[tbp]
 \begin{subfigure}{0.18\linewidth}
  \centering
  \begin{tikzpicture}[shorten >=1pt,scale=0.8,every node/.style={transform shape},every initial by arrow/.style={initial text={}}]
   \def\distance{1.6}
  %% locations
  \node[initial,state,accepting] (l0) at (0,0) [align=center]{$\loc_0$};
  \node[state] (l1) at (\distance,0) [align=center]{$\loc_1$};
  \node[state] (l2) at (0,-\distance) [align=center]{$\loc_2$};
  \node[state] (l3) at (\distance,-\distance) [align=center]{$\loc_3$};

  %% edges
  \path[->]
  (l0) edge [above, bend left=10] node {$\styleact{a}$} (l1)
  (l0) edge [right,bend left=10] node {$\styleact{b}$} (l2)
  (l1) edge [right,bend left=10] node {$\styleact{a}$} (l3)
  (l1) edge [below, bend left=10] node {$\styleact{b}$} (l0)
  (l2) edge [left,bend left=10] node {$\styleact{a}$} (l0)
  (l2) edge [bend left=10] node[above] {$\styleact{b}$} (l3)
  (l3) edge [bend left=10] node[below] {$\styleact{a}$} (l2)
  (l3) edge [left,bend left=10] node {$\styleact{b}$} (l1)
  ;
  \end{tikzpicture}
  \caption{A DFA $\A$}
 \end{subfigure}
 \hfill
 %-----------------------------------------------------------
 \begin{subfigure}{.45\textwidth}
  \scriptsize
  \centering
  \begin{tabular}{c|c c}
   & $\styleact{\emptyword}$ & $\styleact{a}$\\\hline
   $\styleact{\emptyword}$& $\top$ & $\bot$\\
   $\styleact{a}$& $\bot$ & $\bot$\\
   \scalebox{0.8}{$\vdots$}&\\
   $\styleact{a}\styleact{a}$& $\bot$ & $\bot$\\
   % \scalebox{0.8}{$\vdots$}&\\
  \end{tabular}
  $\xrightarrow{\text{add $\styleact{b}$ to $\SuffixSet$}}$
  \begin{tabular}{c|c c c}
   & $\styleact{\emptyword}$ & $\styleact{a}$ & $\styleact{b}$\\\hline
   $\styleact{\emptyword}$& $\top$ & $\bot$ & $\bot$\\
   $\styleact{a}$& $\bot$ & $\bot$ & $\top$\\
   \scalebox{0.8}{$\vdots$}&\\
   $\styleact{a}\styleact{a}$& $\bot$ & $\bot$ & $\bot$\\
   % \scalebox{0.8}{$\vdots$}&\\
  \end{tabular}
  \caption{Intermediate observation tables for learning $\A$. $\styleact{a}$ and $\styleact{a}\styleact{a}$ are deemed equivalent with extensions $\SuffixSet = \{\styleact{\emptyword}, \styleact{a}\}$ but distinguished with $\SuffixSet = \{\styleact{\emptyword}, \styleact{a}, \styleact{b}\}$.}%
  \label{figure:observation_table_illustration}
 \end{subfigure}
 %-----------------------------------------------------------
 \hfill
 \begin{subfigure}{0.30\linewidth}
  \centering
  \begin{tikzpicture}[shorten >=1pt,scale=0.8,every node/.style={transform shape},every initial by arrow/.style={initial text={}}]
   %% locations
   \node[initial,state,accepting] (l0) at (0,0) [align=center]{$\loc_0$};
   \node[state] (l1) at (3.5,0) [align=center]{$\loc_1$};

   %% edges
   \path[->]
   (l0) edge [bend left=10] node[above] {$\styleact{a}, \clock \geq 1 / \clock \coloneqq 0$} (l1)
   (l0) edge [loop above] node {$\styleact{a}, \clock < 1$} (l0)
   (l1) edge [bend left=10] node[below] {$\styleact{a}, \clock \leq 1$} (l0)
   (l1) edge [loop above] node {$\styleact{a}, \clock > 1$} (l1)
   ;
  \end{tikzpicture}
  \caption{A DTA $\A'$ with one clock variable $\clock$}%
  \label{figure:timed_automaton}
 \end{subfigure}
 %-----------------------------------------------------------
 \begin{subfigure}{1.0\linewidth}
  \scriptsize
  \centering
  \begin{tabular}{c|c c}
   & $\{\timestamp'_0 \mid \timestamp'_0 = 0\}$ & $\{\timestamp'_0 \styleact{a} \mid \timestamp'_0 \in (0,1)\}$\\\hline
   $\{\timestamp_0 \mid \timestamp_0 = 0\}$ & $\top$ & $\top$\\
   $\{\timestamp_0 \mid \timestamp_0 \in (0,1)\}$ & $\top$ & $\timestamp_0 + \timestamp'_0 \in (0,1)$\\
   \scalebox{0.8}{$\vdots$}&\\
   $\{\timestamp_0 \styleact{a} \timestamp_1 \mid \timestamp_0 \in (0,1), \timestamp_1 \in (0,1), \timestamp_0 + \timestamp_1 \in (0,1)\} (= \prefix_1)$ & $\top$ & $\timestamp_0 + \timestamp_1 + \timestamp'_0 \in (0,1)$\\
   \scalebox{0.8}{$\vdots$}&\\
   $\{\timestamp_0 \styleact{a} \timestamp_1 \styleact{a} \timestamp_2 \mid \timestamp_0 \in (1,2), \timestamp_1 \in (0,1), \timestamp_2 \in (0,1), \timestamp_1 + \timestamp_2\in (0,1)\} (= \prefix_2)$ & $\top$ & $\timestamp_1 + \timestamp_2 + \timestamp'_0 \in (0,1)$\\
   \scalebox{0.8}{$\vdots$}&\\
  \end{tabular}
  \caption{Timed observation table for learning $\A'$. Each cell is indexed by a pair $(\prefix, \suffix) \in \PrefixSet \times \SuffixSet$ of elementary languages. The cell indexed by $(\prefix, \suffix)$ shows a constraint $\timedCondition$ such that $\word \in \prefix \cdot \suffix$ satisfies $\word \in \targetLg$ if and only if $\timedCondition$ holds. Elementary languages $\prefix_1$ and $\prefix_2$ are deemed equivalent with the equation $\tau^1_0 + \tau^1_1 = \tau^2_1 + \tau^2_2$, where $\tau^j_i$ represents $\tau_i$ in $\prefix_j$.}%
  \label{figure:timed_observation_table_illustration}
 \end{subfigure}
 \caption{Illustration of observation tables in the L* algorithm for DFA learning (\cref{figure:observation_table_illustration}) and our algorithm for DTA learning (\cref{figure:timed_observation_table_illustration})}%
 \label{figure:comparison}
\end{figure}
%-----------------------------------------------------------

% Summary of the original L* algorithm
The L* algorithm learns the minimum DFA recognizing the target language $\targetLg$ using a 2-dimensional array called an \emph{observation table}.
\cref{figure:observation_table_illustration} illustrates observation tables.
The rows and columns of an observation table are indexed with finite sets of words $\PrefixSet$ and $\SuffixSet$, respectively.
Each cell indexed by $(\prefix, \suffix) \in \PrefixSet \times \SuffixSet$ shows if $\prefix \cdot \suffix \in \targetLg$.
The column indices $\SuffixSet$ are the current extensions approximating Nerode's congruence.
The L* algorithm increases $\PrefixSet$ and $\SuffixSet$ until: 1) the equivalence relation defined by $\SuffixSet$ converges to Nerode's congruence and 2) $\PrefixSet$ covers all the classes induced by the congruence.
The equivalence between $\prefix, \prefix'\in\PrefixSet$ under $\SuffixSet$ can be checked by comparing the rows in the observation table indexed with $\prefix$ and $\prefix'$.
For example, \cref{figure:observation_table_illustration} shows that $\styleact{a}$ and $\styleact{a}\styleact{a}$ are deemed equivalent with extensions $\SuffixSet = \{\styleact{\emptyword}, \styleact{a}\}$ but distinguished by adding $\styleact{b}$ to $\SuffixSet$.
The refinement of $\PrefixSet$ and $\SuffixSet$ is driven by certain conditions to validate the DFA construction and 
% are refined by making the observation table \emph{cohesive} and 
by addressing the counterexample obtained by an equivalence query.
 % $\PrefixSet$ and $\SuffixSet$ are refined by making the observation table \emph{cohesive} and addressing the counterexample obtained by an equivalence query.

% **** Show that learning of timed languages is hard
\emph{Timed} words are extensions of conventional words with real-valued dwell time between events.
\emph{Timed} languages, sets of timed words, are widely used to formalize real-time systems and their properties, \eg{} for formal verification.
Among various formalisms representing timed languages, \emph{timed automata (TAs)}~\cite{AD94} is one of the widely used formalisms.
% \mw{Maybe we can shorten here}
% Various formalisms has been proposed to represent timed languages, typically by extending conventional automata or temporal logic with an additional structure to represent timing constraints.
A TA is an extension of an NFA
% \emph{Timed automata} (TAs)~\cite{AD94} are extensions of NFAs
 with finitely many clock variables to represent timing constraints.
\cref{figure:timed_automaton} shows an example.

Despite its practical relevance, learning algorithms for TAs are only available for limited subclasses of TAs, \eg{} real-time automata~\cite{AnWZZZ21,AnZZZ21}, event-recording automata~\cite{GJL10,GJP06}, event-recording automata with unobservable reset~\cite{HJM20}, and one-clock deterministic TAs~\cite{ACZZZ20,XAZ22}.
Timing constraints representable by these classes are limited, \eg{}
by restricting the number of clock variables or
by restricting the edges where a clock variable can be reset.
Such restriction simplifies the inference of timing constraints in learning algorithms.
%

% *** Summary of what we did
\paragraph{Contributions}
In this paper, we propose an active learning algorithm for \emph{deterministic} TAs (DTAs).
The languages recognizable by DTAs are called \emph{recognizable timed languages}~\cite{MP04}.
Our strategy is as follows: first, we develop a Myhill-Nerode style characterization of recognizable timed languages;
then, we extend the L* algorithm for recognizable timed languages using the similarity of the Myhill-Nerode style characterization.

Due to the continuity of dwell time in timed words, it is hard to characterize recognizable timed languages by a Nerode-style congruence between timed words.
For example, for the DTA in \cref{figure:timed_automaton}, for any $\tau, \tau' \in [0,1)$ satisfying $\tau < \tau'$,
$(1 - \tau') \styleact{a}$ distinguishes $\tau$ and $\tau'$ because $\tau (1 - \tau') \styleact{a}$ leads to $\loc_0$ while $\tau (1 - \tau) \styleact{a}$ leads to $\loc_1$.
Therefore, such a congruence can make \emph{infinitely} many classes.

Instead, we define a Nerode-style congruence between sets of timed words called \emph{elementary languages}~\cite{MP04}.
An elementary language is a timed language defined by a word with a conjunction of inequalities constraining the time difference between events.
We also use an equality constraint, which we call, a \emph{renaming equation} to define the congruence.
Intuitively, a renaming equation bridges the time differences in an elementary language and the clock variables in a TA.\@
We note that there can be multiple renaming equations showing the equivalence of two elementary languages.

%
% %----------------------------------------------------------
% \begin{example}
%  Let $\prefix$ and $\prefix'$ be elementary languages $\prefix = \{\tau_0 a \tau_1 \mid \tau_0 \in (0,1), \tau_1 \in (0,1), \tau_0 + \tau_1\in (0,1)\}$ and $\prefix' = \{\tau'_0 a \tau'_1 a \tau'_2 \mid \tau'_0 \in (1,2), \tau'_1 \in (0,1), \tau'_2 \in (0,1), \tau'_1 + \tau'_2\in (0,1)\}$.
%  For the DTA in \cref{figure:timed_automaton}, $\prefix$ and $\prefix'$ are equivalent with $\tau_0 + \tau_1 = \tau'_1 + \tau'_2$
%  because: 1) we reach $\loc_0$ after reading either of $\prefix$ and $\prefix'$ and 2) the value of $\clock$ after reading $\prefix$ and $\prefix'$ is $\tau_0 + \tau_1$ and $\tau'_1 + \tau'_2$, respectively.
% \end{example}
% %----------------------------------------------------------
%
%----------------------------------------------------------
\begin{example}%
 \label{example:intuition_congruence}
 % \ShortVersion{For example, let}
 Let $\prefix_1$ and $\prefix_2$ be elementary languages $\prefix_1 = \{\tau^1_0 \styleact{a} \tau^1_1 \mid \tau^1_0 \in (0,1), \tau^1_1 \in (0,1), \tau^1_0 + \tau^1_1\in (0,1)\}$ and $\prefix_2 = \{\tau^2_0 \styleact{a} \tau^2_1 \styleact{a} \tau^2_2 \mid \tau^2_0 \in (1,2), \tau^2_1 \in (0,1), \tau^2_2 \in (0,1), \tau^2_1 + \tau^2_2\in (0,1)\}$.
 For the DTA in \cref{figure:timed_automaton}, $\prefix_1$ and $\prefix_2$ are equivalent with the renaming equation $\tau^1_0 + \tau^1_1 = \tau^2_1 + \tau^2_2$
 because for any $\word_1 = \tau^1_0 \styleact{a} \tau^1_1 \in \prefix_1$ and $\word_2 = \tau^2_0 \styleact{a} \tau^2_1 \styleact{a} \tau^2_2 \in \prefix_2$:
 1) we reach $\loc_0$ after reading either of $\word_1$ and $\word_2$ and 2) the values of $\clock$ after reading $\word_1$ and $\word_2$ are $\tau^1_0 + \tau^1_1$ and $\tau^2_1 + \tau^2_2$, respectively.
\end{example}
%----------------------------------------------------------

We characterize recognizable timed languages by the finiteness of the equivalence classes defined by the above congruence.
We also show that for any recognizable timed language, there is a finite set $\SuffixSet$ of elementary languages such that the equivalence of any prefixes can be checked by the extensions $\SuffixSet$.

% *** Outline of the learning algorithm
By using the above congruence, we extend the L* algorithm for DTAs.
The high-level idea is the same as the original L* algorithm: 1) the learner makes membership queries to obtain enough information to construct a hypothesis DTA $\hypothesisA$ and 2) the learner makes an equivalence query to check if $\hypothesisA$ recognizes the target language.
The largest difference is in the cells of an observation table.
Since the concatenation $\prefix \cdot \suffix$ of an index pair $(\prefix, \suffix) \in \PrefixSet \times \SuffixSet$ is not a timed word but a set of timed words,
its membership is not defined as a Boolean value.
Instead, we introduce the notion of \emph{symbolic} membership and use it as the value of each cell of the \emph{timed} observation table.
Intuitively, the symbolic membership is the constraint representing the subset of $\prefix \cdot \suffix$ included by $\targetLg$.
Such a constraint can be constructed by finitely many (non-symbolic) membership queries.

%----------------------------------------------------------
\begin{example}%
 \label{example:intuition_timed_observation_table}
 \cref{figure:timed_observation_table_illustration} illustrates a \emph{timed} observation table.
 %where the cell indexed by $(\prefix, \suffix)$ shows the symbolic membership of $\prefix \cdot \suffix$ to $\targetLg$.
 % Similarly to the L* algorithm, 
 The equivalence between $\prefix_1, \prefix_2 \in \PrefixSet$ under $\SuffixSet$ can be checked by comparing the cells in the rows indexed with $\prefix_1$ and $\prefix_2$ with renaming equations.
 For the cells in rows indexed by $\prefix_1$ and $\prefix_2$,
 their constraints are the same by replacing $\tau_0 + \tau_1$ with $\tau_1 + \tau_2$ and vice versa.
 Thus, $\prefix_1$ and $\prefix_2$ are equivalent with the current extensions $\SuffixSet$.
\end{example}
%----------------------------------------------------------

% On the DTA construction
Once the learner obtains enough information, it constructs a DTA via the monoid-based representation of recognizable timed languages~\cite{MP04}.
%Translation of a timed observation table to this monoid-based representation is straightforward.
%
% Some theoretical results
We show that for any recognizable timed language, our algorithm terminates and returns a DTA recognizing it.
We also show that the number of the necessary queries is polynomial to the size of the equivalence class defined by the Nerode-style congruence if symbolic membership queries are allowed and, otherwise, exponential to it.
Moreover, if symbolic membership queries are not allowed, the number of the necessary queries is at most
doubly exponential to the number of the clock variable of a DTA recognizing the target language and
singly exponential to the number of locations of a DTA recognizing the target language.
This worst-case complexity is the same as the one-clock DTA learning algorithm in~\cite{XAZ22}.

% Summary of the experiments
We implemented our DTA learning algorithm in a prototype library \ourTool{}.
Our experiment results show that it is efficient enough for some benchmarks taken from practical applications, \eg{} the FDDI protocol.
This suggests the practical relevance of our algorithm.

% Summary of the contributions
The following summarizes our contribution.
%----------------------------------------------------------
\begin{itemize}
 \item We characterize recognizable timed languages by a Nerode-style congruence.
 % \item We give a Myhill-Nerode style characterization of recognizable timed languages.
 \item Using the above characterization, we give an active DTA learning algorithm.
 \item Our experiment results suggest its practical relevance.
 % \item We prove that the number of the queries required in our DTA learning is polynomial to the size of the equivalence class defined by the Nerode-style congruence if symbolic membership queries are allowed, and otherwise, exponential to it.
\end{itemize}
%----------------------------------------------------------

%%%%%%%%%%%%%%%%%%%%%%%%%%%%%%%%%%%%%%%%%%%%%%%%%%%%%%%%%%%%
%%%%%%%%%%%%%%%%%%%%%%%%%%%%%%%%%%%%%%%%%%%%%%%%%%%%%%%%%%%%
\paragraph{Related work}\label{sec:related_work}
%%%%%%%%%%%%%%%%%%%%%%%%%%%%%%%%%%%%%%%%%%%%%%%%%%%%%%%%%%%%
%%%%%%%%%%%%%%%%%%%%%%%%%%%%%%%%%%%%%%%%%%%%%%%%%%%%%%%%%%%%

Among various characterization of timed languages~\cite{AD94,BPT01,ACM02,MP04,BL12,BRP17},
the characterization by \emph{recognizability}~\cite{MP04} is closest to our Myhill-Nerode-style characterization.
Both of them use finite sets of\LongVersion{ prefix} elementary languages for characterization.
Their main difference is that
\cite{MP04} proposes a formalism to define a timed language by relating prefixes by a morphism, whereas 
we propose a technical gadget to define an equivalence relation over timed words with respect to suffixes using symbolic membership.
This difference makes our definition suitable for an L*-style algorithm, where
the original L* algorithm is based on Nerode's congruence, which defines an equivalence relation over words with respect to suffixes using conventional membership.

As we have discussed so far, active TA learning~\cite{GJL10,GJP06,ACZZZ20,XAZ22,HJM20} has been studied mostly for limited subclasses of TAs,
where
 the number of the clock variables or 
 the clock variables reset at each edge is fixed.
% the timing constraints represented by the clock variables.
\begin{comment}
 Active learning has been studied for some subclasses of TAs, \eg{} event-recording automata~\cite{GJL10,GJP06}, one-clock DTAs~\cite{ACZZZ20,XAZ22}, and event-recording automata with unobservable reset~\cite{HJM20}.
 Since TAs are equipped with clock variables, in addition to the discrete structure of an automaton,
 the number of the clock variables,
 the clock variables reset at each edge, and
the timing constraints represented by the clock variables must be inferred.
\end{comment}
% In timed automata learning, we need to infer
% \begin{ienumeration}
%  \item the set of clock variables; and
%  \item the clock variables reset at each transition.
% \end{ienumeration}
% For the subclasses above, at least one of these features is fixed, making learning easier.
% For instance, event-recording automata memorizes only the dwell time from the latest occurrence of each event, and thus,
% the number of the clock variables and the clock variables reset at each transition are obvious from the events~\cite{GJL10,GJP06}.t
In contrast, our algorithm infers both of the above information.
% This contrasts with our setting, where all these features must be inferred.
Another difference is in the technical strategy.
Most of the existing algorithms are related to the active learning of \emph{symbolic automata}~\cite{DD17,AD18}, enhancing the languages with clock valuations.
In contrast, we take a more semantic approach via the Nerode-style congruence.
% extend the L* algorithm through the Nerode-style congruence.

Another recent direction is to use a \emph{genetic algorithm} to infer TAs in passive~\cite{TALL19} or active~\cite{APT20} learning.
This differs from our learning algorithm based on a formal characterization of timed languages.
Moreover, these algorithms may not converge to the correct automaton due to a genetic algorithm.
% have no formal guarantee of termination and correctness.

% \mw{Maybe we want to say that our counterexample handing is similar to Rivest-Schapire~\cite{RS93,IS14}}

%%%%%%%%%%%%%%%%%%%%%%%%%%%%%%%%%%%%%%%%%%%%%%%%%%%%%%%%%%%%
%%%%%%%%%%%%%%%%%%%%%%%%%%%%%%%%%%%%%%%%%%%%%%%%%%%%%%%%%%%%
\section{Preliminaries}\label{sec:preliminaries}
%%%%%%%%%%%%%%%%%%%%%%%%%%%%%%%%%%%%%%%%%%%%%%%%%%%%%%%%%%%%
%%%%%%%%%%%%%%%%%%%%%%%%%%%%%%%%%%%%%%%%%%%%%%%%%%%%%%%%%%%%

For a set $X$, its powerset is denoted by $\powerset{X}$.
% For a set $X$, its powerset and finite powerset are denoted by $\powerset{X}$ and $\finpowerset{X}$, respectively.
% For a partial function $f\colon X \partfun Y$, we denote the domain by $\domain(f)$.
We denote the empty sequence by $\emptyword$.
For sets $X, Y$, we denote their symmetric difference by $X \triangle Y = \{x \mid x \in X \land x \notin Y\} \cup \{y \mid y \in Y \land y \notin X\}$.
%%%% Restriction is moved to the appendix!!!
%For a function $f\colon X \to Y$ and $X' \subseteq X$, the restriction of $f$ to $X'$ is denoted by $f|_{X'}$.

%%%%%%%%%%%%%%%%%%%%%%%%%%%%%%%%%%%%%%%%%%%%%%%%%%%%%%%%%%%%
\subsection{Timed words and timed automata}\label{subsec:timed_words_timed_automata}
%%%%%%%%%%%%%%%%%%%%%%%%%%%%%%%%%%%%%%%%%%%%%%%%%%%%%%%%%%%%

%-----------------------------------------------------------
\begin{definition}
 [timed word]%
 \label{def:timed_words}
 For a finite alphabet $\Alphabet$, a \emph{timed word} $\word$ is an alternating sequence
 $\wordInside$ of $\Alphabet$ and $\Rnn$.
 The set of timed words over $\Alphabet$ is denoted by $\TimedWords$.
 A \emph{timed language} $\Lg \subseteq \TimedWords$ is a set of timed words.
\end{definition}
%-----------------------------------------------------------

%-----------------------------------------------------------
\begin{comment} 2023-01-31: untimed and duration are not used because the notion of bounded languages is omitted
%%%%%% Untimed/Duration
 For a timed word $\wordWithInside$, the \emph{untimed part} $\untimed{\word}$ is $\untimed{\word} = \actionSequence$ and the \emph{duration} $\duration{\word}$ is $\duration{\word} = \sum_{i = 0}^{n} \timestamp_i$.
The untimed part and the duration are naturally extended for timed languages\LongVersion{, \ie{} $\untimed{\Lg} = \{\untimed{\word} \mid \word \in \Lg\}$ and $\duration{\Lg} = \{\duration{\word} \mid \word \in \Lg\}$ for $\Lg \subseteq \TimedWords$}.
\end{comment}
%-----------------------------------------------------------
%%%%%% Concatenation
For timed words $\wordWithInside$ and $\wordWithInside[']$, 
their concatenation $\word \cdot \word[']$ is
$\word \cdot \word['] = \timestamp_0 \action_1 \timestamp_1 \action_2 \dots \action_{n} (\timestamp_{n} + \timestamp'_0) \action'_1 \timestamp'_1 \action'_2 \dots \action'_{n'} \timestamp'_{n'}$.
The concatenation is naturally extended to timed languages:
for a timed word $\word$ and timed languages $\Lg, \Lg'$,
we let $\word \cdot \Lg = \{ \word \cdot \word_{\Lg} \mid \word_{\Lg} \in \Lg\}$,
$\Lg \cdot \word = \{ \word_{\Lg} \cdot \word \mid \word_{\Lg} \in \Lg\}$, and
$\Lg \cdot \Lg' = \{ \word_{\Lg} \cdot \word_{\Lg'} \mid \word_{\Lg} \in \Lg, \word_{\Lg'} \in \Lg'\}$.
%%%%%% Prefix
For timed words $\word$ and $\word'$, $\word$ is a \emph{prefix} of $\word'$ if there is a timed word $\word''$ satisfying $\word \cdot \word'' = \word'$.
%%%%%% Prefix-closed
A timed language $\Lg$ is \emph{prefix-closed} if for any $\word \in \Lg$, $\Lg$ contains all the prefixes of $\word$.
%-----------------------------------------------------------

%-----------------------------------------------------------
%%%%%% Clock valuation
For a finite set $\Clock$ of clock variables, a \emph{clock valuation} is a function $\cval \in \CVal$.
%%%%%% Restriction
\begin{comment}
 For a clock valuation $\cval \in \CVal$ over $\Clock$ and $\Clock'\subseteq\Clock$,
 we let $\project{\cval}{\Clock'} \in\clockvaluations[\Clock']$ be the clock valuation over $\Clock'$
satisfying $\project{\cval}{\Clock'}(\clock) = \cval(\clock)$ for any $\clock\in\Clock'$.
\end{comment}
%%%%%% Zero clock valuation
% For a finite set $\Clock$ of clock variables, 
We let $\zerovalue[\Clock]$ be the clock valuation\LongVersion{ $\zerovalue[\Clock] \in \CVal$}
 satisfying  $\zerovalue[\Clock](\clock) = 0$ for any $\clock \in \Clock$.
% When $\Clock$ is clear, we just denote $\zerovalue[C]$ by $\zerovalue$.
% Let
% $\intervals = \{[t,t') \subseteq \timedomain \mid t,t' \in \timedomain,
% t \leq t'\}$.
%%%%%%% dwell time
For\LongVersion{ a clock valuation} $\cval \in \CVal$\LongVersion{ over $\Clock$} and $\timestamp\in\timedomain$, we let
 $\cval + \timestamp$ be the clock valuation satisfying $(\cval+\timestamp)(\clock)=\cval(\clock)+\timestamp$ for any $\clock \in \Clock$.
%%%%%%% Reset
For\LongVersion{ a clock valuation} $\cval\in\CVal$ and $\resets \subseteq \Clock$, we let
 $\reset{\cval}{\resets}$ be the clock valuation satisfying
 $(\reset{\cval}{\resets})(x)=0$ for $\clock \in \resets$ and
 $(\reset{\cval}{\resets})(\clock)=\cval(\clock)$ for $\clock \notin \resets$.
%
%%%%%%%% Guards
% Let $\Clock$ be the finite set of variables.
We let $\GuardWithClock$ be the set of constraints defined by a finite conjunction of inequalities $\clock \bowtie \dConstant$, where
$\clock \in \Clock$,
$\dConstant \in \N$, and
${\bowtie} \in \{>, \geq, \leq,<\}$.
%%%%%%%% Constraints
 % For a finite set $\Clock$ of variables,
 We let $\ConstraintWithClock$ be the set of constraints defined by a finite conjunction of inequalities $\clock \bowtie \dConstant$ or $\clock - \clock' \bowtie \dConstant$, where
 $\clock, \clock' \in \Clock$,
 $\dConstant \in \N$, and
 ${\bowtie} \in \{>, \geq, \leq,<\}$.
% When $\Clock$ is clear, we just denote $\ConstraintWithClock$ by $\Constraint$.
 We denote $\bigwedge \emptyset$ by $\top$.
 For\LongVersion{ a clock valuation} $\cval \in \CVal$ and $\constraint \in \ConstraintWithClock \cup \GuardWithClock$, we denote $\cval \models \constraint$ if $\cval$ satisfies $\constraint$.
\begin{comment} 2023-01-31: We only define simple timed conditions
 %%%%%%%% Simple constraints
 A constraint $\constraint \in \ConstraintWithClock$ is \emph{simple} if for each $\clock \in \Clock$, $\constraint$ contains $\dConstant < \clock \land \clock < \dConstant + 1$ or $\dConstant \leq \clock \land \clock \leq \dConstant < \clock$ for some $\dConstant \in \N$, and for each $\clock, \clock' \in \Clock$, $\constraint$ contains $\dConstant < \clock - \clock' \land \clock - \clock' < \dConstant + 1$ or $\dConstant \leq \clock - \clock' \land \clock - \clock' \leq \dConstant$ for some $\dConstant \in \N$.
\end{comment}
%----------------------------------------------------------

%----------------------------------------------------------
\begin{definition}
 [timed automaton]\label{def:ta}
 A \emph{timed automaton} (TA)\LongVersion{\footnote{Unlike~\cite{MP04}, we define the acceptance by the accepting \emph{locations} rather than accepting \emph{conditions}. This choice does not affect the expressive power because accepting conditions can be encoded with unobservable edges and invariants.}} is a 7-tuple $\TAInside$, where:
 $\Alphabet$ is the finite alphabet,
 $\Loc$ is the finite set of locations,
 $\initLoc \in \Loc$ is the initial location,
 $\Clock$ is the finite set of clock variables,
 $\Invariant\colon \Loc \to \Constraint$ is the invariant of each location,
 $\Edge \subseteq \Loc \times \Guard \times (\Alphabet \cup \{\emptyword\}) \times\Resets\times\Loc$ is the set of edges, and
 $\Final \subseteq \Loc$ is the accepting locations.
% $\Final\colon\Loc\to\finpowerset{\Constraint}$ is the accepting condition of each location.
\end{definition}
%----------------------------------------------------------

\LongVersion{An edge is \emph{unobservable} if it is labeled with $\emptyword$, and otherwise, it is \emph{observable}.}
\begin{LongVersionBlock}
 %----------------------------------------------------------
 \begin{definition}
 [DTA]\label{def:deterministic}
 A TA $\TAWithInside$ is \emph{deterministic} if we have the following.
 \begin{itemize}
  \item For any $\action \in \Alphabet$ and $(\loc, \guard, \action, \resets, \loc'), (\loc, \guard', \action, \resets', \loc'') \in \Edge$, $\guard \land \guard'$ is unsatisfiable.
  \item For any $(\loc, \guard, \emptyword, \resets, \loc') \in \Edge$, $\guard \land \InvariantAt{\loc}$ is at most a singleton.
 \end{itemize}
 \end{definition}
%----------------------------------------------------------
\end{LongVersionBlock}
\begin{ShortVersionBlock}
 A TA is \emph{deterministic} if
 1) for any $\action \in \Alphabet$ and $(\loc, \guard, \action, \resets, \loc'), (\loc, \guard', \action, \resets', \loc'') \in \Edge$, $\guard \land \guard'$ is unsatisfiable, or
 2) for any $(\loc, \guard, \emptyword, \resets, \loc') \in \Edge$, $\guard \land \InvariantAt{\loc}$ is at most a singleton.
\end{ShortVersionBlock}
\cref{figure:timed_automaton} shows a deterministic TA (DTA).\@

The semantics of a TA is defined by a \emph{timed transition system (TTS)}.

%----------------------------------------------------------
\begin{definition}
 [semantics of TAs]%
 \label{def:semantics}
 For a TA $\TAWithInside$,
 the \emph{timed transition system (TTS)} is a 4-tuple
 $\TTS = (\TTSState, \initTTSState, \AccTTSState, \TTSTransition)$, where:
 % is as follows:
 \LongVersion{\begin{itemize}}
 \LongVersion{\item} $\TTSState = \Loc\times\clockvaluations$ is the set of \emph{(concrete) states},
 \LongVersion{\item} $\initTTSState = (\initLoc, \zerovalue)$ is the \emph{initial state},
 \LongVersion{\item} $\AccTTSState = \{ (\loc, \cval) \in \TTSState \mid \loc \in \Final\}$ is the set of \emph{accepting states}, and
  % \item $\AccTTSState = \{ (\loc, \cval) \in \TTSState \mid \cval \models \Final(\loc)\}$ is the set of \emph{accepting states};
 \LongVersion{\item} $\TTSTransition \subseteq \TTSState\times\TTSState$ is the \emph{transition relation} consisting of the following\footnote{We use $\TTStransitionWithLabel{\emptyword, \timestamp}$ to avoid the discussion with an arbitrary small dwell time in~\cite{MP04}.}.
        \begin{itemize}
         \item For each $(\loc, \cval) \in \TTSState$ and $\timestamp \in \Rp$, we have $(\loc, \cval)\TTStransitionRelWithLabel{\timestamp} (\loc, \cval + \timestamp)$ if $\cval + \timestamp' \models \InvariantAt{\loc}$ holds for each $\timestamp' \in [0, \timestamp)$. % chktex 9
         \item For each $(\loc, \cval), (\loc', \cval') \in \TTSState$, $\action \in \Alphabet$, and $\edgeInside \in \Edge$, we have $(\loc, \cval) \TTStransitionRelWithLabel{\action} (\loc', \cval')$ if we have $\cval\models\guard$ and $\cval'= \reset{\cval}{\resets}$.
         \item For each $(\loc, \cval), (\loc', \cval') \in \TTSState$, $\timestamp \in \Rp$, and $(\loc, \guard, \emptyword, \resets, \loc') \in \Edge$,
               we have $(\loc, \cval) \TTStransitionRelWithLabel{\emptyword, \timestamp} (\loc', \cval' + \timestamp)$
               if we have $\cval\models\guard$, $\cval'= \reset{\cval}{\resets}$, and $\forall \timestamp' \in [0, \timestamp).\, \cval' + \timestamp' \models \InvariantAt{\loc'}$.
% $\guard$ and $\InvariantAt{\loc}$ are adjacent, \ie{} $\guard$ and $\InvariantAt{\loc}$ are disjoint and their union is an interval, $\cval + \dConstant' \models \guard$ holds for each $\dConstant' \in (0, \dConstant)$, and $\cval'=\reset{(\cval + \dConstant)}{\resets}$.
               % Following~\cite{MP04}, we do not consider that having transitions from $\loc_0$ labeled with $\emptyword, \clock_0 = 0$ and $\styleact{a}, \clock_0 = 0$ as a violation of the determinism since the transition labeled with $\emptyword$ is used at an arbitrary small $t>0$.
        \end{itemize}
\LongVersion{\end{itemize}}
 % For $(\loc, \cval), (\loc', \cval') \in \TTSState$, $\dConstant \in \Rp$, and $\edge = \edgeInside \in \Edge$,
 % we write $\TTSstate \longuefleche{(\edge, \dConstant)} \TTSstate'$ if we have
 % $((\loc, \cval), (\loc, \cval + \dConstant)) \in {\fleche{\dConstant}}$ and
 % $((\loc, \cval + \dConstant), (\loc', \cval')) \in {\fleche{\edge}}$.
\end{definition}
%----------------------------------------------------------

%-----------------------------------------------------------
%%%%%% Run, accepting run, and languages
A \emph{run} of a TA $\A$ is an alternating sequence $\runInside$ of\LongVersion{ states} $\TTSstate_i \in \TTSState$ and\LongVersion{ transitions} $\TTStransition_i \in \TTSTransition$ satisfying $\TTSstate_{i-1} \TTStransitionRel_i \TTSstate_{i}$ for any $i \in \{1,2, \dots,n\}$.
A run $\runInside$ is accepting if $\TTSstate_n \in \AccTTSState$.
Given such a run, the associated timed word is the concatenation of the labels of the transitions.
The timed \emph{language} $\Lg(\TA)$ of a TA $\TA$ is the set of timed words associated with some accepting run of~$\TA$.
%-----------------------------------------------------------

%%%%% Regions
\newcommand{\DefinitionRegions}{%
\emph{Region abstraction} is an important theoretical gadget in the theory of timed automaton:
it reduces the \emph{infinite} state space $\TTSState$ of a TTS $\TTS$
to its \emph{finite} abstraction, the latter being amenable to algorithmic treatments, \eg{} reachability analysis with a breadth-first search.
Specifically, it relies on an equivalence relation $\regionEquiv_{\maxConstant}$ over clock valuations.
Given a TA $\TAWithInside$, for each $\clock \in \Clock$, let $\maxConstant_\clock$ denote the greatest number compared with $\clock$ in $\A$.
Writing $\integer(t)$ and $\fractional(t)$ for the integer and fractional parts of $t \in \Rnn$,
let $\regionEquiv_{\maxConstant}$ be the equivalence relation over clock valuations $\cval, \cval'$ defined as follows.
We have $\cval \regionEquiv_{\maxConstant} \cval'$ if:
\begin{itemize}
 \item for each $\clock\in \Clock$ we have  $\integer(\cval(\clock)) = \integer(\cval'(\clock))$ or ($\cval(\clock) > \maxConstant_{\clock}$ and $\cval'(\clock) > \maxConstant_{\clock}$);
 \item for any $\clock, \clock' \in \Clock$ satisfying $\cval(\clock) \leq \maxConstant_{\clock}$ and $\cval(\clock') \leq \maxConstant_{\clock'}$, $\fractional(\cval(\clock)) < \fractional(\cval(\clock'))$ if and only if $\fractional(\cval'(\clock)) < \fractional(\cval'(\clock'))$; and
 \item for any $\clock \in \Clock$ such that $\cval(\clock) \leq \maxConstant_{\clock}$, $\fractional(\cval(\clock)) = 0$ if and only if $\fractional(\cval'(\clock)) = 0$.
\end{itemize}

The equivalence relation $\regionEquiv_{\maxConstant}$ is extended over the states $\TTSState$ of a TTS $\TTS$ such that $(\loc, \cval) \regionEquiv_{\maxConstant} (\loc', \cval')$ if we have $\loc = \loc'$ and $\cval \regionEquiv_{\maxConstant} \cval'$.
By merging the equivalent states with respect to $\regionEquiv_{\maxConstant}$,
we abstract a TTS and obtain the \emph{region automaton}~\cite{AD94}.

%----------------------------------------------------------
\begin{definition}
 [region automata (RAs)]%
 \label{def:region_automata}
 For a TA $\TA$ and its TTS $\TTSWithInside$,
 the \emph{region automaton} (RA) is the NFA
  $\regionAutomWithInside$ defined as follows:
  $\RegionState = \TTSState / {\regionEquiv_{\maxConstant}}$;
  $\initRegionState \ni \initTTSState$;
  $\AccRegionState  = \{ \regionState \in \RegionState \mid \regionState \subseteq \AccTTSState\}$;
  $\RegionTransition = \{(\regionState, \action, \regionState[']) \mid \exists \TTSstate \in \regionState, \TTSstate' \in \regionState['].\, \TTSstate \TTStransitionRelWithLabel{\action} \TTSstate' \} \cup \{(\regionState, \emptyword, \regionState[']) \mid \exists \TTSstate \in \regionState, \TTSstate' \in \regionState['], \timestamp \in \Rp.\, \TTSstate \TTStransitionRelWithLabel{\timestamp} \TTSstate' \text{ or } \TTSstate \TTStransitionRelWithLabel{\emptyword, \timestamp} \TTSstate'\}$.
\end{definition}
}
\begin{comment}
\DefinitionRegions{}
%----------------------------------------------------------
\end{comment}

%%%%%%%%%%%%%%%%%%%%%%%%%%%%%%%%%%%%%%%%%%%%%%%%%%%%%%%%%%%%
\subsection{Recognizable timed languages}
%%%%%%%%%%%%%%%%%%%%%%%%%%%%%%%%%%%%%%%%%%%%%%%%%%%%%%%%%%%%

Here, we review the \emph{recognizability}~\cite{MP04} of timed languages.
\begin{LongVersionBlock}

%% The notion of ``bounded timed languages'' is not used anymore!
% %-----------------------------------------------------------
%  \begin{definition}
% [bounded timed language]
\end{LongVersionBlock}
% A timed language $\Lg$ is \emph{bounded} if the untimed part $\untimed{\Lg} \subseteq \Alphabet^*$ is finite and the duration $\duration{\Lg} \subseteq \Rnn$ is bounded.
% \begin{LongVersionBlock}
 % \end{definition}
%-----------------------------------------------------------

%-----------------------------------------------------------
\begin{definition}
 [timed condition]
 For a set $\timeVariables = \{\timestampSequence\}$ of ordered variables,
 a \emph{timed condition} $\timedCondition$ is a finite conjunction of inequalities $\sumTimestamp{i}{j} \bowtie \dConstant$, where
 $\sumTimestamp{i}{j} = \sum_{{k = i}}^{j} \timestamp_k$,
 ${\bowtie} \in \{>, \geq, \leq,<\}$, and
 $\dConstant \in \N$.
\end{definition}
%-----------------------------------------------------------

%-----------------------------------------------------------
%%%%%% Simplicity of timed conditions
A timed condition $\timedCondition$ is \emph{simple}\footnote{The notion of simplicity is taken from~\cite{GJL10}} if
for each $\sumTimestamp{i}{j}$,
$\timedCondition$ contains $\dConstant < \sumTimestamp{i}{j} < \dConstant + 1$ or $\dConstant \leq \sumTimestamp{i}{j} \land \sumTimestamp{i}{j} \leq \dConstant$
for some $\dConstant \in \N$.
%%%%%% Canonicity of timed conditions
A timed condition $\timedCondition$ is \emph{canonical} if
we cannot strengthen or add any inequality $\sumTimestamp{i}{j} \bowtie \dConstant$ to $\timedCondition$ without changing its semantics.
%-----------------------------------------------------------

%-----------------------------------------------------------
\begin{definition}
 [elementary language]
 A timed language $\Lg$ is \emph{elementary} if there are $u = \actionSequence \in \Alphabet^*$ and
 a timed condition $\timedCondition$ over $\{\timestampSequence\}$ satisfying
 $\Lg = \{\wordInside \mid \timestampSequence \models \timedCondition\}$, and
 the set of valuations of $\{\timestampSequence\}$ defined by $\timedCondition$ is bounded.
 We denote such\LongVersion{ an elementary language} $\Lg$ by $\elementary$.
 We let $\Elementary$ be the set of elementary languages over $\Alphabet$.
\end{definition}
%-----------------------------------------------------------

%%%%% Empty elementary language
\begin{comment}
 We denote the elementary language $(\styleact{\emptyword}, \timedCondition)$ by $\emptyword$,
 where $\timedCondition$ is the singleton of $\timestamp_0 = 0$.
\end{comment}
%%%%% (Strict) prefix
 For\LongVersion{ elementary languages} $\prefix, \prefix' \in \Elementary$,
 $\prefix$ is a \emph{prefix} of $\prefix'$ if
 for any $\word' \in \prefix'$, 
 there is a prefix $\word \in \prefix$ of $\word'$, and
 for any $\word \in \prefix$, 
 there is $\word' \in \prefix'$ such that $\word$ is a prefix of $\word'$.
%%% 2023-01-31: strict prefix is not used anymore
% A prefix $\prefix$ of $\prefix'$ is a \emph{strict prefix} if \LongVersion{$\prefix$ is not equal to $\prefix'$}\ShortVersion{$\prefix \neq \prefix'$}.
%%%%% Suffix
 % The notion of (strict) suffixes is defined similarly.\todo{Check if the notion of suffixes is currently used}
 % We denote $\prefix \isSuffixOf \prefix'$ if $\prefix$ is a strict suffix of $\prefix'$.
For any elementary language, the number of its prefixes is finite.
%%%%%% Prefix-/suffix-closedness
For a set of elementary languages,\LongVersion{ the notions of} \emph{prefix-closedness} is defined based on the above definition of prefixes.
% For a set of elementary languages, the notions of prefix-closedness and suffix-closedness are defined based on the prefix and suffix notion above.

%%%%% Simple elementary languages
 An elementary language $\elementary$ is \emph{simple} if there is a simple and canonical timed condition $\timedCondition'$ satisfying $\elementary = (u, \timedCondition')$.
 We let $\SimpleElementary$ be the set of simple elementary languages over $\Alphabet$.
 Without loss of generality, we assume that for any $\elementary \in \SimpleElementary$, $\timedCondition$ is simple and canonical.
 We remark that any DTA cannot distinguish timed words in a simple elementary language,
 \ie{} for any\LongVersion{ simple elementary language} $\prefix \in \SimpleElementary$ and a DTA $\TA$,
 we have either $\prefix \subseteq \Lg(\TA)$ or $\prefix \cap \Lg(\TA) = \emptyset$.
 We can decide if $\prefix \subseteq \Lg(\TA)$ or $\prefix \cap \Lg(\TA) = \emptyset$ by taking some $\word \in \prefix$ and checking if $\word \in \Lg(\TA)$.
 % For an elementary language $\prefix = \elementary$, the length $|u|$ of the elements is denoted by $|\prefix|$.

%-----------------------------------------------------------
\begin{definition}
 [immediate exterior]
 Let $\Lg = \elementary$ be an elementary language.
 For $\action \in \Alphabet$,
 the \emph{discrete immediate exterior} $\exterior[\action]{\Lg}$ of $\Lg$ is  $\exterior[\action]{\Lg} = (u \cdot \action, \timedCondition \cup \{\timestamp_{ |u| + 1 } = 0\})$.
 The \emph{continuous immediate exterior} $\exterior[t]{\Lg}$ of $\Lg$ is  $\exterior[t]{\Lg} = (u, \timedCondition^t)$, where
 $\timedCondition^t$ is the timed condition such that each inequality $\sumTimestamp{i}{|u|} = \dConstant$ in $\timedCondition$ is replaced with  $\sumTimestamp{i}{|u|} > \dConstant$ if such an inequality exists, and otherwise,
 the inequality $\sumTimestamp{i}{|u|} < \dConstant$ in $\timedCondition$ with the smallest index $i$ is replaced with $\sumTimestamp{i}{|u|} = \dConstant$.
%  for the smallest index $i$ such that $\timedCondition$ contains $\sumTimestamp{i}{|u|} < c$,
%  the inequality is replaced with  $\sumTimestamp{i}{|u|} = c$.
 The \emph{immediate exterior}\LongVersion{ $\exterior{\Lg}$} of $\Lg$ is
 $\exterior{\Lg} = \bigcup_{\action \in \Alphabet}\exterior[\action]{\Lg} \cup \exterior[t]{\Lg}$.
\end{definition}
%-----------------------------------------------------------

%----------------------------------------------------------
\begin{example}
 For a word $u = \styleact{a} \cdot \styleact{a}$ and a timed condition $\timedCondition = \{\sumTimestamp{0}{0} \in (1,2) \land \sumTimestamp{0}{1} \in (1,2) \land \sumTimestamp{0}{2} \in (1,2) \land \sumTimestamp{1}{2} \in (0,1) \land \sumTimestamp{2}{2} = 0\}$,
 we have
 $1.3 \cdot \styleact{a} \cdot 0.5 \cdot \styleact{a} \cdot 0 \in \elementary$ and
 $1.7 \cdot \styleact{a} \cdot 0.5 \cdot \styleact{a} \cdot 0 \notin \elementary$.
 The discrete and continuous immediate exteriors of $\elementary$ are
 $\exterior[\styleact{a}]{\elementary} = (u \cdot \styleact{a}, \timedCondition^{\styleact{a}})$ and
 $\exterior[t]{\elementary} = (u, \timedCondition^{t})$, where
 $\timedCondition^{\styleact{a}} = \{\sumTimestamp{0}{0} \in (1,2) \land \sumTimestamp{0}{1} \in (1,2) \land \sumTimestamp{0}{2} \in (1,2) \land \sumTimestamp{1}{2} \in (0,1) \land \sumTimestamp{2}{2} = \sumTimestamp{3}{3} = 0\}$ and
 $\timedCondition^{t} = \{\sumTimestamp{0}{0} \in (1,2) \land \sumTimestamp{0}{1} \in (1,2) \land \sumTimestamp{0}{2} \in (1,2) \land \sumTimestamp{1}{2} \in (0,1) \land \sumTimestamp{2}{2} > 0\}$.
\end{example}
%----------------------------------------------------------

%-----------------------------------------------------------
\begin{definition}
 [chronometric timed language]
 % \ShortVersion{A timed language $\Lg$ is \emph{chronometric} if it is representable by a finite union of disjoint elementary languages.}
 A timed language $\Lg$ is \emph{chronometric} if there is a finite set $\{\elementary[_{1}], \elementary[_{2}], \dots, \elementary[_{m}]\}$ of disjoint elementary languages satisfying
 $\Lg = \bigcup_{i \in \{1,2, \dots,m\}} \elementary[_{i}]$.
\end{definition}
%-----------------------------------------------------------

For any elementary language $\Lg$, its immediate exterior $\exterior{\Lg}$ is chronometric.
We naturally extend the notion of exterior to chronometric timed languages, \ie{}
for a chronometric timed language $\Lg = \bigcup_{i \in \{1,2, \dots,m\}} \elementary[_{i}]$,
we let $\exterior{\Lg} = \bigcup_{i \in \{1,2, \dots,m\}} \exterior{\elementary[_{i}]}$, which is also chronometric.
%-----------------------------------------------------------
%%%%%% Time valuation
For a timed word $\wordWithInside$, we denote the valuation of $\timestampSequence$ by $\timeValuation{\word}\LongVersion{ \in \clockvaluations[\timeVariables]}$.
%-----------------------------------------------------------

\emph{Chronometric relational morphism}~\cite{MP04} relates any timed word to a timed word in a certain set $\PrefixSet$, which is later used to define a timed language.
Intuitively, the tuples in $\CRM$ specify a mapping from timed words immediately out of $\PrefixSet$ to timed words in $\PrefixSet$.
By inductively applying it, any timed word is mapped to $\PrefixSet$.

%-----------------------------------------------------------
\begin{definition}
 [chronometric relational morphism]%
 \label{def:chronometric_relational_morphism}
 Let $\PrefixSet$ be a chronometric and prefix-closed timed language.
 Let $\CRMTuple$ be a 5-tuple such that
 $\elementary \subseteq  \exterior{\PrefixSet}$,
 $\elementary['] \subseteq \PrefixSet $, and
 $\Rename$ is a finite conjunction of equations of the form
 $\sumTimestamp{i}{|u|} = \sumTimestamp[']{j}{|u'|}$, where $i \leq |u|$ and $j \leq |u'|$.
 For such a tuple, we let $\sem{\CRMTuple} \subseteq \elementary \times \elementary[']$ be the relation such that
 % for any $\word \in \elementary$ and $\word' \in \elementary[']$,
 $(\word, \word') \in \sem{\CRMTuple}$ if and only if
 $\timeValuation{\word}, \timeValuation{\word'} \models \Rename$.
 For a finite set $\CRM$ of such tuples\LongVersion{ $\CRMTuple$},
 the \emph{chronometric relational morphism} $\sem{\CRM} \subseteq \TimedWords \times \PrefixSet$ is the relation inductively defined as follows:
 1) for $\word \in \PrefixSet$, we have $(\word, \word) \in \sem{\CRM}$;
 2) for $\word \in \exterior{\PrefixSet}$ and $\word' \in \PrefixSet$, we have $(\word, \word') \in \sem{\CRM}$ if we have $(\word, \word') \in \sem{\CRMTuple}$ 
 for one of the tuples $\CRMTuple \in \CRM$;
 3) for $\word \in \exterior{\PrefixSet}$, $\word' \in \TimedWords$, and $\word'' \in \PrefixSet$,
 we have $(\word \cdot \word', \word'') \in \sem{\CRM}$ if there is $\word''' \in \TimedWords$ satisfying
 $(\word, \word''') \in \sem{\CRM}$ and $(\word''' \cdot \word', \word'') \in \sem{\CRM}$.
% , where
%  $\elementary$ and $\elementary[']$ are elementary languages satisfying $\elementary \subseteq  \exterior{\PrefixSet}$ and $\elementary['] \subseteq \PrefixSet $, and
%  $\Rename$ is a finite set of equations of the form
%  $\sum_{k = i}^{|u|} \timestamp_{k} = \sum_{k' = j}^{|u'|} \timestamp'_{k'}$, where $i \leq |u|$ and $j \leq |u'|$.
 We also require that all $\elementary$ in the tuples in $\CRM$ must be disjoint and the union of each such $\elementary$ is $\exterior{\PrefixSet} \setminus \PrefixSet$.
 % For a chronometric relational morphism $\CRM$, $\sem{\CRM} \subseteq \TimedWords \times \PrefixSet$ is inductively defined as follows:.
 % For $\word \in \exterior{\PrefixSet}$ and $\word' \in \PrefixSet$, we have $(\word, \word') \in \sem{\CRM}$ if for some $\CRMTuple \in \CRM$, we have $(\word, \word') \in \sem{\CRMTuple}$.
 % For $\word \in \exterior{\PrefixSet}$, $\word' \in \TimedWords$, and $\word'' \in \PrefixSet$,
 % we have $(\word \cdot \word', \word'') \in \sem{\CRM}$ if we have $(\word''' \cdot \word', \word'') \in \sem{\CRM}$ for some $(\word, \word''') \in \sem{\CRM}$.
\end{definition}
%-----------------------------------------------------------

A chronometric relational morphism $\sem{\CRM}$ is \emph{compatible} with $\RecognizableFinal \subseteq \PrefixSet$ if for each tuple $\CRMTuple$ defining $\sem{\CRM}$, we have either $\elementary['] \subseteq \RecognizableFinal$ or $\elementary['] \cap \RecognizableFinal = \emptyset$.

%-----------------------------------------------------------
\begin{definition}
 [recognizable timed language]
 A timed language $\Lg$ is \emph{recognizable} if there is a chronometric prefix-closed set $\PrefixSet$,
 a chronometric subset $\RecognizableFinal$ of $\PrefixSet$, and
 a chronometric relational morphism $\sem{\CRM} \subseteq \TimedWords \times \PrefixSet$ compatible with $\RecognizableFinal$
 satisfying $\Lg = \{\word \mid \exists \word' \in \RecognizableFinal, (\word, \word') \in \sem{\CRM}\}$.
\end{definition}
%-----------------------------------------------------------

It is known that for any recognizable timed language $\Lg$, we can construct a DTA $\A$ recognizing $\Lg$, and vice versa~\cite{MP04}.

\begin{comment}
 The following characterization of the recognizable timed languages by DTAs is the main theorem of~\cite{MP04}.

 %-----------------------------------------------------------
 \begin{theorem}
 For any timed language $\Lg$,  $\Lg$ is a recognizable timed language if and only if there is a DTA $\A$ recognizing $\Lg$.
 \qed{}
 \end{theorem}
 %-----------------------------------------------------------
\end{comment}

%%%%%%%%%%%%%%%%%%%%%%%%%%%%%%%%%%%%%%%%%%%%%%%%%%%%%%%%%%%%
\subsection{Distinguishing extensions and active DFA learning}\label{subsection:L*}
%%%%%%%%%%%%%%%%%%%%%%%%%%%%%%%%%%%%%%%%%%%%%%%%%%%%%%%%%%%%

%% Nerode's congruence and its approximation with suffixes
\LongVersion{Among the various characterization of regular languages, most}\ShortVersion{Most} DFA learning algorithms are based on\LongVersion{ the characterization with} \emph{Nerode's congruence}~\cite{HMU07}.
For a (not necessarily regular) language $\Lg \subseteq \Alphabet^*$, Nerode's congruence ${\nerode{\Lg}} \subseteq \Alphabet^* \times \Alphabet^*$ is the equivalence relation satisfying $\word \nerode{\Lg} \word'$ if and only if for any $\word'' \in \Alphabet^*$,
we have $\word \cdot \word'' \in \Lg \iff \word' \cdot \word'' \in \Lg$.

Generally, we cannot decide if $\word \nerode{\Lg} \word'$ by testing because it requires infinitely many membership checking.
However, if $\Lg$ is regular, \LongVersion{thanks to the Myhill-Nerode theorem,
we can decide whether $\word \nerode{\Lg} \word'$ only by finitely many membership checking.
More precisely, for any regular language $\Lg$,} there is a finite set of suffixes $\SuffixSet \subseteq \Alphabet^*$ called \emph{distinguishing extensions} satisfying ${\nerode{\Lg}} = {\CellSim^{\SuffixSet}_{\Lg}}$, where ${\CellSim^{\SuffixSet}_{\Lg}}$ is the equivalence relation satisfying $\word \CellSim^{\SuffixSet}_{\Lg} \word'$ if and only if for any $\word'' \in \SuffixSet$,
we have $\word \cdot \word'' \in \Lg \iff \word' \cdot \word'' \in \Lg$.
% we have $\characteristic{\Lg}(\word \cdot \word'') = \characteristic{\Lg}(\word' \cdot \word'')$.
Thus, the minimum DFA recognizing $\targetLg$ can be learned by\footnote{The distinguishing extensions $\SuffixSet$ can be defined locally. For example, the TTT algorithm~\cite{IHS14} is optimized with \emph{local} distinguishing extensions for some prefixes $\word \in \Alphabet^*$. Nevertheless, we use the global distinguishing extensions for simplicity.}:
 i) identifying distinguishing extensions $\SuffixSet$ satisfying ${\nerode{\targetLg}} = {\CellSim^{\SuffixSet}_{\targetLg}}$ and % chktex 10
 ii) constructing the minimum DFA $\A$ corresponding to ${\CellSim^{\SuffixSet}_{\targetLg}}$. % chktex 10

The \emph{L* algorithm}~\cite{Angluin87} is an algorithm to learn the minimum DFA $\hypothesisA$ recognizing the target regular language $\targetLg$ with finitely many \emph{membership} and \emph{equivalence} queries to the teacher.
In a membership query, the learner asks if $\word \in \Alphabet^*$ belongs to the target language $\targetLg$ \ie{} $\word \in \targetLg$.
In an equivalence query, the learner asks if the hypothesis DFA $\hypothesisA$ recognizes the target language $\targetLg$ \ie{} $\Lg(\hypothesisA) = \targetLg$,
where $\Lg(\hypothesisA)$ is the language of the hypothesis DFA $\hypothesisA$.
When we have $\Lg(\hypothesisA) \neq \targetLg$,
the teacher returns a counterexample $\cex \in \Lg(\hypothesisA) \setdiff \targetLg$.
%
%% Observation table + closedness + consistency
% In the \emph{L* algorithm}~\cite{Angluin87},
The information obtained via queries is stored in a 2-dimensional array called an \emph{observation table}.
See \cref{figure:observation_table_illustration} for an illustration.
% the distinguishing extensions $\SuffixSet$ and the minimum DFA recognizing\LongVersion{ the target language} $\targetLg$ are learned with
For finite index sets $\PrefixSet, \SuffixSet \subseteq \Alphabet^*$,
for each pair $(\prefix, \suffix) \in (\PrefixSet \cup \PrefixSet \cdot \Alphabet) \times \SuffixSet$,
the observation table stores whether\ShortVersion{ $\prefix \cdot \suffix \in \targetLg$}\LongVersion{ the concatenation $\prefix \cdot \suffix$ is a member of the target language $\targetLg$}.
\begin{LongVersionBlock}
 Formally, an observation table is a 3-tuple $\ObsTableInside$ of
 a prefix-closed finite set $\PrefixSet$ of words,
 a suffix-closed finite set $\SuffixSet$ of words, and a function
 $\Table\colon (\PrefixSet \cup \PrefixSet \cdot \Alphabet) \times \SuffixSet \to \{\top, \bot\}$
assigning a pair $(\prefix, \suffix)$ of words to the membership $\prefix \cdot \suffix \in \targetLg$ of their concatenation to the target language.
\end{LongVersionBlock}
$\SuffixSet$ is the current candidate of the distinguishing extensions, and $\PrefixSet$ represents
$\Alphabet^* / {\CellSim^{\SuffixSet}_{\targetLg}}$.
Since $\PrefixSet$ and $\SuffixSet$ are finite, one can fill the observation table using finite membership queries.
% \cref{figure:observation_table} illustrates observation tables.

%% DFA construction is not important for us
% The rough idea of the DFA construction in the L* algorithm from an observation table is as follows:
%  i) the state space of the resulting DFA is $\PrefixSet$; % chktex 10
%  ii) for each $\prefix \in \PrefixSet$ and $\action \in \Alphabet$, we make a transition from $\prefix$ to $\prefix \cdot \action$ labeled with $\action$ if $\prefix \cdot \action \in \PrefixSet$; % chktex 10
%  iii) otherwise, the target of the transition from $\prefix$ labeled with $\action$ is not $\prefix \cdot \action$ but $\prefix' \in \PrefixSet$ satisfying $\TableCell{\prefix'}{\suffix} = \TableCell{\prefix \cdot \action}{\suffix}$ for any $\suffix \in \SuffixSet$. % chktex 10
% \LongVersion{For the minimality of the constructed DFA and the termination of the algorithm, we merge the states equivalent up to the equivalence relation $\CellSim^{\SuffixSet}_{\targetLg}$.}

%-----------------------------------------------------------
\begin{algorithm}[tbp]
 \caption{Outline of an L*-style active DFA learning algorithm}%
 \label{algorithm:Lstar}
 \DontPrintSemicolon{}
 % \newcommand{\myCommentFont}[1]{\texttt{\footnotesize{#1}}}
 % \SetCommentSty{myCommentFont}
 \SetKwFunction{FConstructTable}{ConstructTable}
 \SetKwFunction{FConstructA}{ConstructDFA}
 \footnotesize
 $\PrefixSet \gets \{\styleact{\emptyword}\};\; \SuffixSet \gets \{\styleact{\emptyword}\}$\;
 \While{$\top$} {\label{algorithm:Lstar:main_loop:begin}
   \While{the observation table is not closed or consistent} {\label{algorithm:Lstar:close_consistent_loop:begin}
     update $\PrefixSet$ and $\SuffixSet$ so that the observation table is closed and consistent\;
%     fill $\Table$ using membership queries
   }\label{algorithm:Lstar:close_consistent_loop:end}
   $\hypothesisA \gets \FConstructA{\PrefixSet, \SuffixSet, \Table}$\;
   \Switch{$\eqQ[\targetLg]{\hypothesisA}$} {
     \Case{$\top$} {
       \KwReturn{$\hypothesisA$}\label{algorithm:Lstar:eqQ:end}
     }
     \Case{$\cex$} {\label{algorithm:Lstar:eqQ:cex}
       Update $\PrefixSet$ and/or $\SuffixSet$ using $\cex$\;\label{algorithm:Lstar:add_cex}
     }
   }
 }\label{algorithm:Lstar:main_loop:end}
\end{algorithm}
%-----------------------------------------------------------
\cref{algorithm:Lstar} outlines an L*-style algorithm.
We start from $\PrefixSet = \SuffixSet = \{\styleact{\emptyword}\}$ and incrementally increase them.
To construct a hypothesis DFA $\hypothesisA$, the observation table must be \emph{closed} and \emph{consistent}.
\LongVersion{Intuitively, the closedness assures that the target state of each transition is in $\PrefixSet$, and the consistency assures that all the states with the same row behave in the same way.
Formally, an}
\ShortVersion{An} observation table is \emph{closed} if, for each $\prefix \in \PrefixSet \cdot \Alphabet$, there is $\prefix' \in \PrefixSet$ satisfying $\prefix \CellSim^{\SuffixSet}_{\targetLg} \prefix'$.
An observation table is \emph{consistent} if,
for any $\prefix, \prefix' \in \PrefixSet \cup \PrefixSet \cdot \Alphabet$ and $\action \in \Alphabet$,
$\prefix \CellSim^{\SuffixSet}_{\targetLg} \prefix'$ implies $\prefix \cdot \action \CellSim^{\SuffixSet}_{\targetLg} \prefix' \cdot \action$.
% See \cref{figure:observation_table} for examples of the DFA construction.

Once the observation table becomes closed and consistent, the learner constructs a hypothesis DFA $\hypothesisA$ and checks if $\Lg(\hypothesisA) = \targetLg$ by an equivalence query.
If $\Lg(\hypothesisA) = \targetLg$ holds, $\hypothesisA$ is the resulting DFA.\@
Otherwise, the teacher returns\LongVersion{ a counterexample} $\cex \in \Lg(\hypothesisA) \setdiff \targetLg$, which is used to refine the observation table.
There are several variants of the refinement.
In the L* algorithm, all the prefixes of $\cex$ are added to $\PrefixSet$.
In the Rivest-Schapire algorithm~\cite{RS93,IS14}, an extension $\suffix$ strictly refining $\SuffixSet$ is obtained by an analysis of $\cex$, and such $\suffix$ is added to $\SuffixSet$.

 % with its prefixes and \LongVersion{go back to \cref{algorithm:L*:close_consistent_loop:begin}}\ShortVersion{make the observation table closed and consistent again}.

% The L* algorithm gradually refines the equivalence relation $\CellSim^{\SuffixSet}_{\targetLg}$ until it converges to Nerode's congruence $\nerode{\targetLg}$.
% Since a new suffix is added to $\SuffixSet$ only when $\CellSim^{\SuffixSet}_{\targetLg}$ is strictly refined, such a refinement occurs at most $|{\Alphabet^*}/{\nerode{\targetLg}}|$ times.
% This is one of the reasons of the termination of the L* algorithm.

%%%%%%%%%%%%%%%%%%%%%%%%%%%%%%%%%%%%%%%%%%%%%%%%%%%%%%%%%%%%
%%%%%%%%%%%%%%%%%%%%%%%%%%%%%%%%%%%%%%%%%%%%%%%%%%%%%%%%%%%%
\section{A Myhill-Nerode style characterization of recognizable timed languages with elementary languages}%
\label{section:timed_distinguishing_suffixes}
%%%%%%%%%%%%%%%%%%%%%%%%%%%%%%%%%%%%%%%%%%%%%%%%%%%%%%%%%%%%
%%%%%%%%%%%%%%%%%%%%%%%%%%%%%%%%%%%%%%%%%%%%%%%%%%%%%%%%%%%%

%%%%%%%%%%%%%%%%%%%%%%%%%%%%%%%%%%%%%%%%%%%%%%%%%%%%%%%%%%%%
% \subsection{Symbolic membership}%
% \label{subsection:symbolic_membership}
%%%%%%%%%%%%%%%%%%%%%%%%%%%%%%%%%%%%%%%%%%%%%%%%%%%%%%%%%%%%

Unlike the case of regular languages,
any finite set of timed words cannot correctly distinguish recognizable timed languages due to the infiniteness of dwell time in timed words.
Instead, we use a finite set of \emph{elementary languages} to define a Nerode-style congruence.
To define the Nerode-style congruence,
we extend the notion of membership to elementary languages.

\begin{definition}
 [symbolic membership]%
 \label{def:symbolic_membership}
 For{ a timed language} $\Lg \subseteq \TimedWords$ and{ an elementary language} $\elementary \in \Elementary$,
 the \emph{symbolic membership} $\symbolicMemQ{\Lg}(\elementary)$ of $\elementary$ to $\Lg$ is the strongest constraint such that
 for any $\word \in \elementary$, we have $\word \in \Lg$ if and only if $\timeValuation{\word} \models \symbolicMemQ{\Lg}(\Lg)$.
\end{definition}

We discuss how to obtain symbolic membership in \cref{section:symbolic_membership_oracle}.
We define a Nerode-style congruence using symbolic membership.
A naive idea is to distinguish two elementary languages\LongVersion{ $\elementary$ and $\elementary[']$} by the equivalence of their symbolic membership\LongVersion{ $\symbolicMemQ{\Lg}(\elementary)$ and $\symbolicMemQ{\Lg}(\elementary['])$}.
However, this does not capture the semantics of TAs.
For example, for the DTA $\A$ in \cref{figure:timed_automaton}, for any timed word $\word$,
we have $1.3 \cdot \styleact{a} \cdot 0.4 \cdot \word \in \Lg(\A) \iff 0.3 \cdot \styleact{a} \cdot 1.0 \cdot \styleact{a} \cdot 0.4 \cdot \word \in \Lg(\A)$, while they have different symbolic membership.
This is because symbolic membership distinguishes the \emph{position} in timed words where each clock variable is reset, which must be ignored.
We use \emph{renaming equations} to abstract such positional information\LongVersion{ of the clock reset} in symbolic membership.
Note that $\sumTimestamp{i}{n} = \sum_{{k = i}}^{n} \timestamp_k$ corresponds to the value of the clock variable reset at $\timestamp_i$.

\begin{definition}
 [renaming equation]%
 \label{definition:renaming}
 \mw{2022/03/02: Let's use equation since it simplify the formulation}
 Let $\timeVariables = \{\timestampSequence\}$ and $\timeVariables' = \{\timestampSequence[']\}$ be sets of ordered variables.
 A \emph{renaming equation} $\Rename$ over $\timeVariables$ and $\timeVariables'$ is a finite conjunction of equations of the form 
 $\sumTimestamp{i}{n} = \sumTimestamp[']{i'}{n'}$,
 where $i \in \{0, 1, \dots, n\}$,
 $i' \in \{0, 1, \dots, n'\}$,
 $\sumTimestamp{i}{n} = \sum_{{k = i}}^{n} \timestamp_k$ and
 $\sumTimestamp[']{i'}{n'} = \sum_{{k = i'}}^{n'} \timestamp'_k$.
\end{definition}
%-----------------------------------------------------------

%-----------------------------------------------------------
\begin{definition}
 [$\CellSim^{\SuffixSet}_{\Lg}$]%
 \label{definition:equivalence_row}\mw{Perhaps, we can define with languages and show how to check it}
 Let $\Lg \subseteq \TimedWords$ be a timed language, let $\elementary, \elementary['], \elementary[''] \in \Elementary$ be elementary languages, and
 let $\Rename$ be a renaming equation over $\timeVariables$ and $\timeVariables'$, where $\timeVariables$ and $\timeVariables'$ are the variables of $\timedCondition$ and $\timedCondition'$, respectively.
 We let $\elementary \CellSub^{\elementary[''], \Rename}_{\Lg} \elementary[']$ if we have the following:
   for any $\word \in \elementary$, there is $\word' \in \elementary[']$ satisfying $\timeValuation{\word}, \timeValuation{\word'}\models \Rename$;
   $\rename{\symbolicMemQ{\Lg}(\elementary \cdot \elementary[''])}{\Rename} \land \timedCondition'$ is equivalent to $\symbolicMemQ{\Lg}(\elementary['] \cdot \elementary['']) \land \Rename \land \timedCondition$.
 We let $\elementary \CellSim^{\elementary[''], \Rename}_{\Lg} \elementary[']$ if we have
   $\elementary \CellSub^{\elementary[''], \Rename}_{\Lg} \elementary[']$ and
   $\elementary['] \CellSub^{\elementary[''], \Rename}_{\Lg} \elementary$.
% the following:
 % for any $\word' \in \elementary[']$, there is $\word \in \elementary$ satisfying $\timeValuation{\word}, \timeValuation{\word'}\models \Rename$,
 % % $\timedCondition \Rightarrow \rename{\timedCondition[']}{\Rename^{-1}}$,
 % % $\timedCondition['] \Rightarrow \rename{\timedCondition}{\Rename}$,
 % $\rename{\symbolicMemQ{\Lg}(\elementary['] \cdot \elementary[''])}{\Rename^{-1}} \land \timedCondition \Rightarrow \symbolicMemQ{\Lg}(\elementary \cdot \elementary[''])$.
 Let $\SuffixSet \subseteq \Elementary$.
 We let $\elementary \CellSim^{\SuffixSet, \Rename}_{\Lg} \elementary[']$
 if for any $\elementary[''] \in \SuffixSet$,
 we have $\elementary \CellSim^{\elementary[''], \Rename}_{\Lg} \elementary[']$.
 We let $\elementary \CellSim^{\SuffixSet}_{\Lg} \elementary[']$
 if $\elementary \CellSim^{\SuffixSet, \Rename}_{\Lg} \elementary[']$ for some renaming equation $\Rename$.
\end{definition}
%-----------------------------------------------------------

%-----------------------------------------------------------
\begin{example}
 Let $\A$ be the DTA in \cref{figure:timed_automaton} and let $\elementary$, $\elementary[']$, and $\elementary['']$ be elementary languages, where
 $u = \styleact{a}$, $\timedCondition = \{\timestamp_{0} \in (1,2) \land \timestamp_0 + \timestamp_{1} \in (1,2) \land \timestamp_{1} \in (0,1)\}$,
 $u' = \styleact{a}\cdot \styleact{a}$, $\timedCondition' = \{\timestamp'_{0} \in (0,1) \land \timestamp'_{0} + \timestamp'_{1} \in (1,2) \land \timestamp'_{1} + \timestamp'_{2} \in (1,2) \land \timestamp'_{2} \in (0,1)\}$,
 $u'' = \styleact{a}$, and $\timedCondition'' = \{\timestamp_{0} \in (0,1) \land \timestamp_{1} = 0\}$.
 We have
   $\symbolicMemQ{\Lg(\A)}(\elementary \cdot \elementary['']) = \timedCondition \land \timedCondition'' \land \timestamp_1 + \timestamp''_0 \leq 1$ and
   $\symbolicMemQ{\Lg(\A)}(\elementary['] \cdot \elementary['']) = \timedCondition' \land \timedCondition'' \land \timestamp'_2 + \timestamp''_0 \leq 1$.
 Therefore, for the renaming equation $\sumTimestamp{1}{1} = \sumTimestamp[']{2}{2}$,
 we have $\elementary \CellSim^{\elementary[''], \sumTimestamp{1}{1} = \sumTimestamp[']{2}{2}}_{\Lg} \elementary[']$.
\end{example}
%-----------------------------------------------------------
An algorithm to check if $\elementary \CellSim^{\SuffixSet}_{\Lg} \elementary[']$ is shown in \cref{appendix:row_equivalence}.

Intuitively, $\elementary \CellSub^{\suffix, \Rename}_{\Lg} \elementary[']$ shows that any $\word \in \elementary$ can be ``simulated'' by some $\word' \in \elementary[']$ with respect to $\suffix$ and $\Rename$. %, although $\CellSub^{\suffix, \Rename}_{\Lg}$ is \emph{not} a simulation relation.
\mw{Probably $\elementary \CellSub^{\suffix, \Rename}_{\Lg} \elementary[']$ actually forms a simulation relation with some generalization in elementary languages. It might be interesting to investigate something on that direction}
Such intuition is formalized as follows.
%-----------------------------------------------------------
\newcommand{\adequacyEquivalenceStatement}{%
 For any $\Lg \subseteq \TimedWords$ and $\elementary, \elementary['], \elementary[''] \in \Elementary$ satisfying
 $\elementary \CellSub^{\elementary['']}_{\Lg} \elementary[']$,
  for any $\word \in \elementary$, there is $\word' \in \elementary[']$ such that for any $\word'' \in \elementary['']$, $\word \cdot \word'' \in \Lg \iff \word' \cdot \word'' \in \Lg$ holds.
}
\begin{theorem}%
 \label{theorem:adequacy_equivalence}
 \adequacyEquivalenceStatement{}
 \qed{}
\end{theorem}
%-----------------------------------------------------------
\begin{LongVersionBlock}
 \begin{proof}
 [sketch]
 % We prove the first condition of \cref{theorem:adequacy_equivalence}.
 % The proof of the second condition is similar thanks to the symmetricity.
 Let $\Rename$ be a renaming equation satisfying $\elementary \CellSub^{\SuffixSet, \Rename}_{\Lg} \elementary[']$.
 By the definition of $\elementary \CellSub^{\SuffixSet, \Rename}_{\Lg} \elementary[']$,
 for any $\elementary[''] \in \SuffixSet$,
 $\rename{\symbolicMemQ{\Lg}(\elementary \cdot \elementary[''])}{\Rename} \land \timedCondition'$ holds if and only if $\symbolicMemQ{\Lg}(\elementary['] \cdot \elementary['']) \land \Rename \land \timedCondition$.
 % $\rename{\symbolicMemQ{\Lg}(\elementary \cdot \elementary[''])}{\Rename} \land \timedCondition'$ implies $\symbolicMemQ{\Lg}(\elementary['] \cdot \elementary[''])$. %and
 % $\rename{\symbolicMemQ{\Lg}(\elementary['] \cdot \elementary[''])}{\Rename} \land \timedCondition$ implies $\symbolicMemQ{\Lg}(\elementary \cdot \elementary[''])$.
 Therefore, any suffix $\elementary[''] \in \SuffixSet$ cannot distinguish $\word \in \elementary$ and $\word' \in \elementary[']$
 if % they are equivalent up to the renaming, \ie{}
 we have $\timeValuation{\word}, \timeValuation{\word'} \models \Rename$.
 Since $\elementary \CellSub^{\SuffixSet, \Rename}_{\Lg} \elementary[']$,
 For any $\word \in \elementary$, there is $\word' \in \elementary[']$ satisfying $\timeValuation{\word}, \timeValuation{\word'} \models \Rename$.
 Overall, for any $\word \in \elementary$, there is $\word' \in \elementary[']$ such that for any $\elementary[''] \in \SuffixSet$ and for any $\word'' \in \elementary['']$,
 we have $\word \cdot \word'' \in \Lg \iff \word' \cdot \word'' \in \Lg$.
 \qed{}
\end{proof}
\end{LongVersionBlock}
%-----------------------------------------------------------

\mw{Perhaps this corollary is not very important}
By $\bigcup_{\elementary \in \Elementary} \elementary = \TimedWords$, we have the following as a corollary.

%-----------------------------------------------------------
\begin{corollary}%
 \label{corollary:adequacy_equivalence_limit}
 For any timed language $\Lg \subseteq \TimedWords$ and for any elementary languages $\elementary, \elementary['] \in \Elementary$,
 $\elementary \CellSim^{\Elementary}_{\Lg} \elementary[']$ implies the following.
 \begin{itemize}
  \item For any $\word \in \elementary$, there is $\word' \in \elementary[']$ such that for any $\word'' \in \TimedWords$, we have $\word \cdot \word'' \in \Lg \iff \word' \cdot \word'' \in \Lg$.
  \item For any $\word' \in \elementary[']$, there is $\word \in \elementary$ such that for any $\word'' \in \TimedWords$, we have $\word \cdot \word'' \in \Lg \iff \word' \cdot \word'' \in \Lg$.
        \qed{}
 \end{itemize}
\end{corollary}
%-----------------------------------------------------------

The following\LongVersion{ theorem} characterizes recognizable timed languages with $\CellSim^{\Elementary}_{\Lg}$.
%-----------------------------------------------------------
\begin{theorem}
 [Myhill-Nerode style characterization]%
 \label{theorem:finiteness}
 A timed language $\Lg$ is recognizable if and only if the quotient set $\SimpleElementary / {\CellSim^{\Elementary}_{\Lg}}$ is finite.
 \qed{}
\end{theorem}
%-----------------------------------------------------------

By \cref{theorem:finiteness}, we always have a finite set $\SuffixSet$ of distinguishing extensions.
%the following theorem on finite distinguishing extensions also holds.

%-----------------------------------------------------------
\newcommand{\FiniteSuffix}{%
 For any recognizable timed language $\Lg$, there is a finite set $\SuffixSet$ of elementary languages satisfying ${\CellSim^{\Elementary}_{\Lg}} = {\CellSim^{\SuffixSet}_{\Lg}}$.
}
\begin{theorem}%
 \label{theorem:finite_suffix}
 \FiniteSuffix{}
 \qed{}
\end{theorem}
%-----------------------------------------------------------

%-----------------------------------------------------------
% \begin{example}
%  Let $\A$ be the DTA in \cref{figure:timed_automaton}.
%  The following $\PrefixSet$ represents $\SimpleElementary / {\CellSim^{\Elementary}_{\Lg}}$:
%  $\PrefixSet = \{(\varepsilon, \timestamp_0 = 0), (\varepsilon, \timestamp_0 \in (0,1)), (\varepsilon, \timestamp_0 = 1), (a, \timestamp_0 = 1 \land \timestamp_1 = 0), (a, \timestamp_0 = 1 \land \timestamp_1 \in (0,1)), (a, \timestamp_0 = 1 \land \timestamp_1 = 1), (a, \timestamp_0 = 1 \land \timestamp_1 \in (1,2))\}$.
%  We have $\PrefixSet' = \PrefixSet \cup \{(a, \timestamp_0 = 1 \land \timestamp_1 = 2)\}$.
%  This construction is because
%  $\exterior[t]{(a, \timestamp_0 = 1 \land \timestamp_1 \in (1,2))} = (a, \timestamp_0 = 1 \land 1 < \timestamp_1 = 2)$ is empty
%  while $\exterior[t]{(\varepsilon, \timestamp_0 = 1)} = (\varepsilon, \timestamp_0 > 1)$ is nonempty.
% \end{example}
%-----------------------------------------------------------

%%%%%%%%%%%%%%%%%%%%%%%%%%%%%%%%%%%%%%%%%%%%%%%%%%%%%%%%%%%%
%%%%%%%%%%%%%%%%%%%%%%%%%%%%%%%%%%%%%%%%%%%%%%%%%%%%%%%%%%%%
\section{Active learning of deterministic timed automata}\label{section:active_learning}
%%%%%%%%%%%%%%%%%%%%%%%%%%%%%%%%%%%%%%%%%%%%%%%%%%%%%%%%%%%%
%%%%%%%%%%%%%%%%%%%%%%%%%%%%%%%%%%%%%%%%%%%%%%%%%%%%%%%%%%%%

We show our L*-style active learning algorithm for DTAs with the Nerode-style congruence in \cref{section:timed_distinguishing_suffixes}.
% This is our main contribution.
We let $\targetLg\LongVersion{ \subseteq \TimedWords}$ be the target timed language in learning.

For simplicity, we first present our learning algorithm with a smart teacher answering the following three kinds of queries:
membership query $\memQ[\targetLg]{\word}$ asking whether $\word \in \targetLg$,
symbolic membership query asking $\symbolicMemQ{\targetLg}({\elementary})$, and
equivalence query $\eqQ[\targetLg]{\hypothesisA}$ asking whether $\Lg(\hypothesisA) = \targetLg$.
If $\Lg(\hypothesisA) = \targetLg$, $\eqQ[\targetLg]{\hypothesisA} = \top$, and otherwise,
$\eqQ[\targetLg]{\hypothesisA}$ is a timed word $\cex \in \Lg(\hypothesisA) \setdiff \targetLg$.
Later in \cref{section:symbolic_membership_oracle}, we show how to answer a symbolic membership query with finitely many membership queries.
Our task is to construct a DTA $\A$ satisfying $\Lg(\A) = \targetLg$ with finitely many queries.
% membership $\memQKey[\targetLg]$ and equivalence $\eqQKey[\targetLg]$ queries.

% For a timed word $\word$, $\memQ[\targetLg]{\word} = \top$ if $\word \in \targetLg$, and otherwise, $\memQ[\targetLg]{\word} = \bot$.
% For an elementary language $\elementary$, $\symbolicMemQ{\targetLg}{\elementary}$ if $\word \in \targetLg$, and otherwise, $\memQ[\targetLg]{\word} = \bot$.

% As we mentioned in \cref{section:timed_distinguishing_suffixes},
% for any elementary language $\elementary$,
% we can compute the symbolic membership $\symbolicMemQ{\targetLg}{\elementary}$ with finitely many membership queries.
% Given a DTA $\hypothesisA$, $\eqQ[\targetLg]{\hypothesisA} = \top$ if $\Lg(\hypothesisA) = \targetLg$, and otherwise,
%  $\eqQ[\targetLg]{\hypothesisA}$ is a timed word $\word \in \Lg(\hypothesisA) \setdiff \targetLg$.
% Formally, our problem setting is as follows.

%-----------------------------------------------------------
% \smallskip
% \defProblem{Active DTA learning}%
%            {An oracle answering membership $\memQKey[\targetLg]$ and equivalence $\eqQKey[\targetLg]$ queries}%
%            {Construct a DTA $\A$ satisfying $\Lg(\A) = \targetLg$ with finite queries}
%-----------------------------------------------------------

%%%%%%%%%%%%%%%%%%%%%%%%%%%%%%%%%%%%%%%%%%%%%%%%%%%%%%%%%%%%
%%%%%%%%%%%%%%%%%%%%%%%%%%%%%%%%%%%%%%%%%%%%%%%%%%%%%%%%%%%%
\subsection{Successors of simple elementary languages}%
\label{section:fractional_elementary}
%%%%%%%%%%%%%%%%%%%%%%%%%%%%%%%%%%%%%%%%%%%%%%%%%%%%%%%%%%%%
%%%%%%%%%%%%%%%%%%%%%%%%%%%%%%%%%%%%%%%%%%%%%%%%%%%%%%%%%%%%
Similarly to the L* algorithm in \cref{subsection:L*}, we learn a DTA with an observation table.
Reflecting the extension of the underlying congruence,
we use sets of simple elementary languages for the indices.
% Since the congruence $\CellSim^{\SuffixSet}_{\Lg}$ in \cref{definition:equivalence_row} is defined over elementary languages,
% it is natural to use (simple) elementary languages as $\PrefixSet$.
To define the extensions, $\PrefixSet \cup (\PrefixSet \cdot \Alphabet)$ in the L* algorithm,
we introduce \emph{continuous} and \emph{discrete} \emph{successors} for simple elementary languages, 
which are inspired by \emph{regions}~\cite{AD94}.
We note that immediate exteriors do not work for this purpose.
\begin{LongVersionBlock}

 %-----------------------------------------------------------
 \begin{example}
 Let $\elementary = (\styleact{\emptyword}, 0 < \timestamp_0 \land \timestamp_0 < 1)$.
 Since we have $\exterior[t]{\elementary} = (\styleact{\emptyword}, 0 < \timestamp_0 \land \timestamp_0 = 1)$,
 $\exterior[t]{\elementary}$ is empty.
 Let $\elementary['] = (\styleact{a}, \sumTimestamp[']{0}{1} < 2 \land \sumTimestamp[']{1}{1} < 1)$ and $\word = 0.7 \cdot \styleact{a} \cdot 0.9$.
 We have $\word \in \elementary[']$.
 Since we have $\exterior[t]{\elementary[']} = (\styleact{a}, \sumTimestamp[']{0}{1} = 2 \land \sumTimestamp[']{1}{1} < 1)$,
 there is no $t > 0$ satisfying $\word \cdot t \in \exterior[t]{\elementary[']}$.
 \end{example}
 %-----------------------------------------------------------
\end{LongVersionBlock}
\begin{ShortVersionBlock}
%
% However, the continuous immediate exteriors are not suitable for incremental construction of $\PrefixSet$.
 For example,
 % for $\elementary = (\styleact{\emptyword}, 0 < \timestamp_0 \land \timestamp_0 < 1)$,
 % since we have $\exterior[t]{\elementary} = (\styleact{\emptyword}, 0 < \timestamp_0 \land \timestamp_0 = 1)$,
 % $\exterior[t]{\elementary}$ is empty.
 % %
 % Moreover,
 for $\elementary = (\styleact{a}, \{\sumTimestamp{0}{1} < 2 \land \sumTimestamp{1}{1} < 1\})$ and $\word = 0.7 \cdot \styleact{a} \cdot 0.9$,
 we have $\word \in \elementary$ and $\exterior[t]{\elementary} = (\styleact{a}, \{\sumTimestamp{0}{1} = 2 \land \sumTimestamp{1}{1} < 1\})$, but
 there is no $t > 0$ satisfying $\word \cdot t \in \exterior[t]{\elementary}$.
\end{ShortVersionBlock}

% For the incremental construction of $\PrefixSet$ and $\PrefixSet'$, we introduce the discrete and continuous \emph{successors} of simple elementary languages.
% % Instead of using elementary languages,
% % we introduce \emph{regional elementary languages}, which are simple elementary languages with constraints on the fractional parts to construct the successor.
% This idea is much like the successors of \emph{regions}~\cite{AD94}.

%-----------------------------------------------------------
% \begin{definition}
%  [regional elementary language]
%  A \emph{regional elementary language} is a timed language represented by a 3-tuple
%  $\fractionalElementary$ such that
%  $\fractionalElementary = \{\wordInside \mid \timestampSequence \models \timedCondition \land \fractionalCondition \}$, where
%  $u = \actionSequence \in \Alphabet^*$,
%  $\timedCondition$ is a simple and canonical timed condition over $\timeVariables = \{\timestampSequence\}$, and
%  $\fractionalCondition$ is a conjunctive inequality constraint defining a total order over $0$ and the fractional parts $\fractionalSequence$ of $\sumTimestampSequence$.
% \end{definition}
%-----------------------------------------------------------

For any\LongVersion{ simple elementary language} $\elementary \in \SimpleElementary$, we let $\fractionalCondition_{\elementary}$ be the total order over $0$ and the fractional parts $\fractionalSequence$ of $\sumTimestampSequence$. Such an order is uniquely defined because $\timedCondition$ is simple and canonical (\cref{proposition:order_uniqueness}).

% \todo{Check if the following is necessary}
% For a set $\PrefixSet$ of regional elementary languages, we let $\projectElementary(\PrefixSet) = \{\elementary \mid \fractionalElementary \in \PrefixSet\}$.
% For a regional elementary language $\fractionalElementary$ and a set $\PrefixSet$ of regional elementary languages,
% the immediate exteriors $\exterior{\fractionalElementary}$ and $\exterior{\PrefixSet}$ are
% $\exterior{\fractionalElementary} = \exterior{\elementary}$ and $\exterior{\PrefixSet} = \exterior{\projectElementary(\PrefixSet)}$.

%-----------------------------------------------------------
\begin{definition}
 [successor]
 Let $\prefix = \elementary \in \SimpleElementary$ be a simple elementary language.
 The \emph{discrete successor} $\successor[\action]{\prefix}$ of $\prefix$ is
 $\successor[\action]{\prefix} = (u \cdot \action, \timedCondition \land \timestamp_{n+1} = 0)$.
 The \emph{continuous successor} $\successor[t]{\prefix}$ of $\prefix$ is $\successor[t]{\prefix} = (u, \timedCondition^t)$, where
 $\timedCondition^{t}$
 is defined as follows:
  if there is an equation $\sumTimestamp{i}{n} = \dConstant$ in $\timedCondition$,
        all such equations are replaced with
        $\sumTimestamp{i}{n} \in (\dConstant, \dConstant + 1)$;
  otherwise,
        for each greatest $\sumTimestamp{i}{n}$ in terms of $\fractionalCondition_{\elementary}$, %over the fractional part,
        we replace $\sumTimestamp{i}{n} \in (\dConstant, \dConstant + 1)$ with $\sumTimestamp{i}{n} = \dConstant+1$.
 We let $\successor{\prefix} = \bigcup_{\action \in \Alphabet}\successor[\action]{\prefix} \cup \successor[t]{\prefix}$.
 For\LongVersion{ a set} $\PrefixSet \subseteq \SimpleElementary$\LongVersion{ of simple elementary languages},
 we let $\successor{\PrefixSet} = \bigcup_{\prefix \in \PrefixSet} \successor{\prefix}$.
\end{definition}
%-----------------------------------------------------------

\begin{LongVersionBlock}
 One can easily check that both discrete and continuous successors of a simple elementary language are also simple.
\end{LongVersionBlock}

\begin{example}
 Let $u = \styleact{a} \styleact{a}$,
 $\timedCondition = \{\sumTimestamp{0}{0} \in (1,2) \land \sumTimestamp{0}{1} \in (1,2) \land \sumTimestamp{0}{2} \in (1,2) \land \sumTimestamp{1}{1} \in (0,1) \land \sumTimestamp{1}{2} \in (0,1) \land \sumTimestamp{2}{2} = 0\}$.
 The order $\fractionalCondition_{\elementary}$ is such that $0 = \fractional(\sumTimestamp{2}{2}) < \fractional(\sumTimestamp{1}{2}) < \fractional(\sumTimestamp{0}{2})$.
 The continuous successor of $\elementary$ is
 $\successor[t]{\elementary} = (u, \timedCondition^{t})$, where
 $\timedCondition^{t} = \{\sumTimestamp{0}{0} \in (1,2) \land \sumTimestamp{0}{1} \in (1,2) \land \sumTimestamp{0}{2} \in (1,2) \land \sumTimestamp{1}{1} \in (0,1) \land \sumTimestamp{1}{2} \in (0,1) \land \sumTimestamp{2}{2} \in (0,1)\}$.
 The continuous successor of $(u, \timedCondition^{t})$ is
 $\successor[t]{(u, \timedCondition^{t})} = (u, \timedCondition^{\mathit{tt}})$, where
 $\timedCondition^{\mathit{tt}} = \{\sumTimestamp{0}{0} \in (1,2) \land \sumTimestamp{0}{1} \in (1,2) \land \sumTimestamp{0}{2} = 2 \land \sumTimestamp{1}{1} \in (0,1) \land \sumTimestamp{1}{2} \in (0,1) \land \sumTimestamp{2}{2} \in (0,1)\}$.
\end{example}
%----------------------------------------------------------

%%%%%%%%%%%%%%%%%%%%%%%%%%%%%%%%%%%%%%%%%%%%%%%%%%%%%%%%%%%%
\subsection{Timed observation table for active DTA learning}\label{subsection:observation_table}
%%%%%%%%%%%%%%%%%%%%%%%%%%%%%%%%%%%%%%%%%%%%%%%%%%%%%%%%%%%%

We extend the observation table with (simple) elementary languages and symbolic membership
to learn a \emph{recognizable timed language}\LongVersion{ rather than a \emph{regular language}}.
\begin{definition}
 [timed observation table]%
 \label{def:timed_observation_table}
 A \emph{timed observation table}\LongVersion{ for DTA learning} is a 3-tuple $\TimedObsTableInside$ such that:
   $\PrefixSet$ is a prefix-closed finite set of simple elementary languages,
   $\SuffixSet$ is a finite set of elementary languages, and
  % \item $\PrefixSet'$ is the minimum finite set of forward regional elementary languages satisfying
  %       \begin{ienumeration}
  %        \item $\PrefixSet \cap \PrefixSet' = \emptyset$,
  %        \item $\bigcup_{\prefix \in \PrefixSet} \successor{\prefix} \subseteq \PrefixSet \cup \PrefixSet'$,
  %        \item $\PrefixSet \cup \PrefixSet'$ is prefix-closed, and
  %        \item for any $\prefix \in \PrefixSet$, $\suffix \in \SuffixSet$, and $\word \in \prefix \cdot \suffix$, there are $\prefix' \in \PrefixSet \cup \PrefixSet'$, $\suffix' \in \SuffixSet$ satisfying $\word \in \prefix' \cdot \suffix'$ and $\suffix \isSuffixOf \suffix$; and
  %       \end{ienumeration}
  $\Table$ is a function mapping
 $(\prefix, \suffix) \in (\PrefixSet \cup \successor{\PrefixSet}) \times \SuffixSet$ to
 the symbolic membership $\symbolicMemQ{\targetLg}(\prefix \cdot \suffix)$.
\end{definition}
%-----------------------------------------------------------

\cref{figure:DTA_learning:example} illustrates timed observation tables:
each cell indexed by $(\prefix, \suffix)$ show the symbolic membership $\symbolicMemQ{\targetLg}(\prefix \cdot \suffix)$.
\LongVersion{ In the\LongVersion{ original} L* algorithm\LongVersion{ in \cref{subsection:L*}}, closedness and consistency are defined by the equality of the rows\LongVersion{ of $\Table$}.}
For timed observation tables,
we extend\LongVersion{ them}\ShortVersion{ the notion of closedness and consistency} with\LongVersion{ the congruence} $\CellSim^{\SuffixSet}_{\targetLg}$ we introduced in \cref{section:timed_distinguishing_suffixes}.
% The closedness and the consistency of the observation table in our DTA learning algorithm are defined as follows.
We note that consistency is defined only for discrete successors.
This is because a timed observation table does not always become ``consistent'' for continuous successors. %some of the properties does not hold for continuous successors due to the continuity.
See \cref{subsection:consistency} for a detailed discussion.
We also require \emph{exterior-consistency} % and \emph{time-saturation}
since we construct an exterior from a successor.

%-----------------------------------------------------------
\begin{definition}
 [closedness, consistency, exterior-consistency, %, time-saturation
 cohesion]%
 \label{definition:closed_consistent_exterior_consistent}
 Let $\ObsTableWithInside$ be a timed observation table.
 $\ObsTable$ is \emph{closed} if,
 for each $\prefix \in \successor{\PrefixSet} \setminus \PrefixSet$,
 there is $\prefix' \in \PrefixSet$ satisfying $\prefix \CellSim^{\SuffixSet}_{\targetLg} \prefix'$.
% \end{definition}
% %-----------------------------------------------------------
%
% %-----------------------------------------------------------
% \begin{definition}
%  [consistency]
 $\ObsTable$ is \emph{consistent} if,
 for each $\prefix, \prefix' \in \PrefixSet$ and for each $\action \in \Alphabet$,
 $\prefix \CellSim^{\SuffixSet}_{\targetLg} \prefix'$ implies
 % $\successor[t]{\prefix} \CellSim^{\SuffixSet}_{\targetLg} \successor[t]{\prefix'}$ and
 $\successor[\action]{\prefix} \CellSim^{\SuffixSet}_{\targetLg} \successor[\action]{\prefix'}$.
 $\ObsTable$ is \emph{exterior-consistent} if
 for any $\prefix \in \PrefixSet$,
 $\successor[t]{\prefix} \notin \PrefixSet$ implies $\successor[t]{\prefix} \subseteq \exterior[t]{\prefix}$.
% we have the following conditions:
 % \begin{ienumeration}
 %  \item For each $\prefix \in \PrefixSet$ satisfying $\successor[t]{\prefix} \in \PrefixSet'$,
 %        we have $\successor[t]{\prefix} \subseteq \exterior[t]{\prefix}$;
  % \item For each $\prefix, \prefix' \in \PrefixSet$ satisfying
  % $\successor[t]{\prefix} \in \PrefixSet'$,
  % $\successor[t]{\prefix'} \in \PrefixSet'$, and
  % $\exterior[t]{\prefix} \cap \exterior[t]{\prefix'} \neq \emptyset$,
  % we have $\prefix \CellSim^{\SuffixSet}_{\targetLg} \prefix'$;
 %  \item For each $\prefix \in \PrefixSet$ and $\prefix' \in \PrefixSet \cup \PrefixSet'$ satisfying
 %  $\successor[t]{\prefix} \in \PrefixSet'$ and
 %  $\prefix' \subseteq \exterior[t]{\prefix}$,
 %  we have $\successor[t]{\prefix} \CellSim^{\SuffixSet}_{\targetLg} \prefix'$.
 % \end{ienumeration}
% \end{definition}
%-----------------------------------------------------------
%
%-----------------------------------------------------------
% \begin{definition}
%  [time-saturation]
 % $\ObsTable$ is \emph{time-saturated} if
 % for each $\prefix \in \PrefixSet$
 % $\successor[t]{\prefix} \notin \PrefixSet$ implies
 % $\prefix \CellSim^{\SuffixSet, \top}_{\targetLg} \successor[t]{\prefix}$.
 $\ObsTable$ is \emph{cohesive} if it is closed, consistent, and exterior-consistent. %, and time-saturated.
\end{definition}
\SetKwFunction{FConstructDTA}{MakeDTA}
% \paragraph{Construction of a DTA}
%------------------------------------------------------------
\begin{algorithm}[tb]
 \ShortVersion{\scriptsize}
 \Input{A cohesive timed observation table $\TimedObsTableInside$}
 \Output{A DTA $\hypothesisA$ row-faithful to the given timed observation table}
 \BlankLine{}
 \DontPrintSemicolon{}
 \Fn{\FConstructDTA{$\TimedObsTableInsideInside$}}{
     $\CRM \gets \emptyset$; $\RecognizableFinal \gets \{\elementary  \in \PrefixSet \mid \TableCell{\fractionalElementary}{(\styleact{\emptyword}, \timestamp'_0 = 0)} = \{\timedCondition \land \timestamp'_0 = 0\}\}$\;
     \For{$\prefix \in \PrefixSet$ such that $\successor[t]{\prefix} \notin \PrefixSet$ (\resp{} $\successor[\action]{\prefix} \notin \PrefixSet$)}{
         \tcp{Construct $(\elementaryInside, \elementaryInside['], \Rename)$ for some $\prefix' \in \PrefixSet$ and $\Rename$}
         \tcp{Such $\Rename$ is chosen using an exhaustive search}
         \KwPick{} $\prefix' \in \PrefixSet$ and $\Rename$ such that $\successor[t]{\prefix} \CellSim^{\SuffixSet, \Rename}_{\targetLg} \prefix'$ (\resp{} $\successor[\action]{\prefix} \CellSim^{\SuffixSet, \Rename}_{\targetLg} \prefix'$)\;
         \KwPush{} $(\elementaryInside, \elementaryInside['], \Rename)$ \KwTo{} $\CRM$, where $\elementary = \exterior[t]{\prefix}$ (\resp{} $\exterior[\action]{\prefix}$) and $\elementary['] = \prefix'$\;
     }
     \Return{} the DTA $\hypothesisA$ obtained from $(\PrefixSet, \RecognizableFinal, \CRM)$ by the construction in~\cite{MP04}
 }
 \caption{DTA construction from a timed observation table}%
 \label{algorithm:DTA_construction}
\end{algorithm}
%------------------------------------------------------------
From a cohesive timed observation table $\TimedObsTableInside$, we can construct a DTA as outlined in \cref{algorithm:DTA_construction}.
% The construction is almost the same as the construction in the proof of \cref{lemma:finite_to_timed_recognizable}.
We construct a DTA via a recognizable timed language.
The main ideas are as follows:
 % \mw{Merging is already done in~\cite{MP04} and we do not have to do here.}
 % \item we construct $\tilde{\PrefixSet}$ by merging $\prefix$ and $\successor[t]{\prefix}$ if we have $\prefix \subseteq \targetLg \iff \successor[t]{\prefix}$; % there is a renaming equation $\CellRel$ such that for any $\prefix' \in \PrefixSet \cup \successor{\PrefixSet}$,
 % $\prefix \CellSim^{\SuffixSet}_{\targetLg} \prefix' \Rightarrow \prefix \CellSim^{\SuffixSet, \CellRel}_{\targetLg} \prefix'$ and
 % $\successor[t]{\prefix} \CellSim^{\SuffixSet}_{\targetLg} \prefix' \Rightarrow \successor[t]{\prefix} \CellSim^{\SuffixSet, \CellRel}_{\targetLg} \prefix'$;
 % \item when a continuous immediate successor is mapped to the source by $\CRM$, we stay at the same location by removing the upper bound of the invariant;
 1) we approximate $\CellSim^{\Elementary, \CellRel}_{\targetLg}$ by $\CellSim^{\SuffixSet, \CellRel}_{\targetLg}$;
 2) we decide the equivalence class of\LongVersion{ an exterior} $\exterior{\prefix} \in \exterior{\PrefixSet} \setminus \PrefixSet$ in $\Elementary$ from successors.
Notice that there can be multiple renaming equations $\Rename$ showing $\CellSim^{\SuffixSet, \CellRel}_{\targetLg}$.
We use one of them found by an exhaustive search in \cref{appendix:row_equivalence}.
% The first merging is to reduce the size of the resulting DTA.
\LongVersion{The latter estimation is justified by the exterior-consistency of the timed observation table.}
\LongVersion{We denote the resulting DTA of this construction by \FConstructDTA{$\TimedObsTableInsideInside$}.}
%-----------------------------------------------------------

% \begin{example}
%  \mw{We may want to add an example of the constructed DTA}
% \end{example}

The DTA obtained by \FConstructDTA is faithful to the timed observation table in \emph{rows},
\ie{} for any $\prefix \in \PrefixSet \cup \successor{\PrefixSet}$,
$\targetLg \cap \prefix = \Lg(\hypothesisA) \cap \prefix$.
However, unlike the\LongVersion{ original} L* algorithm, this does not\LongVersion{ (at least likely)} hold for each \emph{cell},
\ie{} there may be $\prefix \in \PrefixSet \cup \successor{\PrefixSet}$ and
$\suffix \in \SuffixSet$ satisfying
$\targetLg \cap (\prefix \cdot \suffix) \neq \Lg(\hypothesisA) \cap (\prefix \cdot \suffix)$.
This is because we do not (and actually cannot) enforce consistency for continuous successors.
See \cref{subsection:consistency} for a discussion.
% $\prefix \cdot \predecessor[t]{\suffix} = \successor[t]{\prefix} \cdot \suffix$ does not hold.
\LongVersion{It is future work to find a concrete example violating the faithfulness\LongVersion{ for each cell}.}
Nevertheless, this does not affect the correctness of our\LongVersion{ DTA learning} algorithm partly by \cref{lemma:correct_in_limit}.
% In our DTA learning algorithm,
% % in \cref{subsection:dta_learning}
% we ensure the faithfulness by equivalence queries.
% Our DTA construction is justified by the faithfulness of the constructed DTA \FConstructDTA{$\TimedObsTableInsideInside$} to the timed observation table $\TimedObsTableInside$.

%-----------------------------------------------------------
\newcommand{\rowFaithfulnessStatement}{%
 For any cohesive timed observation table $\TimedObsTableInside$,
 for any $\prefix \in \PrefixSet \cup \successor{\PrefixSet}$,
 $\targetLg \cap \prefix = \Lg(\FConstructDTA(\TimedObsTableInsideInside)) \cap \prefix$ holds.
}
\begin{theorem}
 [row faithfulness]%
 \label{theorem:row_faithfulness}
 \rowFaithfulnessStatement{}
 \qed{}
\end{theorem}
%-----------------------------------------------------------

%-----------------------------------------------------------
\newcommand{\correctnessInLimitStatement}{%
  For any cohesive timed observation table $\TimedObsTableInside$,
  $\CellSim^{\SuffixSet}_{\targetLg} = \CellSim^{\Elementary}_{\targetLg}$ implies
  $\targetLg = \Lg(\FConstructDTA(\TimedObsTableInsideInside))$.
}
\begin{theorem}%
 \label{lemma:correct_in_limit}
 \correctnessInLimitStatement{}
 \ShortVersion{\qed{}}
\end{theorem}
%-----------------------------------------------------------

%-----------------------------------------------------------
% \newcommand{\faithfulnessStatement}{%
%  Let $\TimedObsTableInside$ be a closed, consistent, and exterior-consistent timed observation table.
%  For any $\prefix \in \PrefixSet \cup \PrefixSet'$ and
%  $\suffix \in \SuffixSet$,
%  % $\word \in \elementary$, and
%  % $\word' \in \elementary[']$,
%  we have
% $\targetLg \cap (\prefix \cdot \suffix) = \Lg(\hypothesisA) \cap (\prefix \cdot \suffix)$, where
%  % we have $\word \cdot \word' \in \Lg(\hypothesisA)$ if and only if
%  % $\timeValuation{\word \cdot \word'} \models \timedCondition \land \timedCondition' \land \TableCell{\prefix}{\suffix}$, where
%  $\hypothesisA = \FConstructDTA(\PrefixSet, \SuffixSet, \Table)$.
% }
% \begin{theorem}
%  [faithfulness]%
%  \label{theorem:faithfulness}
%  \faithfulnessStatement{}
%  \qed{}
% \end{theorem}
%-----------------------------------------------------------

%%%%%%%%%%%%%%%%%%%%%%%%%%%%%%%%%%%%%%%%%%%%%%%%%%%%%%%%%%%%
\subsection{Counterexample analysis}\label{appendix:counterexample_handling}
%%%%%%%%%%%%%%%%%%%%%%%%%%%%%%%%%%%%%%%%%%%%%%%%%%%%%%%%%%%%

%-----------------------------------------------------------
\begin{algorithm}[t]
 % \caption{Rivest-Schapire-style counterexample analysis in our DTA learning algorithm}%
 \caption{Counterexample analysis in our DTA learning algorithm}%
\label{algorithm:handle_counterexample}
 \ShortVersion{\scriptsize}
 \DontPrintSemicolon{}
 \newcommand{\myCommentFont}[1]{\texttt{\ShortVersion{\scriptsize}{#1}}}
 \SetCommentSty{myCommentFont}
 \DontPrintSemicolon{}
 \SetKwFunction{FHandleCounterExample}{AnalyzeCEX}
 \Fn{\FHandleCounterExample{$\cex$}}{
   $i \gets 1$;\, $\word_0 \gets \cex$\;
   \While{$\nexists \prefix \in \PrefixSet. \word_{i} \in \prefix$} {
     $i \gets i + 1$ \;
     \KwSplit{} $\word_{i-1}$ \KwInto{} $\word'_{i} \cdot \word''_i$ such that $\word'_{i} \in \prefix'_i$ for some $\prefix'_i \in \successor{\PrefixSet} \setminus \PrefixSet$\;\label{algorithm:handle_counterexample:begin_map}
     \KwLet{} $\prefix_i \in \PrefixSet$ and $\Rename_i$ \KwBe{} such that $\prefix'_i \CellSim^{\SuffixSet, \Rename_i}_{\targetLg} \prefix_i$\;
     \KwLet{} $\overline{\word}_i \in \prefix_i$ \KwBe{} such that $\timeValuation{\word'_i}, \timeValuation{\overline{\word}_i}\models \Rename_i$\;\label{algorithm:handle_counterexample:end_map}
     $\word_i \gets \overline{\word}_i \cdot \word''_i$
   }
   \KwFind{} $j \in \{1,2,\dots,i\}$ such that $\word_{j-1} \in \targetLg \setdiff \Lg(\hypothesisA)$ and $\word_{j} \not\in \targetLg \setdiff \Lg(\hypothesisA)$\;\label{algorithm:handle_counterexample:find_evidence}
   \tcp{We use a binary search with membership queries for $\lceil \log(i) \rceil$ times.}
   \KwReturn{} the simple elementary language including $\word''_{j}$
   % $I \gets \left\{j \in \{1, 2, \dots, i\} \mid \word_{j-1} \in \targetLg \setdiff \Lg(\hypothesisA) \land \word_{j} \not\in \targetLg \setdiff \Lg(\hypothesisA)\right\}$\;\label{algorithm:handle_counterexample:find_evidence}
   % \KwReturn{} $\suffix \in \SimpleElementary$ such that $\exists j \in I.\, \word''_{j} \in \suffix$ and $\suffix \not\in \SuffixSet$ if exists, otherwise $\bot$
 }
\end{algorithm}
%-----------------------------------------------------------
We analyze the counterexample $\cex$ obtained by an equivalence query to refine the set $\SuffixSet$ of suffixes in a timed observation table.
Our analysis, outlined in \cref{algorithm:handle_counterexample}, is inspired by the Rivest-Schapire algorithm~\cite{RS93,IS14}.
The idea is to reduce{ the counterexample} $\cex$ using the mapping defined by{ the congruence} $\CellSim^{\SuffixSet}_{\targetLg}$ (lines~\ref{algorithm:handle_counterexample:begin_map}--\ref{algorithm:handle_counterexample:end_map}), much like $\CRM$ in recognizable timed languages, and to find a suffix $\suffix$ strictly refining $\SuffixSet$ (line~\ref{algorithm:handle_counterexample:find_evidence}), \ie{}
satisfying
$\prefix \CellSim^{\SuffixSet}_{\targetLg} \prefix'$ and $\prefix \NCellSim^{\SuffixSet\cup\{\suffix\}}_{\targetLg} \prefix'$ for some $\prefix \in \successor{\PrefixSet}$ and $\prefix' \in \PrefixSet$.
% We construct an elementary language $\suffix$ from a counterexample $\cex$ obtained by an equivalence query through an analysis inspired by the Rivest-Schapire algorithm~\cite{RS93,IS14}.
% We use such $\suffix$ to refine the set $\SuffixSet$ of suffixes.
% \cref{algorithm:handle_counterexample} outlines our analysis to generate such $\suffix$.
% The idea is to reduce $\cex$ using a mapping defined by $\CellSim^{\SuffixSet}_{\targetLg}$, which is much like $\CRM$ in recognizable timed languages (lines~\ref{algorithm:handle_counterexample:begin_map}--\ref{algorithm:handle_counterexample:end_map}), and
% find a suffix satisfying
% $\prefix \CellSim^{\SuffixSet}_{\targetLg} \prefix'$ and $\prefix \NCellSim^{\SuffixSet\cup\{\suffix\}}_{\targetLg} \prefix'$ for some $\prefix \in \successor{\PrefixSet}$ and $\prefix' \in \PrefixSet$ (line~\ref{algorithm:handle_counterexample:find_evidence}).

By definition of $\cex$, we have $\cex = \word_0 \in \targetLg \setdiff \Lg(\hypothesisA)$.
By \cref{theorem:row_faithfulness}, $\word_n \not\in \targetLg \setdiff \Lg(\hypothesisA)$ holds, where $n$ is the final value of $i$.
By construction of $\hypothesisA$ and $\word_i$, for any $i \in \{1, 2, \dots, n\}$, we have $\word_0 \in \Lg(\hypothesisA) \iff \word_i \in \Lg(\hypothesisA)$.
Therefore, there is $i \in \{1, 2, \dots, n\}$ satisfying $\word_{i-1} \in \targetLg \setdiff \Lg(\hypothesisA)$ and $\word_{i} \not\in \targetLg \setdiff \Lg(\hypothesisA)$.
For such $i$, since we have
$\word_{i-1} = \word'_i \cdot \word''_i \in \targetLg \setdiff \Lg(\hypothesisA)$,
$\word_i = \overline{\word}_i \cdot \word''_i \not\in \targetLg \setdiff \Lg(\hypothesisA)$, and
$\timeValuation{\word'_i}, \timeValuation{\overline{\word}_i} \models \Rename_i$,
such $\word''_i$ is a witness of $\prefix'_i \NCellSim^{\Elementary, \Rename_i}_{\targetLg} \prefix_i$.
Therefore, $\SuffixSet$ can be refined by the simple elementary language $\suffix \in \SimpleElementary$ including $\word''_i$.
% We note that such $\suffix$ may already be in $\SuffixSet$ because $\hypothesisA$ is faithful to the timed observation table only in rows.
% Thus, unlike the Rivest-Schapire algorithm, we cannot use a binary search to find such $\suffix$.
%This is because the DTA construction is faithful to the timed observation table only in row.

%%%%%%%%%%%%%%%%%%%%%%%%%%%%%%%%%%%%%%%%%%%%%%%%%%%%%%%%%%%%
\subsection{L*-style learning algorithm for DTAs}\label{subsection:dta_learning}
%%%%%%%%%%%%%%%%%%%%%%%%%%%%%%%%%%%%%%%%%%%%%%%%%%%%%%%%%%%%

% We show our L*-style algorithm for DTAs.

%-----------------------------------------------------------
\begin{algorithm}[t]
 \caption{Outline of our L*-style algorithm for DTA learning}%
 \label{algorithm:timedLStar}
 \ShortVersion{\scriptsize}
 \DontPrintSemicolon{}
 \newcommand{\myCommentFont}[1]{\texttt{\ShortVersion{\scriptsize}{#1}}}
 \SetCommentSty{myCommentFont}
 \SetKwFunction{FConstructTable}{ConstructTable}
 $\PrefixSet \gets \{(\styleact{\emptyword},\timestamp_0 = 0)\};\; \SuffixSet \gets \{(\styleact{\emptyword}, \timestamp'_0 = 0)\}$\;%
 \label{algorithm:timedLStar:initialize}
 \While{$\top$} {\label{algorithm:timedLStar:main_loop:begin}
   \While{$\TimedObsTableInside$ is not cohesive} {\label{algorithm:timedLStar:close_consistent_loop:begin}
     \If(\tcp*[f]{$\TimedObsTableInside$ is not closed}){$\exists \prefix \in \successor{\PrefixSet} \setminus \PrefixSet.\, \nexists \prefix' \in \PrefixSet.\, \prefix \CellSim^{\SuffixSet}_{\targetLg} \prefix'$} {
         $\PrefixSet \gets \PrefixSet \cup \{\prefix\}$\label{algorithm:timedLStar:close}
         \tcp*{Add such $\prefix$ to $\PrefixSet$}
     }
     \ElseIf{$\exists \prefix,\prefix' \in \PrefixSet, \action \in \Alphabet.\, \prefix\CellSim^{\SuffixSet}_{\targetLg} \prefix' \land \successor[\action]{\prefix}\NCellSim^{\SuffixSet}_{\targetLg} \successor[\action]{\prefix'}$} {
         \tcp{$\TimedObsTableInside$ is inconsistent due to $\action$}
         \KwLet{} $\SuffixSet' \subseteq \SuffixSet$ \KwBe{} a minimal set such that
         $\prefix\NCellSim^{\SuffixSet \cup \{ \{\action \cdot \word \mid \word \in \suffix\} \mid \suffix \in \SuffixSet'\}}_{\targetLg} \prefix'$\;
         $\SuffixSet \gets \SuffixSet \cup \{\{\action \cdot \word \mid \word \in \suffix\} \mid \suffix \in \SuffixSet'\}$\label{algorithm:timedLStar:discrete_inconsistent}\;
     }
     % \ElseIf{$\exists \prefix,\prefix' \in \PrefixSet, \suffix \in \SuffixSet.\, \prefix\CellSim^{\SuffixSet}_{\targetLg} \prefix' \land \prefix\NCellSim^{\SuffixSet \cup \predecessor[t]{\suffix}}_{\targetLg} \prefix'$} {\label{algorithm:timedLStar:continuous_inconsistent:if}
     %     \tcp{$\TimedObsTableInside$ is inconsistent due to dwell time}
     %     $\SuffixSet \gets \SuffixSet \cup \{\predecessor[t]{\suffix}\}$\label{algorithm:timedLStar:continuous_inconsistent}
         % \While{$\prefix\CellSim^{\SuffixSet}_{\targetLg} \prefix'$} {\label{algorithm:timedLStar:continuous_inconsistent:loop_begin}
         %     \KwLet{} $\SuffixSet' \subseteq \SuffixSet$ \KwBe{} a minimal set such that
         %     $\successor[t]{\prefix}\CellSim^{\SuffixSet \setminus \SuffixSet'}_{\targetLg} \successor[t]{\prefix'}$\;
         %     $\SuffixSet \gets \SuffixSet \cup \{\predecessor[t]{\suffix} \mid \suffix \in \SuffixSet'\}$\label{algorithm:timedLStar:continuous_inconsistent}
         % }\label{algorithm:timedLStar:continuous_inconsistent:loop_end}
         % $\SuffixSet \gets \SuffixSet \cup \{\predecessor[t]{\suffix}\}$, where $\suffix \in \SuffixSet$ is such that $\prefix\NCellSim^{\SuffixSet \cup \{\predecessor{\suffix}\}}_{\targetLg} \prefix'$
     % }
     \Else(\tcp*[f]{$\TimedObsTableInside$ is not exterior-consistent}) % {$\exists \prefix \in \PrefixSet.\, \successor[t]{\prefix} \nsubseteq \exterior[t]{\prefix}$}
 {
         $\PrefixSet \gets \PrefixSet \cup \{\prefix' \in \successor[t]{\PrefixSet} \setminus \PrefixSet \mid \exists \prefix \in \PrefixSet.\, \prefix' = \successor[t]{\prefix}\land \prefix' \nsubseteq \exterior[t]{\prefix}\}$\;\label{algorithm:timedLStar:exterior_inconsistent}
     }
     % \Else(\tcp*[f]{$\TimedObsTableInside$ is not time-saturated}) {
     %     $\PrefixSet \gets \PrefixSet \cup \{\successor[t]{\prefix} \mid \exists \prefix \in \PrefixSet.\, \successor[t]{\prefix} \NCellSim^{\SuffixSet,\top}_{\targetLg} \prefix\}$\;\label{algorithm:timedLStar:time_non_saturated}
     % }
     % \ElseIf{$\exists \prefix,\prefix' \in \PrefixSet.\, \prefix \NCellSim^{\SuffixSet}_{\targetLg} \prefix' \land \successor[t]{\prefix} \in \PrefixSet' \land \successor[t]{\prefix'} \in \PrefixSet' \land \exterior[t]{\prefix} \cap \exterior[t]{\prefix'} \neq \emptyset$} {
     %     \KwPick{} a forward regional elementary language $\prefix'' \in \exterior[t]{\prefix} \cap \exterior[t]{\prefix'}$\;
     %     $\PrefixSet \gets \PrefixSet \cup \prefixes{\prefix''}$\;
     % }
     % \Else{
     %     \KwPick{} $\prefix \in \PrefixSet$ and $\prefix' \in (\PrefixSet \cup \PrefixSet') \cap \exterior[t]{\prefix}$ such that $\successor[t]{\prefix} \in \PrefixSet'$ and $\successor[t]{\prefix} \NCellSim^{\SuffixSet}_{\targetLg} \prefix'$\;
     %     $\PrefixSet \gets \PrefixSet \cup \{\successor[t]{\prefix}\}$\;
     % }
     % update $\PrefixSet'$ and $\SuffixSet$ to satisfy the conditions in \cref{def:timed_observation_table}\;
     fill $\Table$ using symbolic membership queries
   }\label{algorithm:timedLStar:close_consistent_loop:end}
   $\hypothesisA \gets \FConstructDTA{\TimedObsTableInsideInside}$\;\label{algorithm:timedLStar:eqQ}
   % \If{$\hypothesisA$ is not faithful to the timed observation table $\TimedObsTableInside$} {
   %   $\PrefixSet \gets \PrefixSet \cup \prefixes{\prefix}$, where $\prefix$ is an evidence of the unfaithfulness\;\label{algorithm:timedLStar:faithful_add_cex}
   % }
   % \Else
   \If{$\cex = \eqQ[\targetLg]{\hypothesisA}$\label{algorithm:timedLStar:eqQ:cex}} {
     \KwPush{} \FHandleCounterExample{$\cex$} \KwTo{} $\SuffixSet$\;\label{algorithm:timedLStar:add_cex}
     % \Switch{\FHandleCounterExample{$\cex$}} {
     %   \lCase{$\suffix$} {
     %     \KwPush{} $\suffix$ \KwTo{} $\SuffixSet$\label{algorithm:timedLStar:add_cex}
     %   }
     %   \lCase{$\bot$} {
     %     \KwPush{} all the prefixes of $\prefix \in \SimpleElementary$ such that $\cex \in \prefix$ \KwTo{} $\PrefixSet$ \label{algorithm:timedLStar:add_cex:prefix}
     %   }
     % }
      %  \While{$\exists \cex' \in \prefixes{\cex}, \prefix \in \successor{\PrefixSet}.\,\cex' \in \prefix$} {
      %      \KwLet{} $\suffix$ \KwBe{} the simple elementary language such that $\cex \in \prefix \cdot \suffix$\;
      %      \If{${\CellSim^{\SuffixSet \cup \{\suffix\}}_{\targetLg}} \neq {\CellSim^{\SuffixSet}_{\targetLg}}$} {
      %          $\SuffixSet \gets \SuffixSet \cup \{\suffix\}$\;
      %          \KwBreak\;
      %      } \lElse{
      %          $\PrefixSet \gets \PrefixSet \cup \{\prefix\}$\label{algorithm:timedLStar:add_cex}
      %      }
      % }
       % \KwLet{} $\prefix$ be the forward regional elementary language containing $\cex$\;
       % $\PrefixSet \gets \PrefixSet \cup \prefixes{\prefix}$\;
   }
   \lElse(\tcp*[f]{It returns $\hypothesisA$ if $\eqQ[\targetLg]{\hypothesisA} = \top$.}) {
       \KwReturn{$\hypothesisA$}\label{algorithm:timedLStar:eqQ:end}
   }
   % \Else (\tcp*[f]{let $\cex = \eqQ[\targetLg]{\hypothesisA}$ and $\prefix$ be the forward regional elementary language containing $\cex$}) {
   %     \label{algorithm:timedLStar:eqQ:cex}
   %     $\PrefixSet \gets \PrefixSet \cup \prefixes{\prefix}$\;\label{algorithm:timedLStar:add_cex}
   % }
 }\label{algorithm:timedLStar:main_loop:end}
\end{algorithm}

\cref{algorithm:timedLStar} outlines our active DTA learning algorithm.
At line~\ref{algorithm:timedLStar:initialize}, we initialize the timed observation table with
 $\PrefixSet = \{(\styleact{\emptyword},\timestamp_0 = 0)\}$ and $\SuffixSet = \{(\styleact{\emptyword},\timestamp'_0 = 0)\}$.
In the loop in lines~\ref{algorithm:timedLStar:main_loop:begin}--\ref{algorithm:timedLStar:main_loop:end},
we refine the timed observation table\LongVersion{ $\TimedObsTableInside$} until the hypothesis DTA $\hypothesisA$ recognizes the target language $\targetLg$, which is checked by equivalence queries. % at \cref{algorithm:timedLStar:eqQ:query}.
% we construct and refine the observation table $\TimedObsTableInside$ until the hypothesis DTA $\hypothesisA$ recognizes the target language $\targetLg$, which is checked by an equivalence query. % at \cref{algorithm:timedLStar:eqQ:query}.
The refinement finishes when
%-----------------------------------------------------------
 the equivalence relation $\CellSim^{\SuffixSet}_{\targetLg}$ defined by the suffixes $\SuffixSet$ converges to $\CellSim^{\Elementary}_{\targetLg}$, and
 the prefixes $\PrefixSet$ covers $\SimpleElementary / {\CellSim^{\Elementary}_{\targetLg}}$.
%-----------------------------------------------------------

In the loop in lines~\ref{algorithm:timedLStar:close_consistent_loop:begin}--\ref{algorithm:timedLStar:close_consistent_loop:end},
we make the timed observation table cohesive.
If the timed observation table is not closed, we move the incompatible row in $\successor{\PrefixSet} \setminus \PrefixSet$ to $\PrefixSet$ (line~\ref{algorithm:timedLStar:close}).
If the timed observation table is inconsistent, we concatenate an event $\action \in \Alphabet$ in front of some of the suffixes in $\SuffixSet$ (line~\ref{algorithm:timedLStar:discrete_inconsistent}).
% The loop in \crefrange{algorithm:timedLStar:continuous_inconsistent:loop_begin}{algorithm:timedLStar:continuous_inconsistent:loop_end} terminates when we have
%  $\{\successor[t]{\prefix} \cdot \suffix \mid \suffix \in \tilde{\SuffixSet} \} \subseteq \{\prefix' \cdot \suffix \mid \prefix' \in \PrefixSet, \suffix \in \SuffixSet \}$,
% where $\tilde{\SuffixSet}$ is the initial $\SuffixSet'$ after each execution of \cref{algorithm:timedLStar:continuous_inconsistent:if}.
If the timed observation table is not exterior-consistent, %or time-saturated,
we move the boundary $\successor[t]{\prefix} \in \successor[t]{\PrefixSet}\setminus\PrefixSet$ satisfying $\successor[t]{\prefix} \nsubseteq \exterior[t]{\prefix}$ to $\PrefixSet$ (\cref{algorithm:timedLStar:exterior_inconsistent}). %or $\successor[t]{\prefix} \NCellSim^{\SuffixSet,\top}_{\targetLg} \prefix$, respectively.
Once we obtain a cohesive timed observation table,
we construct a DTA $\hypothesisA = \FConstructDTA{\TimedObsTableInsideInside}$ and make an equivalence query (line~\ref{algorithm:timedLStar:eqQ}).
% Here, we can easily return a counterexample if the hypothesis DTA is not faithful to the timed observation table.
If we have $\Lg(\hypothesisA) = \targetLg$, we return $\hypothesisA$.
Otherwise, we have a timed word $\cex$ witnessing the difference between the language of the hypothesis DTA $\hypothesisA$ and the target language $\targetLg$.
We refine the timed observation table using \cref{algorithm:handle_counterexample}. %$\cex$.
% Since $\hypothesisA$ is faithful to the timed observation table only in rows,
% we may not be able to refine $\SuffixSet$.
% In such a case, we add all the prefixes of the simple elementary language $\prefix \in \SimpleElementary$ satisfying $\cex \in \prefix$.

% \begin{proposition}
%  [indistinguishability]
%  For any timed automaton $\TA$, simple elementary language $\elementary$, and timed words $\word, \word'$ satisfying $\word \in \elementary$ and $\word' \in \elementary$, we have $\word \in \Lg(\TA) \iff \word' \in \Lg(\TA)$.
% \end{proposition}

%-----------------------------------------------------------
\begin{example}%
 \label{example:DTA_learning}
 %-----------------------------------------------------------
 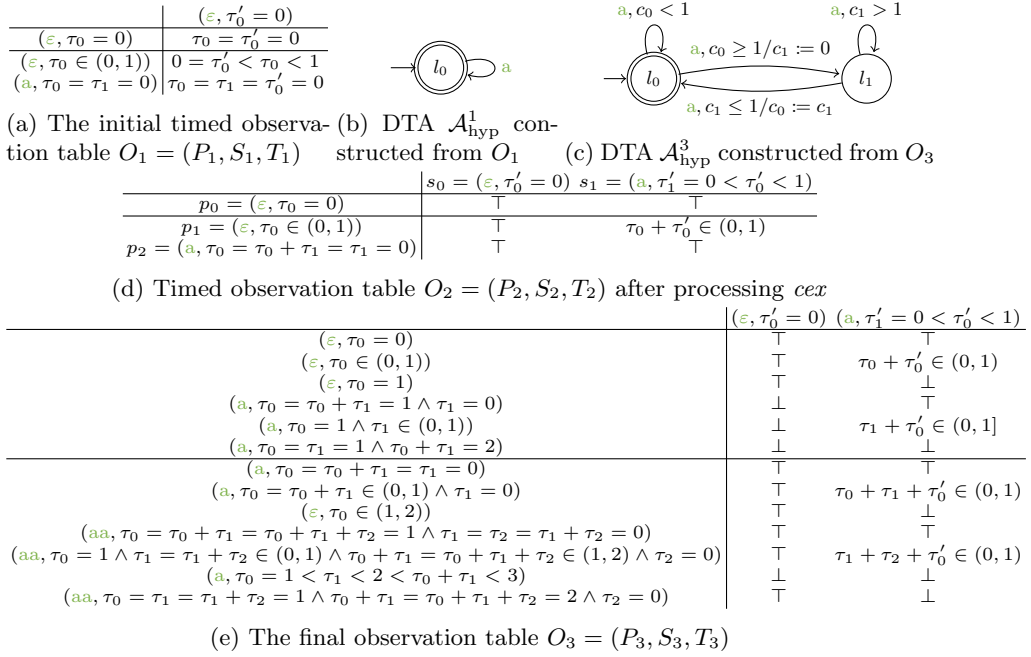
\begin{figure}[t]
 \begin{subfigure}[b]{.35\textwidth}
%  \begin{minipage}[c]{.70\linewidth}
  \scriptsize
   \centering
   \begin{tabular}{c|c}
    & $(\styleact{\emptyword},\timestamp'_0 = 0)$\\\hline
    $(\styleact{\emptyword},\timestamp_0 = 0)$ & $\timestamp_0 = \timestamp'_0 = 0$\\\hline
    $(\styleact{\emptyword}, \timestamp_0 \in (0,1))$ & $0 = \timestamp'_0 < \timestamp_0 < 1$\\
    $(\styleact{a}, \timestamp_0 = \timestamp_1 = 0)$ & $\timestamp_0 = \timestamp_1 = \timestamp'_0 = 0$
   \end{tabular}
  % \end{minipage}
  \caption{The initial timed observation table $\ObsTableWithInside[_1]$}%
  \label{figure:DTA_learning:initial_observation_table}
 \end{subfigure}
  % \hfill
 \begin{subfigure}[b]{.24\textwidth}
  % \begin{minipage}[c]{.25\linewidth}
   \centering
   \begin{tikzpicture}[node distance=2.5cm,on grid,auto,scale=0.8,every node/.style={transform shape},every initial by arrow/.style={initial text={}}]
    \node[state,initial, accepting] (s_0)  {$\loc_0$};

    \path[->]
    % (s_0) edge  [loop above] node {$\styleact{\emptyword}, \clock_0 > 0/\clock_0 \coloneqq 0$} (s_0)
    (s_0) edge  [loop right] node {$\styleact{a}$} (s_0)
    ;
   \end{tikzpicture}
  % \end{minipage}
  \caption{DTA $\hypothesisA^1$ constructed from $\ObsTable_1$}%
  \label{figure:DTA_learning:initial_DTA}
 \end{subfigure}
% \end{figure}
 %-----------------------------------------------------------
%  \hfill
 \begin{subfigure}[b]{.4\textwidth}
  \centering
  \begin{tikzpicture}[shorten >=1pt,scale=0.8,every node/.style={transform shape},every initial by arrow/.style={initial text={}}]
   %% locations
   \node[initial,state,accepting] (l0) at (0,0) [align=center]{$\loc_0$};
   \node[state] (l1) at (3.5,0) [align=center]{$\loc_1$};

   %% edges
   \path[->]
   (l0) edge [bend left=10] node[above] {$\styleact{a}, \clock_0 \geq 1 / \clock_1 \coloneqq 0$} (l1)
   (l0) edge [loop above] node {$\styleact{a}, \clock_0 < 1$} (l0)
   (l1) edge [bend left=10] node[below] {$\styleact{a}, \clock_1 \leq 1 / \clock_0 \coloneqq \clock_1$} (l0)
   (l1) edge [loop above] node {$\styleact{a}, \clock_1 > 1$} (l1)
   ;
  \end{tikzpicture}
  \caption{DTA $\hypothesisA^3$ constructed from $\ObsTable_3$}%
  \label{figure:DTA_learning:final_hypothesis}
 \end{subfigure}
  \hfill
 %-----------------------------------------------------------
% \begin{figure}[tp]
 \begin{subfigure}[b]{1.0\textwidth}
  \scriptsize
  \centering
  \begin{tabular}{c|c c}
   & $\suffix_0 = (\styleact{\emptyword},\timestamp'_0 = 0)$ & $\suffix_1 = (\styleact{a},\timestamp'_1 = 0 < \timestamp'_0 < 1)$\\\hline
   $\prefix_0 = (\styleact{\emptyword},\timestamp_0 = 0)$ & $\top$  & $\top$\\ \hline
   $\prefix_1 = (\styleact{\emptyword}, \timestamp_0 \in (0,1))$ & $\top$ & $\timestamp_0 + \timestamp'_0 \in (0,1)$\\
   $\prefix_2 = (\styleact{a}, \timestamp_0 = \timestamp_0 + \timestamp_1 = \timestamp_1 = 0)$ & $\top$ & $\top$
  \end{tabular}
  \caption{Timed observation table $\ObsTableWithInside[_2]$ after processing $\cex$}%
  \label{figure:DTA_learning:observation_table_after_cex}
 \end{subfigure}
% \end{figure}
 %-----------------------------------------------------------
 %-----------------------------------------------------------
% \begin{figure}[t]
 \begin{subfigure}[b]{1.0\textwidth}
  \centering
  \scriptsize
  \begin{tabular}{c|c c c}
   & $(\styleact{\emptyword},\timestamp'_0 = 0)$ & $(\styleact{a},\timestamp'_1 = 0 < \timestamp'_0 < 1)$\\\hline
   $(\styleact{\emptyword},\timestamp_0 = 0)$ & $\top$ & $\top$\\
   $(\styleact{\emptyword}, \timestamp_0 \in (0,1))$ & $\top$ & $\timestamp_0 + \timestamp'_0 \in (0,1)$\\
   $(\styleact{\emptyword}, \timestamp_0 = 1)$ & $\top$ & $\bot$\\
   $(\styleact{a}, \timestamp_0 = \timestamp_0 + \timestamp_1 = 1 \land \timestamp_1 = 0)$ & $\bot$ & $\top$\\
   $(\styleact{a}, \timestamp_0 = 1 \land \timestamp_1 \in (0,1))$ & $\bot$ & $\timestamp_1 + \timestamp'_0 \in (0,1]$\\
   $(\styleact{a}, \timestamp_0 = \timestamp_1 = 1 \land \timestamp_0 + \timestamp_1 = 2)$ & $\bot$ & $\bot$\\
   \hline
   $(\styleact{a}, \timestamp_0 = \timestamp_0 + \timestamp_1 = \timestamp_1 = 0)$ & $\top$ & $\top$\\
   $(\styleact{a}, \timestamp_0 = \timestamp_0 + \timestamp_1 \in (0,1) \land \timestamp_1 = 0)$ & $\top$ & $\timestamp_0 + \timestamp_1 + \timestamp'_0 \in (0,1)$\\
   $(\styleact{\emptyword}, \timestamp_0 \in (1,2))$ & $\top$ & $\bot$\\
   $(\styleact{a} \styleact{a}, \timestamp_0 = \timestamp_0 + \timestamp_1 = \timestamp_0 + \timestamp_1 + \timestamp_2 = 1 \land \timestamp_1 = \timestamp_2 = \timestamp_1 + \timestamp_2 = 0)$ & $\top$ & $\top$\\
   $(\styleact{a}\styleact{a}, \timestamp_0 = 1 \land \timestamp_1 = \timestamp_1 + \timestamp_2 \in (0,1) \land \timestamp_0 + \timestamp_1 = \timestamp_0 + \timestamp_1 + \timestamp_2 \in (1, 2) \land \timestamp_2 = 0)$ & $\top$ & $\timestamp_1 + \timestamp_2 + \timestamp'_0 \in (0,1)$\\
   $(\styleact{a}, \timestamp_0 = 1 < \timestamp_1 < 2 < \timestamp_0 + \timestamp_1 < 3)$ & $\bot$ & $\bot$\\
   $(\styleact{a} \styleact{a}, \timestamp_0 = \timestamp_1 = \timestamp_1 + \timestamp_2 = 1 \land \timestamp_0 + \timestamp_1 = \timestamp_0 + \timestamp_1 + \timestamp_2 = 2 \land \timestamp_2 = 0)$ & $\top$ & $\bot$
  \end{tabular}
  \caption{The final observation table $\ObsTableWithInside[_3]$}%
  \label{figure:DTA_learning:final_observation_table}
% \end{figure}
 \end{subfigure}
 %-----------------------------------------------------------
% \begin{figure}[t]
  \caption{Timed observation tables $\ObsTable_1, \ObsTable_2, \ObsTable_3$, and the DTAs $\hypothesisA^1$ and $\hypothesisA^3$ made from $\ObsTable_1$ and $\ObsTable_3$, respectively. In $\ObsTable_2$ and $\ObsTable_3$, we only show the constraints non-trivial from $\prefix$ and $\suffix$. The DTAs are simplified without changing the language. The use of clock assignments, which does not change the expressiveness, is from~\cite{MP04}.}%
  \label{figure:DTA_learning:example}
\end{figure}
 %-----------------------------------------------------------
 Let $\targetLg$ be the timed language recognized by the DTA in \cref{figure:timed_automaton}.
 We start from $\PrefixSet = \{(\styleact{\emptyword},\timestamp_0 = 0)\}$ and $\SuffixSet = \{(\styleact{\emptyword},\timestamp'_0 = 0\}$.
 \cref{figure:DTA_learning:initial_observation_table} shows the initial timed observation table $\ObsTable_1$.
 Since the timed observation table $\ObsTable_1$ in \cref{figure:DTA_learning:initial_observation_table} is cohesive,
 we construct a hypothesis DTA $\hypothesisA^1$.\@
 The hypothesis recognizable timed language is $(\PrefixSet_1, \RecognizableFinal_1, \CRM_1)$ is such that $\PrefixSet_1 = \RecognizableFinal_1 = \{(\styleact{\emptyword}, \timestamp_{0} = 0)\}$ and
 $\CRM_{1} = \{(\styleact{\emptyword}, \timestamp_{0} > 0, \styleact{\emptyword}, \timestamp_{0}, \top), (\styleact{a}, \timestamp_{0} = \timestamp_{0} + \timestamp_{1} = \timestamp_{1} = 0, \styleact{\emptyword}, \timestamp_{0}, \top)\}$.
 \cref{figure:DTA_learning:initial_DTA} shows the first hypothesis DTA $\hypothesisA^1$.

 We have $\Lg(\hypothesisA^1) \neq \targetLg$, and the learner obtains a counterexample, \eg{} $\cex = 1.0 \cdot \styleact{a} \cdot 0$, with an equivalence query.
 In \cref{algorithm:handle_counterexample}, we have
 $\word_0 = \cex$,
 $\word_1 = 0.5 \cdot \styleact{a} \cdot 0$,
 $\word_2 = 0 \cdot \styleact{a} \cdot 0$, and
 $\word_3 = 0$.
 We have
 $\word_0 \not\in \Lg(\hypothesisA^1) \setdiff \targetLg$ and
 $\word_1 \in \Lg(\hypothesisA^1) \setdiff \targetLg$, and
 the suffix to distinguish $\word_0$ and $\word_1$ is $0.5 \cdot \styleact{a} \cdot 0$.
 Thus, we add
 $\suffix_1 = (\styleact{a}, \timestamp'_1 = 0 < \timestamp'_0 = \timestamp'_0 + \timestamp'_1 < 1)$ to $\SuffixSet_1$ (\cref{figure:DTA_learning:observation_table_after_cex}).

 In \cref{figure:DTA_learning:observation_table_after_cex},
we observe that $\TableCell[_2]{\prefix_1}{\suffix_1}$ is more strict than $\TableCell[_2]{\prefix_0}{\suffix_1}$, and
 we have $\prefix_1 \NCellSim^{\SuffixSet_2}_{\targetLg} \prefix_0$.
 To make $\ObsTableInside[_2]$ closed, we add $\prefix_1$ to $\PrefixSet_2$.
 By repeating similar operations,
 % Similarly, by making the the timed observation table closed,
 we obtain the timed observation table $\ObsTableWithInside[_3]$ in \cref{figure:DTA_learning:final_observation_table}, 
 which is cohesive.
 \cref{figure:DTA_learning:final_hypothesis} shows the DTA $\hypothesisA^3$ constructed from $\ObsTable_3$.
 Since $\Lg(\hypothesisA^3) = \targetLg$ holds, \cref{algorithm:timedLStar} finishes returning $\hypothesisA^3$.
\end{example}
%-----------------------------------------------------------

By the use of equivalence queries, \cref{algorithm:timedLStar} returns a DTA recognizing the target language if it terminates, which is formally as follows.

%-----------------------------------------------------------
\begin{theorem}
 [correctness]
 For any target timed language $\targetLg$,
 if \cref{algorithm:timedLStar} terminates,
 for the resulting DTA $\hypothesisA$,
 $\Lg(\hypothesisA) = \targetLg$ holds.
 \qed{}
\end{theorem}
%-----------------------------------------------------------

%
Moreover, \cref{algorithm:timedLStar} terminates for any recognizable timed language $\targetLg$\LongVersion{.
This is} essentially because of the finiteness of $\SimpleElementary / {\CellSim^{\Elementary}_{\targetLg}}$.

%-----------------------------------------------------------
\begin{theorem}
 [termination]%
 \label{theorem:termination}
 For any recognizable timed language $\targetLg$,
 \cref{algorithm:timedLStar} terminates and returns a DTA $\A$ satisfying $\Lg(\A) = \targetLg$.
\end{theorem}
%-----------------------------------------------------------

%-----------------------------------------------------------
\begin{proof}
 [\cref{theorem:termination}]
 By the recognizability of $\targetLg$ and \cref{theorem:finiteness}, %\cref{lemma:timed_recognizable_to_finite},
 $\SimpleElementary / {\CellSim^{\Elementary}_{\targetLg}}$ is finite.
 Let $N = |\SimpleElementary / {\CellSim^{\Elementary}_{\targetLg}}|$.
 Since each execution of \cref{algorithm:timedLStar:close} adds $\prefix$ to $\PrefixSet$, where $\prefix$ is such that
 for any $\prefix' \in \PrefixSet$, $\prefix \NCellSim^{\Elementary}_{\targetLg} \prefix'$ holds,
 it is executed at most $N$ times.
 % By \cref{lemma:consistency_in_limit}, 
 Since each execution of \cref{algorithm:timedLStar:discrete_inconsistent}\LongVersion{ strictly} refines $\SuffixSet$,
 \ie{} it increases $|\SimpleElementary / {\CellSim^{\SuffixSet}_{\targetLg}}|$,
 \cref{algorithm:timedLStar:discrete_inconsistent} is executed at most $N$ times.
 For any\LongVersion{ simple elementary language} $\fractionalElementary \in \SimpleElementary$,
 if $\timedCondition$ contains $\sumTimestamp{i}{|u|} = \dConstant$ for some $i \in \{0, 1, \dots, |u|\}$ and $\dConstant \in \N$,
 we have $\successor[t]{\fractionalElementary} \subseteq \exterior[t]{\elementary}$.
 Therefore, \cref{algorithm:timedLStar:exterior_inconsistent} is executed at most $N$ times.
 Since $\SuffixSet$ is strictly refined in \cref{algorithm:timedLStar:add_cex},
 \ie{} it increases $|\SimpleElementary / {\CellSim^{\SuffixSet}_{\targetLg}}|$, \cref{algorithm:timedLStar:add_cex} is executed at most $N$ times.
 By \cref{lemma:correct_in_limit}, once $\CellSim^{\SuffixSet}_{\targetLg}$ saturates to $\CellSim^{\Elementary}_{\targetLg}$, $\FConstructDTA$ returns the correct DTA.\@
 Overall, \cref{algorithm:timedLStar} terminates.
 \qed{}
\end{proof}
%-----------------------------------------------------------

%%%%%%%%%%%%%%%%%%%%%%%%%%%%%%%%%%%%%%%%%%%%%%%%%%%%%%%%%%%%
\subsection{Learning with a normal teacher}\label{section:symbolic_membership_oracle}
%%%%%%%%%%%%%%%%%%%%%%%%%%%%%%%%%%%%%%%%%%%%%%%%%%%%%%%%%%%%

We briefly show how to learn a DTA only with membership and equivalence queries.
See \cref{appendix:detail_symbolic_membership} for detail.
We reduce a symbolic membership query to finitely many membership queries, which can be answered by a normal teacher.

Let $\elementary$ be the elementary language given in a symbolic membership query.
Since $\timedCondition$ is bounded, we can construct a finite and disjoint set of simple and canonical timed conditions $\timedCondition'_1, \timedCondition'_2,\dots, \timedCondition'_n$ satisfying
$\bigvee_{1 \leq i \leq n} \timedCondition'_i = \timedCondition$ by a simple enumeration.
For any simple elementary language $\elementary['] \in \SimpleElementary$ and timed words $\word, \word' \in \elementary[']$, we have $\word \in \Lg \iff \word' \in \Lg$.
Thus, we can construct $\symbolicMemQ{\Lg}(\elementary)$ by making a membership query $\memQ{\word}$ for each such $\elementary['] \subseteq \elementary$ and for some $\word \in \elementary[']$.
We need such an exhaustive search, instead of a binary search, because $\symbolicMemQ{\Lg}(\elementary)$ may be non-convex.

Assume $\timedCondition$ is a canonical timed condition.
Let $M$ be the size of the variables in $\timedCondition$ and $I$ be the largest difference between the upper bound and the lower bound for some $\sumTimestamp{i}{j}$ in $\timedCondition$.
The size $n$ of the above decomposition is bounded by ${(2 \times I + 1)}^{1/2 \times M \times (M + 1)}$, which exponentially blows up with respect to $M$.

 In our algorithm, we only make symbolic membership queries with elementary languages of the form $\prefix \cdot \suffix$, where $\prefix$ and $\suffix$ are simple elementary languages.
Therefore, $I$ is at most 2.
However, even with such an assumption, the number of the necessary membership queries blows up exponentially to the size of the variables in $\timedCondition$.

%%%%%%%%%%%%%%%%%%%%%%%%%%%%%%%%%%%%%%%%%%%%%%%%%%%%%%%%%%%%
\subsection{Complexity analysis}\label{section:complexity_analysis}
%%%%%%%%%%%%%%%%%%%%%%%%%%%%%%%%%%%%%%%%%%%%%%%%%%%%%%%%%%%%
After each equivalence query, our DTA learning algorithm strictly refines $\SuffixSet$ or terminates.
Thus, the number of equivalence queries is at most $N$.
% Let $N'$ be the maximum number of the prefixes of the simple language containing the counterexample returned by equivalence queries.
In the proof of \cref{theorem:termination}, we observe that the size of $\PrefixSet$ is at most $2 N$.
Therefore, the number $(|\PrefixSet| + |\successor{\PrefixSet}|) \times |\SuffixSet|$ of the cells in the timed observation table is at most $(2 N + 2 N \times (|\Alphabet| + 1)) \times N = 2 N^2 |\Alphabet| + 2$.
Let $J$ be the upper bound of $i$ in the analysis of $\cex$ returned by equivalence queries (\cref{algorithm:handle_counterexample}).
For each equivalence query, the number of membership queries in \cref{algorithm:handle_counterexample} is bounded by $\lceil \log J \rceil$, and thus,
it is, in total, bounded by $N \times \lceil \log J \rceil$.
% For each equivalence query, the number of membership queries in \cref{algorithm:handle_counterexample} is bounded by $J$, and thus,
% it is, in total, bounded by $N \times J$.
Therefore, if the learner can use symbolic membership queries, the total number of queries is bounded by a polynomial of $N$ and $J$.
In \cref{section:symbolic_membership_oracle}, we observe that the number of membership queries to implement a symbolic membership query is at most exponential to $M$.
Since $\PrefixSet$ is prefix-closed, $M$ is at most $N$.
Overall, if the learner cannot use symbolic membership queries, the total number of queries is at most exponential to $N$.

Let $\targetA = \TAInside$ be a DTA recognizing $\targetLg$.
As we observe in the proof of \cref{lemma:timed_recognizable_to_finite},
$N$ is bounded by the size of the state space of the \LongVersion{RA}\ShortVersion{region automaton~\cite{AD94}} of $\targetA$,
\LongVersion{and} $N$ is at most $|\Clock|!\times 2^{|\Clock|}\times\prod_{\clock \in \Clock} (2 K_{\clock} + 2) \times |\Loc|$, where $K_\clock$ is the largest constant compared with $\clock \in \Clock$ in $\targetA$.
Thus, without symbolic membership queries, the total number of queries is at most doubly-exponential to $|\Clock|$ and singly exponential to $|\Loc|$.
We remark that when $|\Clock| = 1$, the total number of queries is at most singly exponential to $|\Loc|$ and $K_\clock$,
which coincides with the worst-case complexity of the one-clock DTA learning algorithm in~\cite{XAZ22}.

%%%%%%%%%%%%%%%%%%%%%%%%%%%%%%%%%%%%%%%%%%%%%%%%%%%%%%%%%%%%
%%%%%%%%%%%%%%%%%%%%%%%%%%%%%%%%%%%%%%%%%%%%%%%%%%%%%%%%%%%%
\section{Experiments}\label{section:experiments}
%%%%%%%%%%%%%%%%%%%%%%%%%%%%%%%%%%%%%%%%%%%%%%%%%%%%%%%%%%%%
%%%%%%%%%%%%%%%%%%%%%%%%%%%%%%%%%%%%%%%%%%%%%%%%%%%%%%%%%%%%

We experimentally evaluated our DTA learning algorithm using our prototype library \ourTool{}\footnote{\ourTool{} is publicly available at \url{https://github.com/masWag/LearnTA}. The artifact of the experiments is available at \url{https://doi.org/10.5281/zenodo.7875383}.} implemented in C++.
In \ourTool{}, the equivalence queries are answered by a zone-based reachability analysis using the fact that DTAs are closed under complement~\cite{AD94}.
We pose the following research questions.
\begin{itemize}
 \item[RQ1] How is the scalability of \ourTool{}\LongVersion{ with respect} to the language complexity?
 \item[RQ2] How is the efficiency of \ourTool{} for practical benchmarks?
% \item[RQ3] How is the efficiency of \ourTool{} compared with one-clock DTA learning algorithms?
\end{itemize}

For the benchmarks with one clock variable, we compared \ourTool{} with one of the latest one-clock DTA learning algorithms~\cite{XAZ22,Leslieaj/DOTALearningSMT}, which we call \DOTA{}.
\DOTA{} is implemented in Python with Z3~\cite{DBLP:conf/tacas/MouraB08} for constraint solving.

For each execution, we measured 
\begin{myitemize}
 \item the number of queries and
 \item the total execution time, including the time to answer the queries.
\end{myitemize}
For the number of queries, we report the number with memoization, \ie{} we count the number of the queried timed words (for membership queries) and the counterexamples (for equivalence queries).
We conducted all the experiments on a computing server with Intel Core i9-10980XE 125 GiB RAM that runs Ubuntu 20.04.5 LTS.\@
We used 3 hours as the timeout.

\begin{table}[tbp]
 \caption{Summary of the results for \Random{}. Each row index $|\Loc|\_|\Alphabet|\_K_{\Clock}$ shows the number of locations, the alphabet size, and the upper bound of the maximum constant in the guards, respectively. The row ``count'' shows the number of instances finished in 3 hours. Cells with the best results are highlighted.}%
 \label{table:experiment_results:random}
 \scriptsize
 \centering
 \begin{tabular}{llrrrrrrrrrr}
\toprule
 &  & \multicolumn{3}{c}{\# of Mem.\ queries} & \multicolumn{3}{c}{\# of Eq.\ queries} & \multicolumn{3}{c}{Exec.\ time [sec.]} & count \\
 &  & max & mean & min & max & mean & min & max & mean & min &  \\
\midrule
\multirow[c]{2}{*}{3\_2\_10} & \ourTool{} & 35,268 & 14,241 & 2,830 & \tbcolor{}11 & \tbcolor{}6 & \tbcolor{}4 & 2.32e+00 & 6.68e-01 & \tbcolor{}4.50e-02 & \tbcolor{}10/10 \\
 & \DOTA{} & \tbcolor{} 468 & \tbcolor{} 205 & \tbcolor{}32 & 13 & 8 & 5 & \tbcolor{}9.58e-01 & \tbcolor{}2.89e-01 & 6.58e-02 & \tbcolor{}10/10 \\
\multirow[c]{2}{*}{4\_2\_10} & \ourTool{} & 194,442 & 55,996 & 10,619 & \tbcolor{}14 & \tbcolor{}7 & \tbcolor{}4 & 2.65e+01 & 7.98e+00 & 4.88e-01 & \tbcolor{}10/10 \\
 & \DOTA{} & \tbcolor{}985 & \tbcolor{}451 & \tbcolor{}255 & 16 & 12 & 7 & \tbcolor{}3.53e-01 & \tbcolor{}2.09e-01 & \tbcolor{}1.27e-01 & \tbcolor{}10/10 \\
\multirow[c]{2}{*}{4\_4\_20} & \ourTool{} & 1,681,769 & 858,759 & 248,399 & \tbcolor{}21 & \tbcolor{}15 & \tbcolor{}10 & 8.34e+03 & 1.41e+03 & 3.23e+01 & 8/10 \\
 & \DOTA{} & \tbcolor{}5,329 & \tbcolor{}3,497 & \tbcolor{}1,740 & 42 & 32 & 26 & \tbcolor{}2.19e+00 & \tbcolor{}1.42e+00 & \tbcolor{}8.27e-01 & \tbcolor{}10/10 \\
\multirow[c]{2}{*}{5\_2\_10} & \ourTool{} & 627,980 & 119,906 & 8,121 & \tbcolor{}19 & \tbcolor{}8 & \tbcolor{}5 & 1.67e+02 & 2.28e+01 & \tbcolor{}1.96e-01 & \tbcolor{}10/10 \\
 & \DOTA{} & \tbcolor{}1,332 & \tbcolor{}876 & \tbcolor{}359 & 22 & 16 & 12 & \tbcolor{}5.20e-01 & \tbcolor{}3.66e-01 & 2.58e-01 & \tbcolor{}10/10 \\
\multirow[c]{2}{*}{6\_2\_10} & \ourTool{} & 555,939 & 106,478 & 2,912 & \tbcolor{}14 & \tbcolor{}9 & \tbcolor{}6 & 2.44e+02 & 2.81e+01 & \tbcolor{}4.40e-02 & \tbcolor{}10/10 \\
 & \DOTA{} & \tbcolor{}3,929 & \tbcolor{}1,894 & \tbcolor{}104 & 35 & 20 & 11 & \tbcolor{}1.72e+00 & \tbcolor{}8.01e-01 & 1.73e-01 & \tbcolor{}10/10 \\
\bottomrule
\end{tabular}

\end{table}
%-----------------------------------------------------------
%%%%%%%%%%%%%%%%%%%%%%%%%%%%%%%%%%%%%%%%%%%%%%%%%%%%%%%%%%%%
\subsection{RQ1: Scalability with respect to the language complexity}\label{section:scalability}
%%%%%%%%%%%%%%%%%%%%%%%%%%%%%%%%%%%%%%%%%%%%%%%%%%%%%%%%%%%%

To evaluate the scalability of \ourTool{}, we used randomly generated DTAs from~\cite{ACZZZ20} (denoted as \Random{}) and our original DTAs (denoted as \unbalanced{}).
\Random{} consists of five classes: $3\_2\_10$,  $4\_2\_10$,  $4\_4\_20$, $5\_2\_10$, and $6\_2\_10$,
where each value of $|\Loc|\_|\Alphabet|\_K_{\Clock}$ is the number of locations, the alphabet size, and the upper bound of the maximum constant in the guards in the DTAs, respectively.
Each class consists of 10 randomly generated DTAs.
\unbalanced{} is our original benchmark inspired by the ``unbalanced parentheses'' timed language from~\cite{ACM02}.
\unbalanced{} consists of five DTAs with different complexity of timing constraints.
\cref{table:experiment_results:unbalanced} summarizes their complexity.

%-----------------------------------------------------------
\begin{figure}[tp]
 \begin{subfigure}{.295\linewidth}
  \centering
  \scalebox{0.4}{% This file was created with tikzplotlib v0.10.1.
\begin{tikzpicture}
\huge
\definecolor{darkgray176}{RGB}{176,176,176}
\definecolor{darkorange25512714}{RGB}{255,127,14}
\definecolor{lightgray204}{RGB}{204,204,204}
\definecolor{steelblue31119180}{RGB}{31,119,180}

\begin{axis}[
legend cell align={left},
legend style={
  fill opacity=0.8,
  draw opacity=1,
  text opacity=1,
  at={(0.03,0.97)},
  anchor=north west,
  draw=lightgray204
},
tick align=outside,
tick pos=left,
x grid style={darkgray176},
xlabel={\# of locations},
xmin=2.85, xmax=6.15,
xtick style={color=black},
y grid style={darkgray176},
ylabel={\# of Mem.\ queries},
ymin=-5780.07, ymax=125891.47,
ytick style={color=black}
]
\addplot [very thick, steelblue31119180, mark=triangle*, mark size=5, mark options={solid}]
table {%
3 14241.7
4 55996.8
5 119906.4
6 106478.5
};
\addlegendentry{\Large \ourTool{}}
\addplot [very thick, darkorange25512714, mark=asterisk, mark size=5, mark options={solid}]
table {%
3 205
4 451
5 876
6 1894.8
};
\addlegendentry{\Large \DOTA{}}
\end{axis}

\end{tikzpicture}}  
  \caption{Membership queries}%
  \label{figure:experiment_results:queries_vs_locations:membership}
 \end{subfigure}
 \hfill
 \begin{subfigure}{.394\linewidth}
  \centering
  \scalebox{0.4}{% This file was created with tikzplotlib v0.10.1.
\begin{tikzpicture}
\huge
\definecolor{darkgray176}{RGB}{176,176,176}
\definecolor{darkorange25512714}{RGB}{255,127,14}
\definecolor{lightgray204}{RGB}{204,204,204}
\definecolor{steelblue31119180}{RGB}{31,119,180}

\begin{axis}[
legend cell align={left},
legend style={
  fill opacity=0.8,
  draw opacity=1,
  text opacity=1,
  at={(0.91,0.5)},
  anchor=east,
  draw=lightgray204
},
log basis y={10},
tick align=outside,
tick pos=left,
x grid style={darkgray176},
xlabel={\# of locations},
xmin=2.85, xmax=6.15,
xtick style={color=black},
y grid style={darkgray176},
ylabel={\# of Mem.\ queries},
ymin=149.073152557152, ymax=164890.938296727,
ymode=log,
ytick style={color=black},
ytick={10,100,1000,10000,100000,1000000,10000000},
yticklabels={
  \(\displaystyle {10^{1}}\),
  \(\displaystyle {10^{2}}\),
  \(\displaystyle {10^{3}}\),
  \(\displaystyle {10^{4}}\),
  \(\displaystyle {10^{5}}\),
  \(\displaystyle {10^{6}}\),
  \(\displaystyle {10^{7}}\)
}
]
\addplot [very thick, steelblue31119180, mark=triangle*, mark size=5, mark options={solid}]
table {%
3 14241.7
4 55996.8
5 119906.4
6 106478.5
};
\addlegendentry{\Large \ourTool{}}
\addplot [very thick, darkorange25512714, mark=asterisk, mark size=5, mark options={solid}]
table {%
3 205
4 451
5 876
6 1894.8
};
\addlegendentry{\Large \DOTA{}}
\end{axis}

\end{tikzpicture}}  
  \caption{Membership queries (log scale)}%
  \label{figure:experiment_results:queries_vs_locations:membership_log}
 \end{subfigure}
 \hfill
 \begin{subfigure}{.295\linewidth}
  \centering
  \scalebox{0.4}{% This file was created with tikzplotlib v0.10.1.
\begin{tikzpicture}
\huge
\definecolor{darkgray176}{RGB}{176,176,176}
\definecolor{darkorange25512714}{RGB}{255,127,14}
\definecolor{lightgray204}{RGB}{204,204,204}
\definecolor{steelblue31119180}{RGB}{31,119,180}

\begin{axis}[
legend cell align={left},
legend style={
  fill opacity=0.8,
  draw opacity=1,
  text opacity=1,
  at={(0.03,0.97)},
  anchor=north west,
  draw=lightgray204
},
tick align=outside,
tick pos=left,
x grid style={darkgray176},
xlabel={\# of locations},
xmin=2.85, xmax=6.15,
xtick style={color=black},
y grid style={darkgray176},
ylabel={\# of Eq.\ queries},
ymin=5.89, ymax=21.51,
ytick style={color=black}
]
\addplot [very thick, steelblue31119180, mark=triangle*, mark size=5, mark options={solid}]
table {%
3 6.6
4 7.9
5 8.9
6 9.9
};
\addlegendentry{\Large \ourTool{}}
\addplot [very thick, darkorange25512714, mark=asterisk, mark size=5, mark options={solid}]
table {%
3 8.8
4 12.1
5 16.2
6 20.8
};
\addlegendentry{\Large \DOTA{}}
\end{axis}

\end{tikzpicture}}  
  \caption{Equivalence queries}%
  \label{figure:experiment_results:queries_vs_locations:equivalence}
 \end{subfigure}
 \caption{The number of locations and the number of queries for $|\Loc|\_2\_10$ in \Random{}, where $|\Loc| \in \{3,4,5,6\}$}%
 \label{figure:experiment_results:queries_vs_locations}
\end{figure}
%-----------------------------------------------------------
%-----------------------------------------------------------
\begin{table}[tbp]
 \caption{Summary of the target DTAs and the results for \unbalanced{}. $|\Loc|$ is the number of locations, $|\Alphabet|$ is the alphabet size, $|\Clock|$ is the number of  clock variables, and $K_{\Clock}$ is the maximum constant in the guards in the DTA.}%
 \label{table:experiment_results:unbalanced}
 \scriptsize
 \centering
 \begin{tabular}{llrrrrrrr}
\toprule
 & & $|\Loc|$ & $|\Alphabet|$ & $|\Clock|$ & $K_{\Clock}$ & \# of Mem.\ queries & \# of Eq.\ queries & Exec.\ time [sec.] \\
\midrule
\unbalanced{}:1 & \ourTool{} & 5 & 1 & 1 & 2 & 51 & 2 & 2.00e-03 \\
\unbalanced{}:2 & \ourTool{} & 5 & 1 & 2 & 4 & 576,142 & 3 & 3.64e+01 \\
\unbalanced{}:3 & \ourTool{} & 5 & 1 & 3 & 4 & 403,336 & 4 & 2.24e+01 \\
\unbalanced{}:4 & \ourTool{} & 5 & 1 & 4 & 6 & 4,142,835 & 5 & 2.40e+02 \\
\unbalanced{}:5 & \ourTool{} & 5 & 1 & 5 & 6 & 10,691,400 & 5 & 8.68e+02 \\
\bottomrule
\end{tabular}

\end{table}
%-----------------------------------------------------------
% \todo{The prefixes and successors are quite different in OneSMT and LearnTA: OneSMT's definitions are discrete (e.g., prefixes are only wrt events with time elapse), whereas ours are both continuous and discrete (e.g., we also consider prefixes by trimming the dwell time in the end). Since our definition makes significantly more prefixes, LearnTA tends to require much more membership queries. However, OneSMT mines timing constraints mainly by equivalence queries and tends to require more equivalence queries.}
\cref{table:experiment_results:random,figure:experiment_results:queries_vs_locations} summarize the results for \Random{}, and
\cref{table:experiment_results:unbalanced} summarizes the results for \unbalanced{}.
\cref{table:experiment_results:random} shows that \ourTool{} requires more membership queries than \DOTA{}.
This is likely because of the difference in the definition of prefixes and successors:
\DOTA{}'s definitions are discrete (\eg{} prefixes are only with respect to events with time elapse), whereas ours are both continuous and discrete (\eg{} we also consider prefixes by trimming the dwell time in the end);
Since our definition makes significantly more prefixes, LearnTA tends to require much more membership queries.
Another, more high-level reason is that \ourTool{} learns a DTA without knowing the number of the clock variables, and many more timed words are potentially helpful for learning.
\cref{table:experiment_results:random} shows that \ourTool{} requires significantly many membership queries for $4\_4\_20$.
This is likely because of the exponential blowup with respect to $K_{\Clock}$, as discussed in \cref{section:complexity_analysis}.
In \cref{figure:experiment_results:queries_vs_locations}, we observe that for both \ourTool{} and \DOTA{}, the number of membership queries increases nearly exponentially to the number of locations.
This coincides with the discussion in \cref{section:complexity_analysis}.

In contrast, \cref{table:experiment_results:random} shows that \ourTool{} requires fewer equivalence queries than \DOTA{}.
This suggests that the cohesion in \cref{definition:closed_consistent_exterior_consistent} successfully detected contradictions in observation before generating a hypothesis, whereas \DOTA{} mines timing constraints mainly by equivalence queries and tends to require more equivalence queries.
In \cref{figure:experiment_results:queries_vs_locations:equivalence}, we observe that for both \ourTool{} and \DOTA{}, the number of equivalence queries increases nearly linearly to the number of locations.
This also coincides with the complexity analysis in \cref{section:complexity_analysis}.
\cref{figure:experiment_results:queries_vs_locations:equivalence} also shows that the number of equivalence queries increases faster in \DOTA{} than in \ourTool{}.

\cref{table:experiment_results:unbalanced} also suggests a similar tendency: the number of membership queries rapidly increases to the complexity of the timing constraints; In contrast, the number of equivalence queries increases rather slowly.
Moreover, \ourTool{} is scalable enough to learn a DTA with five clock variables within 15 minutes.

\cref{table:experiment_results:random} also suggests that \ourTool{} does not scale well to the maximum constant in the guards, as observed\LongVersion{ in the analysis} in \cref{section:complexity_analysis}.
However, we still observe that \ourTool{} requires fewer equivalence queries than \DOTA{}.
Overall, compared with \DOTA{}, \ourTool{} has better scalability in the number of equivalence queries and worse scalability in the number of membership queries.

%%%%%%%%%%%%%%%%%%%%%%%%%%%%%%%%%%%%%%%%%%%%%%%%%%%%%%%%%%%%
\subsection{RQ2: Performance on practical benchmarks}\label{section:practicality}
%%%%%%%%%%%%%%%%%%%%%%%%%%%%%%%%%%%%%%%%%%%%%%%%%%%%%%%%%%%%

%-----------------------------------------------------------
\begin{table}[tbp]
 \caption{Summary of the target DTA and the results for practical benchmarks. The columns are the same as \cref{table:experiment_results:unbalanced}. Cells with the best results are highlighted.}%
 \label{table:experiment_results:practical}
 \footnotesize
 \centering
 \begin{tabular}{llrrrrrrr}
\toprule
 &  & $|\Loc|$ & $|\Alphabet|$ & $|\Clock|$ & $K_{\Clock}$ & \# of Mem.\ queries & \# of Eq.\ queries & Exec.\ time [sec.] \\
\midrule
\multirow[c]{2}{*}{\AKM{}} & \ourTool{} & 17 & 12 & 1 & 5 & 12,263 & \tbcolor{}11 & \tbcolor{}5.85e-01 \\
 & \DOTA{} & 17 & 12 & 1 & 5 & \tbcolor{}3,453 & 49 & 7.97e+00 \\
\multirow[c]{2}{*}{\CAS{}} & \ourTool{} & 14 & 10 & 1 & 27 & 66,067 & \tbcolor{}17 & \tbcolor{}4.65e+00 \\
 & \DOTA{} & 14 & 10 & 1 & 27 & \tbcolor{}4,769 & 18 & 9.58e+01 \\
\multirow[c]{2}{*}{\light{}} & \ourTool{} & 5 & 5 & 1 & 10 & 3,057 & \tbcolor{}7 & \tbcolor{}3.30e-02 \\
 & \DOTA{} & 5 & 5 & 1 & 10 & \tbcolor{}210 & 7 & 9.32e-01 \\
\multirow[c]{2}{*}{\PC{}} & \ourTool{} & 26 & 17 & 1 & 10 & 245,134 & \tbcolor{}23 & \tbcolor{}6.49e+01 \\
 & \DOTA{} & 26 & 17 & 1 & 10 & \tbcolor{}10,390 & 29 & 1.24e+02 \\
\multirow[c]{2}{*}{\TCP{}} & \ourTool{} & 22 & 13 & 1 & 2 & 11,300 & \tbcolor{}15 & \tbcolor{}3.82e-01 \\
 & \DOTA{} & 22 & 13 & 1 & 2 & \tbcolor{}4,713 & 32 & 2.20e+01 \\
\multirow[c]{2}{*}{\Train{}} & \ourTool{} & 6 & 6 & 1 & 10 & 13,487 & \tbcolor{}8 & \tbcolor{}1.72e-01 \\
 & \DOTA{} & 6 & 6 & 1 & 10 & \tbcolor{}838 & 13 & 1.13e+00 \\
\FDDI{} & \ourTool{} & 16 & 5 & 7 & 6 & \tbcolor{}9,986,271 & \tbcolor{}43 & \tbcolor{}3.00e+03 \\
\bottomrule
\end{tabular}

\end{table}
%-----------------------------------------------------------
To evaluate the practicality of \ourTool{}, we used seven benchmarks: \AKM{}, \CAS{}, \light{}, \PC{}, \TCP{}, \Train{}, and \FDDI{}. 
\cref{table:experiment_results:practical} summarizes their complexity.
All the benchmarks other than \FDDI{} are taken from~\cite{XAZ22} (or its implementation~\cite{Leslieaj/DOTALearningSMT}).
\FDDI{} is taken from TChecker~\cite{ticktac-project/tchecker}.
We use the instance of \FDDI{} with two processes.
% a benchmark modeling the FDDI protocol.
% We took a model from TChecker~\cite{ticktac-project/tchecker} with two processes.

\cref{table:experiment_results:practical} summarizes the results for the benchmarks from practical applications.
We observe, again, that \ourTool{} requires more membership queries and fewer equivalence queries than \DOTA{}.
However, for these benchmarks, the difference in the number of membership queries tends to be much smaller than in \Random{}.
This is because these benchmarks have simpler timing constraints than \Random{} for the exploration by \ourTool{}.
In \AKM{}, \light{}, \PC{}, \TCP{}, and \Train{}, the clock variable can be reset at every edge without changing the language\LongVersion{, \ie{} there are \emph{real-time automata}~\cite{DBLP:journals/jalc/Dima01} recognizing the same language}.
For such a DTA, all simple elementary languages are equivalent in terms of the Nerode-style congruence\LongVersion{ in \cref{section:timed_distinguishing_suffixes}} if we have the same edge at their last event and the same dwell time after it.
If two simple elementary languages are equivalent, \ourTool{} explores the successors of only one of them, and the exploration is relatively efficient.
We have a similar situation in \CAS{}.
Moreover, in many of these DTAs, only a few edges have guards.
Overall, despite the large number of locations and alphabets, these languages' complexities are mild for \ourTool{}.

We also observe that, surprisingly, for all of these benchmarks, \ourTool{} took a shorter time for DTA learning than \DOTA{}.
This is partly because of the difference in the implementation language (\ie{} C++ vs. Python) but also because of the small number of equivalence queries and the mild number of membership queries.
Moreover, although it requires significantly more queries, \ourTool{} successfully learned \FDDI{} with seven clock variables.
Overall, such efficiency on benchmarks from practical applications suggests the potential usefulness of \ourTool{} in some realistic scenarios. %, especially with frequent reset of clock variables.
\section{Conclusions and future work}\label{section:conclusion}
%%%%%%%%%%%%%%%%%%%%%%%%%%%%%%%%%%%%%%%%%%%%%%%%%%%%%%%%%%%%
%%%%%%%%%%%%%%%%%%%%%%%%%%%%%%%%%%%%%%%%%%%%%%%%%%%%%%%%%%%%

Extending the L* algorithm, we proposed an active learning algorithm for DTAs.
Our extension is by our Nerode-style congruence for recognizable timed languages.
We proved the termination and the correctness of our algorithm.
We also proved that our learning algorithm requires a polynomial number of queries with a smart teacher and an exponential number of queries with a normal teacher.
Our experiment results also suggest the practical relevance of our algorithm.

One of the future directions is to extend more recent automata learning algorithms (\eg{} TTT algorithm~\cite{IHS14} to improve the efficiency) to DTA learning.
Another direction is constructing a \emph{passive} DTA learning algorithm based on our congruence and an existing passive DFA learning algorithm.
It is also a future direction to apply our learning algorithm for practical usage, \eg{} identification of black-box systems and testing black-box systems with black-box checking~\cite{PVY99,MP19,Waga20}.
% In our algorithm, we exploited the simplicity and canonicity of the timed conditions $\timedCondition$, which makes the size of the observation table huge.
% This is much like the practical inefficiency of the \emph{region} construction for reachability analysis of TAs.
% One of the future directions is to use \emph{timed decision trees}~\cite{GJP06} for a more efficient algorithm.
% This potential performance improvement is much like the \emph{zone} construction for the reachability analysis of TAs.
Optimization of the algorithm, \eg{} by incorporating clock information is also a future direction.
% For example, the candidate renaming equations might be restricted by examining the minimum clock size to represent the hypothesis timing constraint, which potentially reduces the number of equivalence queries.

%%%%%%%%%%%%%%%%%%%%%%%%%%%%%%%%%%%%%%%%%%%%%%%%%%%%%%%%%%%%
\subsubsection*{Acknowledgements.}
%%%%%%%%%%%%%%%%%%%%%%%%%%%%%%%%%%%%%%%%%%%%%%%%%%%%%%%%%%%%

This work is partially supported
    by
    JST ACT-X Grant No.\ JPMJAX200U,
    JST PRESTO Grant No.\ JPMJPR22CA,
    JST CREST Grant No.\ JPMJCR2012, and
    JSPS KAKENHI Grant No.\ 22K17873.

%%
%% Bibliography
%%

%% Please use bibtex, 
\newpage
\bibliographystyle{splncs04}
\bibliography{dblp_refs}
\ifdefined\VersionLong
\newpage
\appendix

%%%%%%%%%%%%%%%%%%%%%%%%%%%%%%%%%%%%%%%%%%%%%%%%%%%%%%%%%%%%
%%%%%%%%%%%%%%%%%%%%%%%%%%%%%%%%%%%%%%%%%%%%%%%%%%%%%%%%%%%%
\section{Omitted proofs}
%%%%%%%%%%%%%%%%%%%%%%%%%%%%%%%%%%%%%%%%%%%%%%%%%%%%%%%%%%%%
%%%%%%%%%%%%%%%%%%%%%%%%%%%%%%%%%%%%%%%%%%%%%%%%%%%%%%%%%%%%

%%%%%%%%%%%%%%%%%%%%%%%%%%%%%%%%%%%%%%%%%%%%%%%%%%%%%%%%%%%%
\subsection{Proof of \cref{theorem:adequacy_equivalence}}
%%%%%%%%%%%%%%%%%%%%%%%%%%%%%%%%%%%%%%%%%%%%%%%%%%%%%%%%%%%%

Here, we prove \cref{theorem:adequacy_equivalence}. First, we show the following lemma.

\begin{lemma}%
 \label{lemma:membership_renaming}
 Let $\Lg \subseteq \TimedWords$ be a timed language and let $\elementary, \elementary['], \elementary[''] \in \Elementary$ be elementary languages.
 For any timed words $\word \in \elementary$, $\word' \in \elementary[']$, and $\word'' \in \elementary['']$,
 if we have
 $\rename{\symbolicMemQ{\Lg}(\elementary \cdot \elementary[''])}{\Rename} \land \timedCondition' \Rightarrow \symbolicMemQ{\Lg}(\elementary['] \cdot \elementary[''])$,
 $\timeValuation{\word}, \timeValuation{\word'} \models \Rename$, and
 $\word \cdot \word'' \in \Lg$,
 we also have $\word' \cdot \word'' \in \Lg$.
\end{lemma}

\begin{proof}
 Since we have $\word \cdot \word'' \in \Lg$, $\timeValuation{\word}, \timeValuation{\word''} \models \symbolicMemQ{\Lg}(\elementary \cdot \elementary[''])$ holds.
 % Therefore, for any renaming equation $\Rename$,
 % we have $\rename{\timeValuation{\word} \land \timeValuation{\word''}}{\Rename} \imply \rename{\symbolicMemQ{\Lg}(\elementary \cdot \elementary[''])}{\Rename}$.
 % Since $\Rename$ is a renaming equation over $\timeVariables$ and $\timeVariables'$,
 % we have $\rename{\timeValuation{\word}}{\Rename} \land \timeValuation{\word''} \imply \rename{\symbolicMemQ{\Lg}(\elementary \cdot \elementary[''])}{\Rename}$.
 Since we have $\timeValuation{\word}, \timeValuation{\word'} \models \Rename$,
 we have $\timeValuation{\word}, \timeValuation{\word'}, \timeValuation{\word''} \models \rename{\symbolicMemQ{\Lg}(\elementary \cdot \elementary[''])}{\Rename}$.
 Since $\word' \in \elementary[']$, we have $\timeValuation{\word'} \models \timedCondition'$, and
 we have $\timeValuation{\word}, \timeValuation{\word'}, \timeValuation{\word''} \models \rename{\symbolicMemQ{\Lg}(\elementary \cdot \elementary[''])}{\Rename} \land \timedCondition'$.
 Because of $\rename{\symbolicMemQ{\Lg}(\elementary \cdot \elementary[''])}{\Rename} \land \timedCondition' \imply \symbolicMemQ{\Lg}(\elementary['] \cdot \elementary[''])$,
 we have
 $\timeValuation{\word}, \timeValuation{\word'}, \timeValuation{\word''} \models \symbolicMemQ{\Lg}(\elementary['] \cdot \elementary[''])$.
 Since $\timeValuation{\word}$ is over $\timeVariables$ and
 $\symbolicMemQ{\Lg}(\elementary['] \cdot \elementary[''])$ is over $\timeVariables'$ and $\timeVariables''$,
 %
 % we have $\rename{(\rename{\timeValuation{\word'}}{\Rename})}{\Rename} \land \timeValuation{\word''} \imply \rename{\symbolicMemQ{\Lg}(\elementary \cdot \elementary[''])}{\Rename}$.
 % By the definition of renaming,
 % $\timeValuation{\word'} \imply \rename{(\rename{\timeValuation{\word'}}{\Rename})}{\Rename}$ holds, and thus, we have
 %
 % $\timeValuation{\word'} \land \timeValuation{\word''} \imply \rename{\symbolicMemQ{\Lg}(\elementary \cdot \elementary[''])}{\Rename}$.
 % Since $\word' \in \elementary[']$, we have $\timeValuation{\word'} \imply \timedCondition'$.
 % Therefore, we have
 % $\timeValuation{\word'} \land \timeValuation{\word''} \imply \rename{\symbolicMemQ{\Lg}(\elementary \cdot \elementary[''])}{\Rename} \land \timedCondition'$.
 % Since $\rename{\symbolicMemQ{\Lg}(\elementary \cdot \elementary[''])}{\Rename} \land \timedCondition' \imply \symbolicMemQ{\Lg}(\elementary['] \cdot \elementary[''])$ holds,
 we have
 $\timeValuation{\word'}, \timeValuation{\word''} \models \symbolicMemQ{\Lg}(\elementary['] \cdot \elementary[''])$.
 Therefore, we have $\word' \cdot \word'' \in \Lg$
\end{proof}

The following proves \cref{theorem:adequacy_equivalence}.

\recallResult{theorem:adequacy_equivalence}{\adequacyEquivalenceStatement}

%-----------------------------------------------------------
\begin{proof}
 % We prove the first condition of \cref{theorem:adequacy_equivalence}.
 % The proof of the second condition is similar thanks to the symmetricity.
 Let $\Rename$ be the renaming equation satisfying $\elementary \CellSub^{\elementary[''], \Rename}_{\Lg} \elementary[']$.
 By the definition of $\elementary \CellSub^{\elementary[''], \Rename}_{\Lg} \elementary[']$,
 % for any $\elementary[''] \in \SuffixSet$, 
 we have
 $\rename{\symbolicMemQ{\Lg}(\elementary \cdot \elementary[''])}{\Rename} \land \timedCondition' $ if and only if $\symbolicMemQ{\Lg}(\elementary['] \cdot \elementary['']) \land \Rename \land \timedCondition$.
 By \cref{lemma:membership_renaming},
 % for any $\elementary[''] \in \SuffixSet$,
 for any $\word \in \elementary$, $\word' \in \elementary[']$, and $\word'' \in \elementary['']$
 satisfying $\timeValuation{\word}, \timeValuation{\word'} \models \Rename$,
 we have $\word \cdot \word'' \in \Lg$ if and only if $\word' \cdot \word'' \in \Lg$.
 By the definition of $\elementary \CellSub^{\elementary[''], \Rename}_{\Lg} \elementary[']$,
 % we have $\timedCondition \Rightarrow \rename{\timedCondition[']}{\Rename}$.
 % Therefore,
 for any $\word \in \elementary$, there is $\word' \in \elementary[']$ satisfying $\timeValuation{\word}, \timeValuation{\word} \models \Rename$.
 % For such $\word$ and $\word'$, we also have
 % $\rename{(\rename{\timeValuation{\word}}{\Rename})}{\Rename} = \rename{\timeValuation{\word'}}{\Rename}$.
 Therefore, for any $\word \in \elementary$, there is $\word' \in \elementary[']$ such that for any $\word'' \in \elementary['']$, we have $\word \cdot \word'' \in \Lg \iff \word' \cdot \word'' \in \Lg$.
 %
 % Similarly,
 % by the definition of $\elementary \CellSim^{\SuffixSet, \Rename}_{\Lg} \elementary[']$,
 % for any $\elementary[''] \in \SuffixSet$, we have
 % $\rename{\symbolicMemQ{\Lg}(\elementary['] \cdot \elementary[''])}{\Rename} \land \timedCondition \Rightarrow \symbolicMemQ{\Lg}(\elementary \cdot \elementary[''])$.
 % By \cref{lemma:membership_renaming},
 % for any $\elementary[''] \in \SuffixSet$,
 % for any $\word \in \elementary$, $\word' \in \elementary[']$, and $\word'' \in \elementary['']$,
 % if we have $\rename{\timeValuation{\word'}}{\Rename} = \rename{(\rename{\timeValuation{\word}}{\Rename})}{\Rename}$ and
 % $\word' \cdot \word'' \in \Lg$, we have $\word \cdot \word'' \in \Lg$.
 %
 % Moreover,
 % for any $\word \in \elementary$ and $\word' \in \elementary[']$,
 % we have $\rename{(\rename{\timeValuation{\word}}{\Rename})}{\Rename} = \rename{\timeValuation{\word'}}{\Rename}$ if and only if
 % we have $\rename{(\rename{(\rename{\timeValuation{\word}}{\Rename})}{\Rename})}{\Rename} = \rename{\timeValuation{\word}}{\Rename} = \rename{(\rename{\timeValuation{\word'}}{\Rename})}{\Rename}$.
 % Therefore, for any $\word \in \elementary$, there is $\word' \in \elementary[']$ such that for any $\elementary[''] \in \SuffixSet$ and for any $\word'' \in \elementary['']$, we have $\word \cdot \word'' \in \Lg \iff \word' \cdot \word'' \in \Lg$.
\end{proof}
\subsection{Proof of \cref{theorem:finiteness}}\label{section:proof:theorem:finiteness}
%%%%%%%%%%%%%%%%%%%%%%%%%%%%%%%%%%%%%%%%%%%%%%%%%%%%%%%%%%%%

\cref{theorem:finiteness} is immediately obtained from the following lemmas.

%-----------------------------------------------------------
\newcommand{\TimedRecognizableToFinite}{%
 For any recognizable timed language $\Lg \subseteq \TimedWords$, the quotient set $\SimpleElementary / {\CellSim^{\Elementary}_{\Lg}}$ is finite.
}
\begin{lemma}%
 \label{lemma:timed_recognizable_to_finite}
 \TimedRecognizableToFinite{}\qed{}
\end{lemma}
\begin{lemma}%
 \label{lemma:finite_to_timed_recognizable}
 For any timed language $\Lg \subseteq \TimedWords$, if the quotient set $\SimpleElementary / {\CellSim^{\Elementary}_{\Lg}}$ is finite, $\Lg$ is recognizable.
\end{lemma}
%-----------------------------------------------------------
\begin{proof}
 [sketch]
 Let $\PrefixSet \subseteq \SimpleElementary$ be a finite set of simple elementary languages representing $\SimpleElementary / {\CellSim^{\Elementary}_{\Lg}}$, \ie{} for any $\prefix \in \SimpleElementary$, there is $\prefix' \in \PrefixSet$ satisfying $\prefix \CellSim^{\Elementary}_{\Lg} \prefix'$.
 We construct $\tilde{\PrefixSet}$ by augmenting $\PrefixSet$ as follows:
 1) for each $\elementary \in \PrefixSet$, we add all its prefixes to $\tilde{\PrefixSet}$ so that $\tilde{\PrefixSet}$ is prefix-closed;
 2) for each $\elementary \in \PrefixSet$ such that $\timedCondition$ is a simple and canonical timed condition if $\timedCondition$ does not contain equalities and $\PrefixSet$ does not contain any of its time successors $\elementary[']$, \ie{}
 $\elementary['] \neq \elementary$ such that
 there are $\word \in \elementary$ and $t \in \Rp$ satisfying $\word \cdot t \in \elementary[']$,
 we add such $\elementary[']$ with the shortest dwell time to $\tilde{\PrefixSet}$.
 We note that $\tilde{\PrefixSet}$ is also a finite set.
 Let $\RecognizableFinal = \{\elementary \in \tilde{\PrefixSet} \mid \elementary \subseteq \Lg \}$.
 Let $\CRM = \{(\elementaryInside, \elementaryInside['], \Rename) \mid \elementary \in \exterior{\tilde{\PrefixSet}}, \elementary['] \in \tilde{\PrefixSet}, \elementary \CellSim^{\Elementary, \Rename}_{\Lg} \elementary[']\}$.
 % $\Rename = \bigwedge_{(\timestamp_{i}, \timestamp'_{j}) \in \CellRel} \sumTimestamp{i}{|u|} = \sumTimestamp[']{j}{|u'|}$, and
 % $(\timestamp_{i}, \timestamp'_{j}) \in \CellRel$ shows such that $\CellRel(\timestamp_{i}) = \timestamp'_{j}$.
 Then, $\Lg$ is equal to the recognizable language defined by $\tilde{\PrefixSet}$, $\RecognizableFinal$, and $\CRM$.
\end{proof}
%-----------------------------------------------------------

%% 2022-11-20: We moved the proof to the main text.
%% 2023-01-28: We moved the proof to the appendix again.
%%%%%%%%%%%%%%%%%%%%%%%%%%%%%%%%%%%%%%%%%%%%%%%%%%%%%%%%%%%%
\subsubsection{Proof of \cref{lemma:timed_recognizable_to_finite}}\label{section:proof:lemma:timed_recognizable_to_finite}
%%%%%%%%%%%%%%%%%%%%%%%%%%%%%%%%%%%%%%%%%%%%%%%%%%%%%%%%%%%%

First, we review \emph{regions} and \emph{region automata}~\cite{AD94}.
\DefinitionRegions{}

For a timed word $\word$ and concrete states $\TTSstate, \TTSstate' \in \TTSState$, we write $\TTSstate \xrightarrow{\word} \TTSstate'$ if there is a run from $\TTSstate$ to $\TTSstate'$ associated with $\word$.

The proof of \cref{lemma:timed_recognizable_to_finite} is as follows.

\recallResult{lemma:timed_recognizable_to_finite}{\TimedRecognizableToFinite}

\begin{proof}
 [sketch]
 Let $\TAWithInside$ be a DTA recognizing $\Lg$, let $\TTSWithInside$ be the TTS of $\A$, let $\regionAutomWithInside$ be the RA of $\A$, and
 for each $\clock \in \Clock$, let $\maxConstant_{\clock}$ be the maximum constant compared with $\clock$ in $\A$.
 Without loss of generality, we can assume that $\A$ is complete, \ie{} for any $\TTSstate \in \TTSState$ and $\action \in \Alphabet$, there is $\TTSstate' \in \TTSState$ satisfying $\TTSstate \TTStransitionRelWithLabel{\action} \TTSstate'$.
 % Moreover, since $\A$ is a DTA, $\InitRegionState$ is a singleton $\InitRegionState = \{\initRegionState\}$.

 For a simple elementary language $\elementary$,
 we denote $\initRegionState \xrightarrow{\elementary} \regionState$ for $\regionState \in \RegionState$ such that some $\TTSstate \in \regionState$ is reachable by some $\word \in \elementary$.
 For any $\regionState \in \RegionState$ and
 for any simple elementary languages $\elementary$ and $\elementary[']$ satisfying
 $\initRegionState \xrightarrow{\elementary} \regionState$ and $\initRegionState \xrightarrow{\elementary[']} \regionState$,
 let $\Rename$ be the renaming equation that contains
 $\sumTimestamp{i}{|u|} = \sumTimestamp[']{j}{|u'|}$
 if % and only if
 there is a clock $\clock \in \Clock$ whose value in $\regionState$ is equal to $\sumTimestamp{i}{|u|}$ and $\sumTimestamp[']{j}{|u'|}$, and
 the value of $\clock$ is smaller than or equal to $K_{\clock}$.
 For such a renaming equation $\Rename$,
 we have $\elementary \CellSim^{\Elementary, \Rename}_{\Lg} \elementary[']$.
 Therefore, for any such $\elementary$ and $\elementary[']$, $\elementary \CellSim^{\Elementary}_{\Lg} \elementary[']$ holds, and thus,
 for any $\elementary,\elementary[']$, if they load to the same state of the RA,
 we have $\elementary \CellSim^{\Elementary}_{\Lg} \elementary[']$.
 Since the state space of the RA is finite,
 we have $|\SimpleElementary / {\CellSim^{\Elementary}_{\Lg}}| \leq |\RegionState|$.
\end{proof}

\subsection{Proof of \cref{theorem:finite_suffix}}

The proof of \cref{theorem:finite_suffix} is as follows.

\recallResult{theorem:finite_suffix}{\FiniteSuffix}

%-----------------------------------------------------------
\begin{proof}
 By the definition of $\CellSim^{\SuffixSet}_{\Lg}$,
 for any $\SuffixSet \subseteq \SuffixSet'$,
 ${\CellSim^{\SuffixSet'}_{\Lg}}$ is finer than ${\CellSim^{\SuffixSet}_{\Lg}}$, and
 we have $|\SimpleElementary / {\CellSim^{\SuffixSet}_{\Lg}}| \leq |\SimpleElementary / {\CellSim^{\SuffixSet'}_{\Lg}}|$.
 Let $\SuffixSet \subseteq \Elementary$ be such that
 $|\SimpleElementary / {\CellSim^{\SuffixSet}_{\Lg}}| < |\SimpleElementary / {\CellSim^{\Elementary}_{\Lg}}|$.
 Since  $|\SimpleElementary / {\CellSim^{\SuffixSet}_{\Lg}}| < |\SimpleElementary / {\CellSim^{\Elementary}_{\Lg}}|$,
 there are $\elementary,\elementary['] \in \SimpleElementary$ satisfying
 $\elementary \CellSim^{\SuffixSet}_{\Lg} \elementary[']$ and $\elementary \NCellSim^{\Elementary}_{\Lg} \elementary[']$.
 Let $\timeVariables$ and $\timeVariables'$ be the domain of $\timedCondition$ and $\timedCondition'$, respectively.
 By definition of $\elementary \NCellSim^{\Elementary}_{\Lg} \elementary[']$,
 for any renaming equation $\Rename$ over $\timeVariables$ and $\timeVariables'$ satisfying
 $\forall \word \in \elementary.\, \exists \word' \in \elementary['].\, \timeValuation{\word},\timeValuation{\word'} \models \Rename$,
 there is $\suffix \in \Elementary \setminus \SuffixSet$ satisfying
 $\rename{\symbolicMemQ{\Lg}(\elementary \cdot \suffix)}{\Rename} \land \timedCondition' \centernot\iff \symbolicMemQ{\Lg}(\elementary['] \cdot \suffix) \land \Rename \land \timedCondition$.
 Therefore, for each $\Rename$, by adding such $\suffix$ to $\SuffixSet$, we obtain $\SuffixSet' \supsetneq \SuffixSet$ satisfying
 $|\SimpleElementary / {\CellSim^{\SuffixSet}_{\Lg}}| < |\SimpleElementary / {\CellSim^{\SuffixSet'}_{\Lg}}|$.
 Since there are only finitely many renaming equations $\Rename$ over $\timeVariables$ and $\timeVariables'$,
 such refinement increase $\SuffixSet$ only finitely.

 Let a sequence $\SuffixSet_0 \subsetneq \SuffixSet_1 \subsetneq \dots$ of elementary languages
 such that $\SuffixSet_0 = \emptyset$ and each $\SuffixSet_i$ is obtained by the above refinement of $\SuffixSet_{i-1}$.
 By \cref{lemma:timed_recognizable_to_finite}, $|\SimpleElementary / {\CellSim^{\Elementary}_{\Lg}}|$ is finite, and
 such a sequence is finite.
 Therefore, there is $n \in \N$ such that $\SimpleElementary / {\CellSim^{\SuffixSet_n}_{\Lg}} = \SimpleElementary / {\CellSim^{\Elementary}_{\Lg}}$.
 Moreover, by the construction of $\SuffixSet_i$, each $\SuffixSet_i$ is finite, and thus, $\SuffixSet_n$ is a finite set of elementary languages satisfying ${\CellSim^{\Elementary}_{\Lg}} = {\CellSim^{\SuffixSet_n}_{\Lg}}$.
\end{proof}
%-----------------------------------------------------------

% %%%%%%%%%%%%%%%%%%%%%%%%%%%%%%%%%%%%%%%%%%%%%%%%%%%%%%%%%%%%
% \subsection{Formal definition of predecessors}\label{section:definition:predecessor}
% %%%%%%%%%%%%%%%%%%%%%%%%%%%%%%%%%%%%%%%%%%%%%%%%%%%%%%%%%%%%

% The formal definition of predecessor is as follows.

% \DefinitionPredecessor{}

%%%%%%%%%%%%%%%%%%%%%%%%%%%%%%%%%%%%%%%%%%%%%%%%%%%%%%%%%%%%
\subsection{Uniqueness of the fractional part}
%%%%%%%%%%%%%%%%%%%%%%%%%%%%%%%%%%%%%%%%%%%%%%%%%%%%%%%%%%%%

Here, we show the following property used in \cref{section:fractional_elementary}.

%-----------------------------------------------------------
\begin{proposition}
 \label{proposition:order_uniqueness}
 For any simple elementary language $\elementary \in \SimpleElementary$ and
 for any timed words
 $\word, \word' \in \elementary$ such that
 $\wordWithInside$ and $\word' = \timestamp'_0 \action_1 \timestamp'_1 \action'_2 \dots \action_{n} \timestamp'_{n}$,
 for any $i,j \in \{0,1,\dots,n\}$, we have the following, where
 $\sumTimestamp{i}{n} = \sum_{{k = i}}^{n} \timestamp_k$ and
 $\sumTimestamp[']{i}{n} = \sum_{{k = i}}^{n} \timestamp'_k$.
 \begin{itemize}
  \item $\fractional(\sumTimestamp{i}{n}) = 0 \iff \fractional(\sumTimestamp[']{i}{n}) = 0$
  \item $\fractional(\sumTimestamp{i}{n}) \leq \fractional(\sumTimestamp{j}{n}) \iff \fractional(\sumTimestamp[']{i}{n}) \leq \fractional(\sumTimestamp[']{j}{n})$
 \end{itemize}
\end{proposition}
%-----------------------------------------------------------
\begin{proof}
 If we have $\fractional(\sumTimestamp{i}{n}) = 0$,
 since $\timedCondition$ is simple and canonical,
 $\timedCondition$ implies $\sumTimestamp{i}{n} = \dConstant$ for some $\dConstant \in \N$.
 Since $\word' \in \elementary$, we also have $\fractional(\sumTimestamp[']{i}{n}) = 0$.

 Assume we have $\fractional(\sumTimestamp{i}{n}) \leq \fractional(\sumTimestamp{j}{n})$.
 If we have $\fractional(\sumTimestamp{i}{n}) = 0$,
 by the first part of the property, 
 we have $\fractional(\sumTimestamp[']{i}{n}) = 0$ and
 $\fractional(\sumTimestamp[']{i}{n}) \leq \fractional(\sumTimestamp[']{j}{n})$ holds.
 If we have $\fractional(\sumTimestamp{j}{n}) = 0$,
 by $\fractional(\sumTimestamp{i}{n}) \leq \fractional(\sumTimestamp{j}{n})$, we have 
 $\fractional(\sumTimestamp{i}{n}) = 0$, and 
 $\fractional(\sumTimestamp[']{i}{n}) \leq \fractional(\sumTimestamp[']{j}{n})$ holds.
 If we have $\fractional(\sumTimestamp{i}{j}) = 0$,
 since $\timedCondition$ is simple and canonical, 
 we also have $\fractional(\sumTimestamp[']{i}{j}) = 0$.
 Therefore, we have 
 $\fractional(\sumTimestamp{i}{n}) = \fractional(\sumTimestamp{j}{n})$ and
 $\fractional(\sumTimestamp[']{i}{n}) = \fractional(\sumTimestamp[']{j}{n})$.
 If none of $\fractional(\sumTimestamp{i}{n})$, $\fractional(\sumTimestamp{j}{n})$, and, $\fractional(\sumTimestamp{i}{j})$ is zero,
 let $a,b,c \in \N$ be such that $\sumTimestamp{i}{n} \in (a, a + 1)$,
 $\sumTimestamp{j}{n} \in (b, b + 1)$, and
 $\sumTimestamp{i}{j} \in (c, c + 1)$.
 Clearly, we have either $c = a - b$ or $c = a - b - 1$.
 If we have $c = a - b$, we have
 $\fractional(\sumTimestamp{i}{n}) < \fractional(\sumTimestamp{j}{n})$ and $\fractional(\sumTimestamp[']{i}{n}) < \fractional(\sumTimestamp[']{j}{n})$.
 Otherwise, we have
 $\fractional(\sumTimestamp{i}{n}) > \fractional(\sumTimestamp{j}{n})$ and $\fractional(\sumTimestamp[']{i}{n}) > \fractional(\sumTimestamp[']{j}{n})$.
 Overall, we have $\fractional(\sumTimestamp{i}{n}) \leq \fractional(\sumTimestamp{j}{n}) \iff \fractional(\sumTimestamp[']{i}{n}) \leq \fractional(\sumTimestamp[']{j}{n})$.
 \qed{}
\end{proof}
%-----------------------------------------------------------

%%%%%%%%%%%%%%%%%%%%%%%%%%%%%%%%%%%%%%%%%%%%%%%%%%%%%%%%%%%%
\subsection{Proof of \cref{theorem:row_faithfulness}}
%%%%%%%%%%%%%%%%%%%%%%%%%%%%%%%%%%%%%%%%%%%%%%%%%%%%%%%%%%%%

First, we show the following lemma on successors.

%-----------------------------------------------------------
\begin{lemma}%
 \label{lemma:successor-jump}
 Let $\TimedObsTableInside$ be a closed, consistent, and exterior-consistent timed observation table.
 For any $\prefix \in \successor{\PrefixSet}$,
 there is $\prefix' \in \PrefixSet$ such that
 $\prefix \CellSim^{\SuffixSet}_{\targetLg} \prefix'$ and
 for any $\word \in \prefix$,
 there is $\word' \in \prefix'$ satisfying $(\word, \word') \in \sem{\CRM}$, where $\CRM$ is the chronometric relational morphism constructed in \cref{algorithm:DTA_construction}.
\end{lemma}
%-%-%-%-%-%-%-%-%-%-%-%-%-%-%-%-%-%-%-%-%-%-%-%-%-%-%-%-%-%-
\begin{proof}
 If $\prefix \in \PrefixSet$, for any $\word \in \prefix$, we have $(\word, \word) \in \sem{\CRM}$.
 Assume $\prefix \notin \PrefixSet$.
 Let $\tilde{\prefix} \in \PrefixSet$ satisfying $\prefix \in \successor{\tilde{\prefix}}$.
 Since the timed observation table is exterior-consistent,
 there is $\elementary[^{\mathrm{ext}}] \in \exterior{\tilde{\prefix}}$ satisfying
 $\prefix \subseteq \elementary[^{\mathrm{ext}}]$.
 By closedness of the timed observation table and by the construction of $\CRM$,
 there are $\prefix' = \fractionalElementary['] \in \PrefixSet$ and a renaming equation $\CellRel$ satisfying
 $\prefix \CellSim^{\SuffixSet,\CellRel}_{\targetLg} \prefix'$,
 and
 $(\elementaryInside[^{\mathrm{ext}}], \elementaryInside['], \Rename) \in \CRM$. %, where
 % $\Rename = \bigwedge_{(\timestamp^{\mathrm{ext}}_{i}, \timestamp'_{j}) \in \CellRel} \sumTimestamp[^{\mathrm{ext}}]{i}{|u^{\mathrm{ext}}|} = \sumTimestamp[']{j}{|u'|}$.
 By $\prefix \CellSim^{\SuffixSet,\CellRel}_{\targetLg} \prefix'$,
 % we have $\timedCondition^{\mathrm{ext}} \Rightarrow \rename{\timedCondition[']}{\Rename}$.
 % Therefore,
 for any $\word \in \prefix$,
 there is $\word' \in \elementary[']$ satisfying $\timeValuation{\word}, \timeValuation{\word'} \models \Rename$.
 Therefore, we have $(\word, \word') \in \sem{\CRM}$.
\end{proof}
%-----------------------------------------------------------

The proof of \cref{theorem:row_faithfulness} is as follows.

\recallResult{theorem:row_faithfulness}{\rowFaithfulnessStatement}

% We prove \cref{theorem:faithfulness} using the following special case (when $\suffix = (\emptyword, \timestamp_0 = 0, \fractional(\timestamp_0) = 0)$) as a lemma.

%-----------------------------------------------------------
% \begin{lemma}
%  [row faithfulness]%
%  \label{lemma:p-faithfulness}
%  Let $\TimedObsTableInside$ be a closed, consistent, and exterior-consistent timed observation table.
%  For any $\prefix = \fractionalElementary \in \PrefixSet \cup \PrefixSet'$,
%  we have $\targetLg \cap \prefix = \Lg(\hypothesisA) \cap \prefix$, where
%  $\hypothesisA = \FConstructDTA(\TimedObsTableInsideInside)$.
% \end{lemma}
%-%-%-%-%-%-%-%-%-%-%-%-%-%-%-%-%-%-%-%-%-%-%-%-%-%-%-%-%-%-
\begin{proof}
 Let $(\PrefixSet,\RecognizableFinal,\CRM)$ be the recognizable timed language we construct in \cref{algorithm:DTA_construction}.
 Since $\prefix$ is a simple elementary language,
 for any recognizable timed language $\Lg$,
 we have either $\prefix \subseteq \Lg$ or $\prefix \cap \Lg = \emptyset$.
 % Therefore,
 % for any $\word \in \elementary$,
 % we have  $\timestampSequence \models \TableCell{\prefix}{\varepsilon}$ if and only if
 % $\TableCell{\prefix}{\emptyword} = \top$, where
 % $\wordWithInside$.

 If $\prefix \in \PrefixSet$,
 by the construction of $\RecognizableFinal$ in \cref{algorithm:DTA_construction},
 $\prefix \subseteq \RecognizableFinal$ holds if and only if 
 we have $\prefix \subseteq \targetLg$.
 By definition of recognizable timed languages,
 we have $\prefix \subseteq \Lg(\hypothesisA)$ if and only if $\prefix \subseteq \targetLg$.
 Therefore,
 we have $\targetLg \cap \prefix = \Lg(\hypothesisA) \cap \prefix$.

 If $\prefix \in \successor{\PrefixSet} \setminus \PrefixSet$,
 we let $\prefix' \in \PrefixSet$ be such that
 $\prefix \CellSim^{\SuffixSet}_{\targetLg} \prefix'$ and
 for any $\word \in \prefix$,
 there is $\word' \in \prefix'$ satisfying $(\word, \word') \in \sem{\CRM}$.
 The existence of such $\prefix'$ is guaranteed by \cref{lemma:successor-jump}.
 By the construction of $\hypothesisA$,
 for such $\word \in \prefix$ and $\word' \in \prefix'$,
 we have
 $\word \in \Lg(\hypothesisA) \iff \word' \in \Lg(\hypothesisA)$.
 Since $\prefix \CellSim^{\SuffixSet}_{\targetLg} \prefix'$,
 we have $\prefix \cap \targetLg = \emptyset \iff \prefix' \cap \targetLg = \emptyset$.
 Since $\prefix$ and $\prefix'$ are simple elementary languages, we have
 $\targetLg \cap \prefix = \Lg(\hypothesisA) \cap \prefix$.
 %
 % Assume $n > 2$.
 % Let $\word \in \prefix_2$ and $\word' \in \prefix'_1$ be such that $(\word, \word') \in \CRM$.
 % By definition of recognizable timed language,
 % for such $\word$ and $\word'$, and for any timed word $\word''$, we have
 % $\word \cdot \word'' \in \Lg(\hypothesisA) \iff \word' \cdot \word'' \in \Lg(\hypothesisA)$.
 %
 % \todo{Fill the gap here}
 % Moreover, there are $\word \in \prefix_2$, $\word' \in \prefix'_1$, and $\word'' \in \TimedWords$ satisfying
 % $\word \cdot \word'' \in \prefix_n$ and $\word' \cdot \word'' \in \prefix'_{n-1}$.
 %
 % Since $\prefix_n$ and $\prefix'_{n-1}$ are forward regional elementary languages,
 % % we have $\prefix_n \subseteq \Lg(\hypothesisA)$ if and only if we have $\prefix'_{n-1} \subseteq \Lg(\hypothesisA)$.
 % % Since $\prefix_n$ and $\prefix'_{n-1}$ are forward regional elementary languages, we have
 % $\prefix_n \cap \Lg(\hypothesisA) = \prefix'_{n-1} \cap \Lg(\hypothesisA)$.
 % By induction hypothesis, we have $\prefix'_{n-1} \cap \Lg(\hypothesisA) = \prefix'_{n-1} \cap \targetLg$.
 % Since we have $\prefix_n \CellSim^{\SuffixSet}_{\targetLg} \prefix'_{n-1}$, we have $\prefix_{n} \cap \targetLg = \prefix'_{n-1} \cap \targetLg$.
 % Overall, we have $\targetLg \cap \prefix = \Lg(\hypothesisA) \cap \prefix$.
\end{proof}
\subsection{Proof of \cref{lemma:correct_in_limit}}
%%%%%%%%%%%%%%%%%%%%%%%%%%%%%%%%%%%%%%%%%%%%%%%%%%%%%%%%%%%%

The proof of \cref{lemma:correct_in_limit} is as follows.

\recallResult{lemma:correct_in_limit}{\correctnessInLimitStatement{}}

%-----------------------------------------------------------
\begin{proof}
 Let $\word \in \TimedWords$ be a timed word.
 Let $\hypothesisA = \FConstructDTA{\TimedObsTableInsideInside}$.
 Let $\word_0, \word_1, \dots, \word_n$, $\word'_1, \word'_2, \dots, \word'_n$, $\word''_1, \word''_2, \dots, \word''_n$, $\overline{\word}_1, \overline{\word}_2, \dots, \overline{\word}_n$, be timed words,
 let $\prefix_1, \prefix_2,\dots, \prefix_n \in \PrefixSet$, and $\prefix'_1, \prefix'_2,\dots, \prefix'_n \in \successor{\PrefixSet}$ be simple elementary languages, and
 let $\Rename_1, \Rename_2, \dots, \Rename_n$ be renaming equations satisfying the following properties.
 We note that the above exists for some $n > 0$ because the timed observation table is closed.
 %-----------------------------------------------------------
 \begin{itemize}
 \item $\word_0 = \word$
 \item $\word''_n = \varepsilon$
 \item For each $i \in \{1, 2, \dots, n\}$, we have
        $\word'_i \in \prefix'_i$,
        $\overline{\word}_i \in \prefix_i$,
        $\word_{i-1} = \word'_i \cdot \word''_i$,
        $\word_i = \overline{\word}_i \cdot \word''_i$,
        $\prefix'_i \CellSim^{\SuffixSet, \Rename_i}_{\targetLg} \prefix_i$, and
        $\timeValuation{\word'_i}, \timeValuation{\overline{\word}_i} \models \Rename_i$.
 \end{itemize}
 %-----------------------------------------------------------
 By the construction of the hypothesis DTA $\hypothesisA$, we have $\word \in \Lg(\hypothesisA)$ if and only if $\word_n \in \Lg(\hypothesisA)$.
 By \cref{lemma:membership_renaming}, for each $i \in \{1, 2, \dots, n\}$, we have $\word_{i-1} \in \Lg(\hypothesisA)$ if and only if $\word_{i-1} \in \Lg(\hypothesisA)$.
 By \cref{theorem:row_faithfulness}, we have $\word_n \in \Lg(\hypothesisA)$ if and only if $\word_n \in \targetLg$.
 Overall, we have $\word \in \targetLg$ if and only if $\word \in \Lg(\hypothesisA)$.
 Therefore, we have $\targetLg = \Lg(\hypothesisA)$.
\end{proof}
%-----------------------------------------------------------

%%%%%%%%%%%%%%%%%%%%%%%%%%%%%%%%%%%%%%%%%%%%%%%%%%%%%%%%%%%%
%%%%%%%%%%%%%%%%%%%%%%%%%%%%%%%%%%%%%%%%%%%%%%%%%%%%%%%%%%%%
\section{Detail of the algorithm}
%%%%%%%%%%%%%%%%%%%%%%%%%%%%%%%%%%%%%%%%%%%%%%%%%%%%%%%%%%%%
%%%%%%%%%%%%%%%%%%%%%%%%%%%%%%%%%%%%%%%%%%%%%%%%%%%%%%%%%%%%

%%%%%%%%%%%%%%%%%%%%%%%%%%%%%%%%%%%%%%%%%%%%%%%%%%%%%%%%%%%%
\subsection{Detail of symbolic membership}\label{appendix:detail_symbolic_membership}
%%%%%%%%%%%%%%%%%%%%%%%%%%%%%%%%%%%%%%%%%%%%%%%%%%%%%%%%%%%%

%-----------------------------------------------------------
\begin{algorithm}[tbp]
 \caption{Outline of the construction of symbolic membership}%
 \label{algorithm:symbolic_membership}
 \DontPrintSemicolon{}
 \newcommand{\myCommentFont}[1]{\texttt{\footnotesize{#1}}}
 \SetCommentSty{myCommentFont}
 \SetKwFunction{FConstructTable}{ConstructTable}
 \Input{An elementary language $\elementary$ and the function $\memQKey$ mapping a timed word $\word$ to its membership to a timed language $\Lg$}
 \Output{The symbolic membership query $\guard$ of $\elementary$ in $\Lg$}
 \Fn{$\symbolicMemQ{\Lg}(\elementary)$} {
     $\guard \gets \bot$\;
     \ForEach{simple and canonical $\timedCondition' \subseteq \timedCondition$} {%
         \label{algorithm:symbolic_membership:enumerate}
         \tcp{The result is independent of the choice of $\word$}
         \KwPick{} $\word \in (u, \timedCondition')$\;
         \If{$\memQ{\word} = \top$} {
             $\guard \gets \guard \lor \timedCondition'$\;
         }
     }
     \KwReturn{} $\guard$
 }
\end{algorithm}
%-----------------------------------------------------------

We show the detail of our algorithm to obtain the symbolic membership using finite membership queries.
\cref{algorithm:symbolic_membership} shows an outline of our construction.
In \cref{algorithm:symbolic_membership},
we enumerate all the simple and canonical timed conditions $\timedCondition' \subseteq \timedCondition$.
Such enumeration can be done, \eg{} by a depth-first-search:
each depth corresponds to $(i,j)$ satisfying $0 \leq i \leq j \leq |u|$ and
at each branch, we pick $\sumTimestamp{i}{j} \in (k, k + 1)$ or $\sumTimestamp{i}{j} = k$ for some $k$.
Such an exhaustive trial (instead of, \eg{} a binary search) is necessary because $\symbolicMemQ{\Lg}(\elementary)$ is potentially non-convex.

%-----------------------------------------------------------
\begin{example}
 %-----------------------------------------------------------
 \begin{figure}[tbp]
  \begin{subfigure}[b]{.5\textwidth}
   \centering
   \begin{tikzpicture}[node distance=2.0cm,on grid,auto,scale=.8,every node/.style={transform shape},,every initial by arrow/.style={initial text={}}]
    \node[state,initial] (init)  {$\loc_0$};

    \node[state] (a) [right=of init] {$\loc_1$};

    \node[state] (aa_1) [above right=of a] {$\loc_2$};

    \node[state] (aa_2) [below right=of a] {$\loc_3$};

    \node[state] (aaa_1_1) [right=of aa_1] {$\loc_4$};

    \node[state,accepting] (aaa_1_2) [below right=of aa_1] {$\loc_5$};

    \node[state,accepting] (aaa_2_1) [above right=of aa_2] {$\loc_5$};

    \node[state] (aaa_2_2) [right=of aa_2] {$\loc_6$};

    \path[->]
    (init) edge  [above] node {$a / y \coloneqq 0$} (a)
    (a) edge  [above left] node {$a, x < 1$} (aa_1)
    (a) edge  [below left] node {$a, x \geq 1$} (aa_2)
    (aa_1) edge  [above] node {$a, y < 1$} (aaa_1_1)
    (aa_1) edge  [below left] node {$a, y \geq 1$} (aaa_1_2)
    (aa_2) edge  [above left] node {$a, y < 1$} (aaa_2_1)
    (aa_2) edge  [above] node {$a, y \geq 2$} (aaa_2_2)
    ;
   \end{tikzpicture}
   \caption{DTA recognizing $\Lg$}%
   \label{figure:non_convex_symbolic_membership:DTA}
  \end{subfigure}
  \begin{subfigure}[b]{.5\textwidth}
   \centering
   \begin{tikzpicture}[node distance=2.0cm,on grid,auto,scale=1.15,every initial by arrow/.style={initial text={}}]
    \draw (1.0, -0.1) -- (1.0, 0.1);
    \draw (-0.1, 1.0) -- (0.1, 1.0);
    \draw[very thick,-stealth] (0.0, -0.1) -- (0.0, 2.1);
    \draw[very thick,-stealth] (-0.1, 0.0) -- (2.1, 0.0);
    \draw[thick] (0.0, 1.0) -- (0.95, 1.0);
    \draw[thick] (1.0, 1.05) -- (1.0, 2.0);
    \draw (1,1) circle (0.05);
    \node at (0.2,-0.2) {$0$};
    \node at (2.05,-0.2) {$\sumTimestamp{0}{1}$};
    \node at (-0.35,2.05) {$\sumTimestamp{1}{2}$};
    \node at (1.0,-0.2) {$1.0$};
    \node at (-0.35,1.0) {$1.0$};
    \fill[pattern=north west lines] (0,1) -- (1,1) -- (1,2) -- (0,2) -- (0,1);
    \fill[pattern=north west lines] (1,0) -- (2,0) -- (2,1) -- (1,1) -- (1,0);
   \end{tikzpicture}
   \caption{$\symbolicMemQ{\Lg}(\elementary)$ projected to $\sumTimestamp{0}{1}$ and $\sumTimestamp{1}{2}$}
  \end{subfigure}
  \caption{Example of non-convex symbolic membership $\symbolicMemQ{\Lg}(\elementary)$, where $u = aaa$ and $\timedCondition = \bigwedge_{0 < i \leq j \leq 2, (i,j) \not\in \{(0,1),(1,2)\}} \sumTimestamp{i}{j} \in (0,1) \land \sumTimestamp{0}{1} \in (0,2) \land \sumTimestamp{1}{2} \in (0,2)$}%
  \label{figure:non_convex_symbolic_membership}
 \end{figure}
 %-----------------------------------------------------------
 Let $u = aaa$ and $\timedCondition = \bigwedge_{0 < i \leq j \leq 2, (i,j) \not\in \{(0,1),(1,2)\}} \sumTimestamp{i}{j} \in (0,1) \land \sumTimestamp{0}{1} \in (0,2) \land \sumTimestamp{1}{2} \in (0,2)$.
 Let $\Lg$ be the timed language recognized by the DTA in \cref{figure:non_convex_symbolic_membership:DTA}.
 The accepting location $\loc_5$ is reachable by following $\loc_0 \to \loc_1 \to \loc_2 \to \loc_5$ or $\loc_0 \to \loc_1 \to \loc_3 \to \loc_5$, and the union of the corresponding timing constraints along these paths are non-convex.
 Therefore, the symbolic membership is non-convex, as shown in \cref{figure:non_convex_symbolic_membership}.
\end{example}
%-----------------------------------------------------------

The number of membership queries used in \cref{algorithm:symbolic_membership} is the same as the number of the simple and canonical timed conditions $\timedCondition' \subseteq \timedCondition$, which exponentially blows up to the number of variables in $\timedCondition$.
In our DTA learning algorithm, \cref{algorithm:symbolic_membership} is used only for elementary languages of the form $\elementary \cdot \elementary[']$,
where $\elementary$ and $\elementary[']$ are simple.
The following example shows that even with such restriction, the number of the necessary membership queries may blow up exponentially.

%-----------------------------------------------------------
\begin{example}
 Let $\timedCondition$, and $\timedCondition'$ be
 $\timedCondition = \bigwedge_{0 \leq i \leq j \leq n} \sumTimestamp{i}{j} \in (0,1)$ and
 $\timedCondition' = \bigwedge_{0 \leq i \leq j \leq n'} \sumTimestamp[']{i}{j} \in (0,1)$.
 Let $\timedCondition \cdot \timedCondition'$ be the timed condition such that
 $\timedCondition \cdot \timedCondition' = \{\cval \cdot \cval' \mid \cval \in \timedCondition, \cval' \in \timedCondition'\}$.
 For any simple and canonical timed condition $\timedCondition'' \subseteq \timedCondition \cdot \timedCondition'$ and for each $0 \leq i \leq j \leq n + n' - 1$,
 we have $\sumTimestamp['']{i}{j} \in (0,1)$, $\sumTimestamp['']{i}{j} = 1$, or $\sumTimestamp['']{i}{j} \in (1,2)$.
 Moreover, if we have $\sumTimestamp['']{i}{j} = 1$ or $\sumTimestamp['']{i}{j} \in (1,2)$,
 for any $i'$ and $j'$ satisfying $0 \leq i' \leq i$, $j \leq j' \leq n + n' - 1$, and $(i, j) \neq (i', j')$,
 we have $\sumTimestamp['']{i'}{j'} \in (1,2)$.
 Therefore, we can underapproximate the number of the simple and canonical timed conditions $\timedCondition'' \subseteq \timedCondition \cdot \timedCondition'$ by
 counting the possible combination of such boundaries, \ie{} the pair $(i, j)$ satisfying $\sumTimestamp['']{i}{j} = 1$ or $\sumTimestamp['']{i}{j} \in (1,2)$.
 If there are no such boundaries, we have $\sumTimestamp{i}{j} \in (0,1)$ for any pair $(i,j)$, and the number of the corresponding $\timedCondition''$ is one.
 Let $m > 0$ be the number of the boundaries.
 For indices $i, j$ corresponding to one of such boundaries satisfying $0 \leq i \leq n \leq j \leq n + n' - 1$,
 for any indices $i', j'$ corresponding to any other boundaries, we have $0 \leq i < i' \leq n \leq j < j' \leq n + n' - 1$ or $0 \leq i' < i \leq n \leq j' < j \leq n + n' - 1$.
 Therefore, the correspondence between $i$ and $j$ (and $i'$ and $j'$) is uniquely determined by order of the indices (\ie{} $i <i'$ or $i' < i$), and it suffices to count the combination of the indices independently.
 Overall, the number of the simple and canonical timed conditions $\timedCondition'' \subseteq \timedCondition \cdot \timedCondition'$ is greater than the following, and we indeed have an exponential blowup, where $\binom{n}{m}$ is the binomial coefficient.
 \begin{displaymath}
  1 + \sum_{m = 0}^{\min\{n, n'\}} \binom{n}{m} \times \binom{n'}{m} > \sum_{m = 0}^{\min\{n, n'\}} \binom{\min\{n, n'\}}{m} = 2^{\min\{n, n'\}}
 \end{displaymath}
\end{example}
%-----------------------------------------------------------

%%%%%%%%%%%%%%%%%%%%%%%%%%%%%%%%%%%%%%%%%%%%%%%%%%%%%%%%%%%%
\subsection{Checking if $\elementary \CellSim^{\SuffixSet}_{\Lg} \elementary[']$}\label{appendix:row_equivalence}
%%%%%%%%%%%%%%%%%%%%%%%%%%%%%%%%%%%%%%%%%%%%%%%%%%%%%%%%%%%%

First, we show how to check if $\elementary \CellSim^{\elementary[''],\Rename}_{\Lg} \elementary[']$
for a timed language $\Lg$,
simple elementary languages $\elementary, \elementary[']$, an elementary language $\elementary['']$, and
a renaming equation $\Rename$.
We have $\elementary \CellSim^{\elementary[''],\Rename}_{\Lg} \elementary[']$ if and only if the following three conditions hold.
\begin{enumerate}
 \item For any $\word \in \elementary$, there is $\word' \in \elementary[']$ satisfying $\timeValuation{\word},\timeValuation{\word'}\models \Rename$.\label{list:row_equivalence:left_to_right}
 \item For any $\word' \in \elementary[']$, there is $\word \in \elementary$ satisfying $\timeValuation{\word},\timeValuation{\word'}\models \Rename$.\label{list:row_equivalence:right_to_left}
 \item $\rename{\symbolicMemQ{\Lg}(\elementary \cdot \elementary[''])}{\Rename} \land \timedCondition'$ holds if and only if $\symbolicMemQ{\Lg}(\elementary['] \cdot \elementary['']) \land \Rename \land \timedCondition$ holds.\label{list:row_equivalence:equivalence}
\end{enumerate}

Without loss of generality, we assume that $\timedCondition$, $\timedCondition'$, and $\timedCondition''$ are simple and canonical.
Since $\Rename$ is a finite conjunction of equations of the form $\sumTimestamp{i}{|u|} = \sumTimestamp[']{j}{|u'|}$,
 \cref{list:row_equivalence:left_to_right} holds if and only if
for any $\word \in \elementary$,
there is $\word' \in \elementary[']$ such that for any $\sumTimestamp{i}{|u|} = \sumTimestamp[']{j}{|u'|}$ in $\Rename$,
$\timeValuation{\word}(\sumTimestamp{i}{|u|}) = \timeValuation{\word'}(\sumTimestamp[']{j}{|u'|})$.
This holds if and only if $\timedCondition \Rightarrow \project{(\Rename \land \timedCondition')}{\timeVariables}$ is a tautology,
where $\project{(\Rename \land \timedCondition')}{\timeVariables}$ is the constraint $\Rename \land \timedCondition'$ (over $\timeVariables \cup \timeVariables'$) restricted to $\timeVariables$.
If $\timedCondition \Rightarrow \project{(\Rename \land \timedCondition')}{\timeVariables}$ is a tautology,
we have $\timedCondition \land \project{(\Rename \land \timedCondition')}{\timeVariables} = \timedCondition$.
Since $\timedCondition$ and $\project{(\Rename \land \timedCondition')}{\timeVariables}$ are simple and canonical,
$\timedCondition \land \project{(\Rename \land \timedCondition')}{\timeVariables} = \timedCondition$ holds if and only if
$\timedCondition \land \project{(\Rename \land \timedCondition')}{\timeVariables}$ is satisfiable.
Moreover, this is equivalent to the satisfiability of $\timedCondition \land \timedCondition' \land \Rename$.
Therefore, \cref{list:row_equivalence:left_to_right} can be checked by examining if $\timedCondition \land \timedCondition' \land \Rename$ is satisfiable.
Similarly, \cref{list:row_equivalence:right_to_left} can be checked by examining if $\timedCondition \land \timedCondition' \land \Rename$ is satisfiable.

%-----------------------------------------------------------
\begin{algorithm}[tbp]
 \caption{Outline of the checking of $\elementary \CellSim^{\SuffixSet,\Rename}_{\Lg} \elementary[']$}%
 \label{algorithm:row_equivalence_wihtout_renaming}
 \DontPrintSemicolon{}
 \newcommand{\myCommentFont}[1]{\texttt{\footnotesize{#1}}}
 \SetCommentSty{myCommentFont}
 \Input{A timed language $\Lg$, elementary languages $\elementary,\elementary[']$, a set of elementary languages $\SuffixSet$, and a renaming equation $\Rename$}
 \Output{If $\elementary \CellSim^{\SuffixSet,\Rename}_{\Lg} \elementary[']$ or not}
 \SetKwFunction{FCheckRowEquivalenceSingle}{CheckRowEquivalenceSingle}
 \Fn{\FCheckRowEquivalenceSingle} {
     \If{$\timedCondition \land \timedCondition' \land \Rename$ is unsatisfiable} {
         \KwReturn{} $\bot$
     }
     \For{$\suffix \in \SuffixSet$} {
         \If{$\rename{\symbolicMemQ{\Lg}(\elementary \cdot \suffix)}{\Rename} \land \timedCondition' \neq \symbolicMemQ{\Lg}(\elementary['] \cdot \suffix) \land \Rename \land \timedCondition$} {
             \KwReturn{} $\bot$
         }
     }
     \KwReturn{} $\top$
 }
\end{algorithm}
%-----------------------------------------------------------

One can check if $\elementary \CellSim^{\SuffixSet,\Rename}_{\Lg} \elementary[']$ by
checking $\elementary \CellSim^{\elementary[''],\Rename}_{\Lg} \elementary[']$ for each $\elementary[''] \in \SuffixSet$.
\cref{algorithm:row_equivalence_wihtout_renaming} shows an outline of it.

To check if $\elementary \CellSim^{\SuffixSet}_{\Lg} \elementary[']$,
we generate candidate renaming equations $\Rename$ and check if $\elementary \CellSim^{\SuffixSet,\Rename}_{\Lg} \elementary[']$.
Each $\Rename$ can be seen as a bipartite graph $(V_1, V_2, E)$ such that $V_1 = \{\sumTimestamp{i}{n} \mid 0 \leq i \leq n\}$, $V_2 = \{\sumTimestamp[']{i'}{n'} \mid 0 \leq i' \leq n'\}$, and $E = \{(\sumTimestamp{i}{n}, \sumTimestamp{i'}{n'}) \mid \text{$R$ contains $\sumTimestamp{i}{n} = \sumTimestamp{i'}{n'}$}\}$.
Therefore, constructing a candidate renaming equation $\Rename$ is equivalent to constructing a set of edges $E$.

Before showing our construction, we prove the following properties.

%-----------------------------------------------------------
\begin{proposition}%
 \label{prop:renaming_only_equal}
 For any $\elementary$, $\elementary[']$, $\SuffixSet$, $\Lg$, and $\Rename$ satisfying
 $\elementary \CellSim^{\SuffixSet,\Rename}_{\Lg} \elementary[']$,
 if $\timedCondition$ and $\timedCondition'$ are simple and canonical,
 for each equation $\sumTimestamp{i}{|u|} = \sumTimestamp[']{i'}{|u'|}$ in $\Rename$,
 the range of $\sumTimestamp{i}{|u|}$ in $\timedCondition$ and
 the range of $\sumTimestamp{i'}{|u'|}$ in $\timedCondition'$ are the same.
\end{proposition}
%-----------------------------------------------------------
\begin{proof}
 For $i \in \{0,1,\dots,|u|\}$ (\resp{} $i' \in \{0,1,\dots,|u'|\}$),
 let $I_i$ (\resp{} $I'_{i'}$) be such that
 $\timedCondition$ (\resp{} $\timedCondition'$) contains
 $\sumTimestamp{i}{|u|} = I_i$ (\resp{} $\sumTimestamp[']{i'}{|u'|} = I'_{i'}$).
 Since $\timedCondition$ and $\timedCondition'$ are simple and canonical,
 $I_i \neq I'_{i'}$ implies $I_i \cap I'_{i'} = \emptyset$.
 Since we have $\elementary \CellSim^{\SuffixSet,\Rename}_{\Lg} \elementary[']$,
 $\timedCondition \land \timedCondition' \land \Rename$ must be satisfiable.
 Therefore,
 for each equation $\sumTimestamp{i}{|u|} = \sumTimestamp[']{i'}{|u'|}$ in $\Rename$,
 we have $I_i = I'_{i'}$.
\end{proof}
\begin{definition}
 Let $\Lg$ be a timed language
 and $\elementary$ and $\elementary[']$ be elementary languages over $\timeVariables$ and $\timeVariables'$, respectively.
 For $i \in \{0,1,\dots,n\}$ and $i' \in \{0,1,\dots,n'\}$,
 we call $\sumTimestamp{i}{n} + \sumTimestamp[']{0}{i'}$ is \emph{non-trivial} in $\symbolicMemQ{\Lg}(\elementary \cdot \elementary['])$
 if the range of $\sumTimestamp{i}{n} + \sumTimestamp[']{0}{i'}$ is different between $\timedCondition \land \timedCondition'$ and $\symbolicMemQ{\Lg}(\elementary \cdot \elementary['])$.
\end{definition}
%-----------------------------------------------------------

%-----------------------------------------------------------
\begin{proposition}
 For any $\elementary$, $\elementary[']$, $\SuffixSet$, $\Lg$, and $\Rename$ satisfying
 $\elementary \CellSim^{\SuffixSet,\Rename}_{\Lg} \elementary[']$,
 if $\timedCondition$ and $\timedCondition'$ are simple and canonical,
 for each equation $\sumTimestamp{i}{|u|} = \sumTimestamp[']{i'}{|u'|}$ in $\Rename$,
 for each $\elementary[''] \in \SuffixSet$, and
 for each $i'' \in \{0,1,\dots,|u''|\}$,
 $\sumTimestamp{i}{|u|} + \sumTimestamp['']{0}{i''}$ is non-trivial in $\symbolicMemQ{\Lg}(\elementary \cdot \elementary[''])$ if and only if
 $\sumTimestamp[']{i'}{|u'|} + \sumTimestamp['']{0}{i''}$ is non-trivial in $\symbolicMemQ{\Lg}(\elementary['] \cdot \elementary[''])$.
\end{proposition}
%-----------------------------------------------------------
\begin{proof}
 Assume $\sumTimestamp{i}{|u|} + \sumTimestamp['']{0}{i''}$ is non-trivial in $\symbolicMemQ{\Lg}(\elementary \cdot \elementary[''])$,
 \ie{} the range of $\sumTimestamp{i}{|u|} + \sumTimestamp['']{0}{i''}$ is different
 between $\timedCondition \land \timedCondition''$ and $\symbolicMemQ{\Lg}(\elementary \cdot \elementary[''])$.
 Since $\symbolicMemQ{\Lg}(\elementary \cdot \elementary[''])$ implies $\timedCondition \land \timedCondition''$,
 the range of $\sumTimestamp{i}{|u|} + \sumTimestamp['']{0}{i''}$ remains different
 by strengthening $\symbolicMemQ{\Lg}(\elementary \cdot \elementary[''])$.
 Therefore, the range of $\sumTimestamp{i}{|u|} + \sumTimestamp['']{0}{i''}$ is different
 between $\timedCondition \land \timedCondition''$ and
 $\symbolicMemQ{\Lg}(\elementary \cdot \elementary['']) \land \timedCondition' \land \Rename$.
 By $\elementary \CellSim^{\SuffixSet,\Rename}_{\Lg} \elementary[']$, we have $\timedCondition = \timedCondition \land \project{(\timedCondition' \land \Rename)}{\timeVariables}$.
 Therefore, the range of $\sumTimestamp{i}{|u|} + \sumTimestamp['']{0}{i''}$ is different
 between $\timedCondition \land \project{(\timedCondition' \land \Rename)}{\timeVariables} \land \timedCondition''$ and
 $\symbolicMemQ{\Lg}(\elementary \cdot \elementary['']) \land \timedCondition' \land \Rename$.
 Since $\timedCondition' \land \Rename$ is a constraint over $\timeVariables \cup \timeVariables'$,
 the range of $\sumTimestamp{i}{|u|} + \sumTimestamp['']{0}{i''}$ is different
 between $\timedCondition \land \timedCondition' \land \timedCondition'' \land \Rename$ and
 $\symbolicMemQ{\Lg}(\elementary \cdot \elementary['']) \land \timedCondition' \land \Rename$.
 % Similarly, since $\symbolicMemQ{\Lg}(\elementary \cdot \elementary[''])$ implies $\timedCondition \land \timedCondition''$,
 % the range of $\sumTimestamp{i}{|u|} + \sumTimestamp['']{0}{i''}$ is the same in
 % $\symbolicMemQ{\Lg}(\elementary \cdot \elementary[''])$ and $\symbolicMemQ{\Lg}(\elementary \cdot \elementary['']) \land \timedCondition' \land \Rename$.
 % Therefore, the range of $\sumTimestamp{i}{|u|} + \sumTimestamp['']{0}{i''}$ is different
 % between $\timedCondition \land \timedCondition' \land \timedCondition'' \land \Rename$ and
 % $\symbolicMemQ{\Lg}(\elementary \cdot \elementary['']) \land \timedCondition' \land \Rename$.
 Since $\Rename$ contains $\sumTimestamp{i}{|u|} = \sumTimestamp[']{i'}{|u'|}$,
 the range of $\sumTimestamp[']{i'}{|u'|} + \sumTimestamp['']{0}{i''}$ is also different
 between $\timedCondition \land \timedCondition' \land \timedCondition'' \land \Rename$ and
 $\symbolicMemQ{\Lg}(\elementary \cdot \elementary['']) \land \timedCondition' \land \Rename$.
 Since we have $\elementary \CellSim^{\elementary[''],\Rename}_{\Lg} \elementary[']$,
 $\rename{\symbolicMemQ{\Lg}(\elementary \cdot \elementary[''])}{\Rename} \land \timedCondition'$ holds if and only if
 $\symbolicMemQ{\Lg}(\elementary['] \cdot \elementary['']) \land \Rename \land \timedCondition$ holds.
 Therefore,
 the range of $\sumTimestamp[']{i'}{|u'|} + \sumTimestamp['']{0}{i''}$ is different
 between $\timedCondition \land \timedCondition' \land \timedCondition'' \land \Rename$ and
 $\symbolicMemQ{\Lg}(\elementary['] \cdot \elementary['']) \land \timedCondition \land \Rename$.
 Since $\symbolicMemQ{\Lg}(\elementary['] \cdot \elementary['']) \land \timedCondition \land \Rename$ implies $\timedCondition' \land \timedCondition'' \land \timedCondition \land \Rename$,
 the range of $\sumTimestamp[']{i'}{|u'|} + \sumTimestamp['']{0}{i''}$ remains different
 by relaxing $\timedCondition' \land \timedCondition'' \land \timedCondition \land \Rename$.
 Therefore,
 the range of $\sumTimestamp[']{i'}{|u'|} + \sumTimestamp['']{0}{i''}$ is different
 between $\timedCondition' \land \timedCondition''$ and
 $\symbolicMemQ{\Lg}(\elementary['] \cdot \elementary['']) \land \timedCondition \land \Rename$.
 Since $\symbolicMemQ{\Lg}(\elementary['] \cdot \elementary[''])$ implies $\timedCondition'$,
 the range of $\sumTimestamp[']{i'}{|u'|} + \sumTimestamp['']{0}{i''}$ is different
 between $\timedCondition' \land \timedCondition''$ and
 $\symbolicMemQ{\Lg}(\elementary['] \cdot \elementary['']) \land \timedCondition \land \timedCondition' \land \Rename$.
 Since $\sumTimestamp[']{i'}{|u'|} + \sumTimestamp['']{0}{i''}$ does not contain any variables in $\timeVariables$,
 the range of $\sumTimestamp[']{i'}{|u'|} + \sumTimestamp['']{0}{i''}$ is different
 between $\timedCondition' \land \timedCondition''$ and
 $\symbolicMemQ{\Lg}(\elementary['] \cdot \elementary['']) \land \timedCondition' \land \project{(\timedCondition \land \Rename)}{\timeVariables'}$.
 By $\elementary \CellSim^{\SuffixSet,\Rename}_{\Lg} \elementary[']$, we have $\timedCondition' = \timedCondition' \land \project{(\timedCondition \land \Rename)}{\timeVariables'}$.
 Therefore,
 the range of $\sumTimestamp[']{i'}{|u'|} + \sumTimestamp['']{0}{i''}$ is different
 between $\timedCondition' \land \timedCondition''$ and
 $\symbolicMemQ{\Lg}(\elementary['] \cdot \elementary['']) \land \timedCondition'$.
 Since $\symbolicMemQ{\Lg}(\elementary['] \cdot \elementary[''])$ implies $\timedCondition'$,
 we have $\symbolicMemQ{\Lg}(\elementary['] \cdot \elementary['']) = \symbolicMemQ{\Lg}(\elementary['] \cdot \elementary['']) \land \timedCondition'$, and thus,
 the range of $\sumTimestamp[']{i'}{|u'|} + \sumTimestamp['']{0}{i''}$ is different
 between $\timedCondition' \land \timedCondition''$ and
 $\symbolicMemQ{\Lg}(\elementary['] \cdot \elementary[''])$.
 Therefore,
 $\sumTimestamp{i'}{|u'|} + \sumTimestamp['']{0}{i''}$ is non-trivial in $\symbolicMemQ{\Lg}(\elementary['] \cdot \elementary[''])$.
 The proof of the opposite direction is similar.
\end{proof}
%-----------------------------------------------------------

% Existence of ``non-redundant'' renaming equation
%-----------------------------------------------------------
\begin{proposition}%
 \label{prop:only_non_trivial_renaming}
 For any $\elementary$, $\elementary[']$, $\SuffixSet$, and $\Lg$ satisfying
 $\elementary \CellSim^{\SuffixSet}_{\Lg} \elementary[']$,
 if $\timedCondition$ and $\timedCondition'$ are simple and canonical,
 there is a renaming equation $\Rename$ such that
 for any equation $\sumTimestamp{i}{|u|} = \sumTimestamp[']{i'}{|u'|}$ in $\Rename$,
 there are $\elementary[''] \in \SuffixSet$ and $i'' \in \{0,1,\dots,|u''|\}$ such that
 $\sumTimestamp{i}{|u|} + \sumTimestamp['']{0}{i''}$ is non-trivial in $\symbolicMemQ{\Lg}(\elementary \cdot \elementary[''])$.
\end{proposition}
%-----------------------------------------------------------
\begin{proof}
 Let $\Rename' = \Rename \land (\sumTimestamp{i}{|u|} = \sumTimestamp[']{i'}{|u'|})$ be a renaming equation satisfying
 $\elementary \CellSim^{\SuffixSet,\Rename'}_{\Lg} \elementary[']$, and
 for any $\elementary[''] \in \SuffixSet$ and for any $i'' \in \{0,1,\dots,|u''|\}$,
 $\sumTimestamp{i}{|u|} + \sumTimestamp['']{0}{i''}$ is trivial in $\symbolicMemQ{\Lg}(\elementary \cdot \elementary[''])$.
 We show that $\elementary \CellSim^{\SuffixSet,\Rename}_{\Lg} \elementary[']$ also holds.
 By definition of $\elementary \CellSim^{\SuffixSet,\Rename'}_{\Lg} \elementary[']$,
 we have the following three conditions.
 \begin{itemize}
  \item For any $\word \in \elementary$, there is $\word' \in \elementary[']$ satisfying $\timeValuation{\word},\timeValuation{\word'}\models \Rename'$.
  \item For any $\word' \in \elementary[']$, there is $\word \in \elementary$ satisfying $\timeValuation{\word},\timeValuation{\word'}\models \Rename'$.
  \item For any $\elementary['']\in\SuffixSet$, $\rename{\symbolicMemQ{\Lg}(\elementary \cdot \elementary[''])}{\Rename'} \land \timedCondition'$ holds if and only if $\symbolicMemQ{\Lg}(\elementary['] \cdot \elementary['']) \land \Rename' \land \timedCondition$ holds.
 \end{itemize}
 Since $\Rename'$ implies $\Rename$, we immediately have the following two conditions.
 \begin{itemize}
  \item For any $\word \in \elementary$, there is $\word' \in \elementary[']$ satisfying $\timeValuation{\word},\timeValuation{\word'}\models \Rename$.
  \item For any $\word' \in \elementary[']$, there is $\word \in \elementary$ satisfying $\timeValuation{\word},\timeValuation{\word'}\models \Rename$.
 \end{itemize}
 By definition of $\Rename'$,
 for any $\elementary['']\in\SuffixSet$,
 $\rename{\symbolicMemQ{\Lg}(\elementary \cdot \elementary[''])}{\Rename \land (\sumTimestamp{i}{|u|} = \sumTimestamp[']{i'}{|u'|})} \land \timedCondition'$ holds
 if and only if $\symbolicMemQ{\Lg}(\elementary['] \cdot \elementary['']) \land \Rename \land (\sumTimestamp{i}{|u|} = \sumTimestamp[']{i'}{|u'|}) \land \timedCondition$ holds.
 Since $\symbolicMemQ{\Lg}(\elementary \cdot \elementary[''])$ implies $\timedCondition \land \timedCondition''$,
 $\rename{\symbolicMemQ{\Lg}(\elementary \cdot \elementary[''])}{\Rename \land (\sumTimestamp{i}{|u|} = \sumTimestamp[']{i'}{|u'|})} \land \timedCondition'$ is equivalent to
 $\rename{\symbolicMemQ{\Lg}(\elementary \cdot \elementary[''])}{\Rename \land (\sumTimestamp{i}{|u|} = \sumTimestamp[']{i'}{|u'|})} \land \timedCondition \land \timedCondition' \land \timedCondition''$.
 Since for any $i'' \in \{0,1,\dots,|u''|\}$,
 $\sumTimestamp{i}{|u|} + \sumTimestamp['']{0}{i''}$ is trivial in $\symbolicMemQ{\Lg}(\elementary \cdot \elementary[''])$,
 $\rename{\symbolicMemQ{\Lg}(\elementary \cdot \elementary[''])}{\Rename \land (\sumTimestamp{i}{|u|} = \sumTimestamp[']{i'}{|u'|})} \land \timedCondition \land \timedCondition' \land \timedCondition''$ is equivalent to
 $\rename{\symbolicMemQ{\Lg}(\elementary \cdot \elementary[''])}{\Rename} \land \timedCondition \land \timedCondition' \land \timedCondition''$, which is also equivalent to
 $\rename{\symbolicMemQ{\Lg}(\elementary \cdot \elementary[''])}{\Rename} \land \timedCondition'$.
 Therefore,
 for any $\elementary['']\in\SuffixSet$
 $\rename{\symbolicMemQ{\Lg}(\elementary \cdot \elementary[''])}{\Rename'} \land \timedCondition'$ holds if and only if 
 $\rename{\symbolicMemQ{\Lg}(\elementary \cdot \elementary[''])}{\Rename} \land \timedCondition'$ holds.
 By a similar discussion,
 for any $\elementary['']\in\SuffixSet$, 
 $\rename{\symbolicMemQ{\Lg}(\elementary['] \cdot \elementary[''])}{\Rename'} \land \timedCondition$ holds if and only if 
 $\rename{\symbolicMemQ{\Lg}(\elementary['] \cdot \elementary[''])}{\Rename} \land \timedCondition$ holds.
 Overall, 
 for any $\elementary['']\in\SuffixSet$, 
 $\rename{\symbolicMemQ{\Lg}(\elementary \cdot \elementary[''])}{\Rename} \land \timedCondition'$ holds if and only if $\symbolicMemQ{\Lg}(\elementary['] \cdot \elementary['']) \land \Rename \land \timedCondition$ holds, and thus, we have
 $\elementary \CellSim^{\SuffixSet,\Rename}_{\Lg} \elementary[']$.
\end{proof}
%-----------------------------------------------------------

%-----------------------------------------------------------
\begin{algorithm}[tbp]
 \caption{Outline of an algorithm to check if $\elementary \CellSim^{\SuffixSet}_{\Lg} \elementary[']$}%
 \label{algorithm:row_equivalence}
 \DontPrintSemicolon{}
 % \newcommand{\myCommentFont}[1]{\texttt{\footnotesize{#1}}}
 % \SetCommentSty{myCommentFont}
 \Input{A timed language $\Lg$, elementary languages $\elementary,\elementary[']$, and a set of elementary languages $\SuffixSet$}
 \Output{If $\elementary \CellSim^{\SuffixSet}_{\Lg} \elementary[']$ or not}
 \textbf{construct} $(\tilde{V_1}, \tilde{V_2}, \tilde{E})$\label{algorithm:row_equivalence:construct_bipartite_graph}\;
 $\overline{\Rename} \gets \{\top\}$\;
\tcp{Generate candidate renaming equations}
\ForEach{disjoint subgraph $(\tilde{V'_1}, \tilde{V'_2}, \tilde{E'})$ of $(\tilde{V_1}, \tilde{V_2}, \tilde{E})$} {
      \tcp{Pick one edge for each disjoint subgraph}
      $\overline{\Rename} \gets \{\Rename \land \sumTimestamp{i}{|u|} = \sumTimestamp[']{i'}{|u'|} \mid \Rename \in \overline{\Rename}, (\timestamp_i, \timestamp'_{i'}) \in \tilde{E'}\}$\;
 }
 \While{$\forall \Rename \in \overline{\Rename}.\, \neg \FCheckRowEquivalenceSingle{$\Lg, \elementary, \elementary['], \SuffixSet, \Rename$}$} {
     $\overline{\Rename_{\mathrm{new}}} \gets \emptyset$\;
     \tcp{Try to add another equation}
     \ForEach{disjoint subgraph $(\tilde{V'_1}, \tilde{V'_2}, \tilde{E'})$ of $(\tilde{V_1}, \tilde{V_2}, \tilde{E})$} {
         \ForEach{$\Rename \in \overline{\Rename}$ and $(\timestamp_i, \timestamp'_{i'}) \in \tilde{E'}$} {
             \If{$\timedCondition \land \timedCondition' \land \Rename \land \sumTimestamp{i}{|u|} = \sumTimestamp[']{i'}{|u'|}$ is satisfiable} {
                 \KwPush{} $\Rename \land \sumTimestamp{i}{|u|} = \sumTimestamp[']{i'}{|u'|}$ \KwTo{} $\overline{\Rename_{\mathrm{new}}}$
             }
         }
     }
     \If{$\overline{\Rename_{\mathrm{new}}} = \overline{\Rename} \lor \overline{\Rename_{\mathrm{new}}} = \emptyset$} {
         \KwReturn{} $\bot$
     }
     $\overline{\Rename} \gets \overline{\Rename_{\mathrm{new}}}$
 }
 \KwReturn{} $\top$
\end{algorithm}
%-----------------------------------------------------------

\cref{algorithm:row_equivalence} outlines our construction of the candidate renaming equations.
In \cref{algorithm:row_equivalence:construct_bipartite_graph},
we construct the bipartite graph $(\tilde{V_1}, \tilde{V_2}, \tilde{E})$ such that:
%-----------------------------------------------------------
\begin{itemize}
 \item $\tilde{V_1} \subseteq \timeVariables$ and $\tilde{V_2} \subseteq \timeVariables'$ are such that $\timestamp_i \in \tilde{V_1}$ (\resp{} $\timestamp'_{i'} \in \tilde{V_2}$) if and only if there are $\elementary[''] \in \SuffixSet$ and $i'' \in \{0,1,\dots,|u''|\}$ such that $\sumTimestamp{i}{|u|} + \sumTimestamp['']{0}{i''}$ (\resp{} $\sumTimestamp[']{i'}{|u'|} + \sumTimestamp['']{0}{i''}$) is non-trivial in $\symbolicMemQ{\Lg}(\elementary \cdot \elementary[''])$ (\resp{} $\symbolicMemQ{\Lg}(\elementary['] \cdot \elementary[''])$);
 \item $\tilde{E}$ is such that $(\timestamp_i, \timestamp'_{i'}) \in \tilde{E}$ if and only if
       the range of $\sumTimestamp{i}{|u|}$ in $\timedCondition$ and
       the range of $\sumTimestamp[']{i'}{|u'|}$ in $\timedCondition'$ are the same.
\end{itemize}
%-----------------------------------------------------------
We encode a renaming equation $\Rename$ by a set $\tilde{E}_{\Rename} \subseteq \tilde{E}$ of edges such that
$\Rename$ contains $\sumTimestamp{i}{|u|} = \sumTimestamp[']{i'}{|u'|}$ if and only if $(\timestamp_{i}, \timestamp'_{i'}) \in \tilde{E}_{\Rename}$.
By \cref{prop:only_non_trivial_renaming}, each candidate renaming equation $\Rename$ does not have to contain the variables not in $\tilde{V_1} \cup \tilde{V_2}$.
The edges $\tilde{E}$ consists of the candidate equations satisfying the condition in \cref{prop:renaming_only_equal}.
Therefore, we can check if $\elementary \CellSim^{\SuffixSet}_{\Lg} \elementary[']$ by enumerating $\tilde{E}_{\Rename} \subseteq \tilde{E}$ and checking if we have $\elementary \CellSim^{\SuffixSet,\Rename}_{\Lg} \elementary[']$.
We note that each disjoint subgraph of this bipartite graph is complete.

Since each variable in $\tilde{V_1} \cup \tilde{V_2}$ must be constrained by $\Rename$,
$\Rename$ has to contain at least one equation for each disjoint subgraph of the bipartite $(\tilde{V_1}, \tilde{V_2}, \tilde{E})$.
For each candidate renaming equation $\Rename$, we check if $\elementary \CellSim^{\SuffixSet,\Rename}_{\Lg} \elementary[']$.
If $\elementary \CellSim^{\SuffixSet,\Rename}_{\Lg} \elementary[']$ is satisfied by some $\Rename$, we conclude $\elementary \CellSim^{\SuffixSet}_{\Lg} \elementary[']$.

We remark that, overall, the checking of $\elementary \CellSim^{\SuffixSet}_{\Lg} \elementary[']$ only uses the information in the timed observation table,
and does not affect the overall complexity of the learning algorithm in terms of the number of queries.

\section{On the consistency in L*-style algorithms}\label{subsection:consistency}
%%%%%%%%%%%%%%%%%%%%%%%%%%%%%%%%%%%%%%%%%%%%%%%%%%%%%%%%%%%%
%%%%%%%%%%%%%%%%%%%%%%%%%%%%%%%%%%%%%%%%%%%%%%%%%%%%%%%%%%%%

The following example shows that 
%Here, we present an example showing that 
 we cannot refine $\SuffixSet$ by a ``continuous consistency'' defined by the continuous successors.

%-----------------------------------------------------------
\begin{example}%
 \label{example:continuous-consistency}
 Let $\Lg = \{\timestamp_0 a \timestamp_1 b \timestamp_2 c \timestamp_3 \mid \timestamp_0 + \timestamp_1 = \timestamp_1 + \timestamp_2 = 1\}$ and
 let $\prefix = \fractionalElementary$ and $\prefix' = \fractionalElementary[']$ be such that
 $u = u' = ab$,
 $\timedCondition = \{\timestamp_0 \in (0,1) \land \timestamp_1 \in (0,1) \land \timestamp_2 = 0 \land \timestamp_0 + \timestamp_1 = 1 \land \timestamp_1 + \timestamp_2 \in (0,1) \land \timestamp_0 + \timestamp_1 + \timestamp_2 = 1 \}$,
 $\timedCondition' = \{\timestamp_0 \in (0,1) \land \timestamp_1 \in (0,1) \land \timestamp_2 \in (0,1) \land \timestamp_0 + \timestamp_1 = 1 \land \timestamp_1 + \timestamp_2 \in (0,1) \land \timestamp_0 + \timestamp_1 + \timestamp_2 \in (1,2) \}$.
 See \cref{figure:continuous-consistency:languages} for an illustration.
 The orders on the fractional parts are as follows.
 \begin{itemize}
  \item $\fractionalCondition_{\prefix}$:  $0 = \fractional(\sumTimestamp{2}{2}) < \fractional(\sumTimestamp{0}{2}) < \fractional(\sumTimestamp{1}{2})$
  \item $\fractionalCondition_{\prefix'}$: $0 < \fractional(\sumTimestamp{2}{2}) < \fractional(\sumTimestamp{0}{2}) < \fractional(\sumTimestamp{1}{2})$
 \end{itemize}

 Let $\suffix = \elementary[_{\suffix}]$ be an elementary language.
 $(\prefix \cdot \suffix) \cap \Lg$ is nonempty if and only if we have
 $u_{\suffix} = c$ and the range of $\timestamp^{\suffix}_0$ in $\timedCondition_{\suffix}$ contains $(0, 1)$.
 $(\prefix' \cdot \suffix) \cap \Lg$ is also nonempty if and only if we have the same condition.
 % For any elementary language
 % we have $(\prefix \cdot \suffix) \cap \Lg = \emptyset$ if and only if
 % $(\prefix' \cdot \suffix) \cap \Lg = \emptyset$.
 Moreover, for any such $\suffix$, we have
 $\symbolicMemQ{\Lg}(\prefix \cdot \suffix) = \timedCondition \land \timedCondition_{\suffix} \land (\timestamp_1 + \timestamp^{\suffix}_0 = 1)$ and
 $\symbolicMemQ{\Lg}(\prefix' \cdot \suffix) = \timedCondition' \land \timedCondition_{\suffix} \land (\timestamp'_1 + \timestamp'_2 + \timestamp^{\suffix}_0 = 1)$.
 Therefore, we have
 $\prefix \CellSim^{\Elementary,\Rename}_{\Lg} \prefix'$, with $\Rename = \{\timestamp_1 = \timestamp'_1 + \timestamp'_2\}$.

 In contrast, for $\successor[t]{\prefix} = \prefix'$ and $\successor[t]{\prefix'}$,
 for $\suffix' = \elementary[_{\suffix'}]$ such that $u_{\suffix'} = c$ and $\timedCondition_{\suffix'} = \{\timestamp^{\suffix'}_0 = \timestamp^{\suffix'}_1 = 0\}$,
 we have $\successor[t]{\prefix} \cdot \suffix' \cap \Lg = \emptyset$ and $\successor[t]{\prefix'} \cdot \suffix' \cap \Lg \neq \emptyset$.
 Therefore, there is no renaming equation $\Rename$ satisfying
 $\successor[t]{\prefix} \NCellSim^{\Elementary,\Rename}_{\Lg} \successor[t]{\prefix'}$.
 Overall, in general, the ``continuous consistency'' does not hold even for $\SuffixSet = \Elementary$, and thus,
 we cannot refine $\SuffixSet$ by enforcing ``consistency'' for continuous successors.
 %-----------------------------------------------------------
 \begin{figure}[tbp]
  \centering
  \begin{tikzpicture}
   %============== Lg =========================
   \node at (-0.6,0) {$\Lg: $};
   \draw[->] (-0.3, 0.0) -- (3.3, 0.0);
   \draw (0.0, -0.1) -- (0.0, 0.1);
   \draw (1.0, -0.1) -- (1.0, 0.1);
   \draw (2.0, -0.1) -- (2.0, 0.1);
   \draw (3.0, -0.1) -- (3.0, 0.1);
   \node at (0.0, -0.3) {0};
   \node at (1.0, 0.3) {$\styleact{a}$};
   \node at (2.0, 0.3) {$\styleact{b}$};
   \node at (3.0, 0.3) {$\styleact{c}$};
   \node at (3.3, -0.3) {$t$};
   \draw (0.0, 0.5) to [bend left=30] node [above,midway] {$\timestamp_0 + \timestamp_1 = 1$} (2.0, 0.5);
   \draw (1.0, -0.2) to [bend right=30] node [below,midway] {$\timestamp_1 + \timestamp_2 = 1$} (3.0, -0.2);
  \end{tikzpicture}
  \hfill
  \begin{tikzpicture}[xscale=1.5]
   %============== prefix =========================
   \node at (-0.6,0) {$\prefix: $};
   \draw[->] (-0.3, 0.0) -- (2.3, 0.0);
   \draw (0.0, -0.1) -- (0.0, 0.1);
   \draw (1.0, -0.1) -- (1.0, 0.1);
   \draw (2.0, -0.1) -- (2.0, 0.1);
   \node at (0.0, -0.3) {0};
   \node at (1.0, 0.3) {$\styleact{a}$};
   \node at (2.0, 0.3) {$\styleact{b}$};
   \node at (2.3, -0.3) {$t$};
   \draw (0.0, 0.5) to [bend left=30] node [above,midway] {$\timestamp_0 \in (0,1)$} (1.0, 0.5);
   \draw (1.0, 0.5) to [bend left=30] node [above,midway] {$\timestamp_1 \in (0,1)$} (2.0, 0.5);
   \draw (0.0, -0.2) to [bend right=30] node [below,midway] {$\timestamp_0 + \timestamp_1 =1$} (2.0, -0.2);
  \end{tikzpicture}
  \begin{tikzpicture}[xscale=2.0]
   %============== prefix' =========================
   \node at (-0.6,0) {$\prefix': $};
   \draw[->] (-0.3, 0.0) -- (3.3, 0.0);
   \draw (0.0, -0.1) -- (0.0, 0.1);
   \draw (1.0, -0.1) -- (1.0, 0.1);
   \draw (2.0, -0.1) -- (2.0, 0.1);
   \draw (3.0, -0.1) -- (3.0, 0.1);
   \node at (0.0, -0.3) {0};
   \node at (1.0, 0.3) {$\styleact{a}$};
   \node at (2.0, 0.3) {$\styleact{b}$};
   \node at (3.0, 0.3) {};
   \node at (3.3, -0.3) {$t$};
   \draw (0.0, 0.5) to [bend left=80] node [above,midway] {$\timestamp_0 + \timestamp_1 + \timestamp_2 \in (1,2)$} (3.0, 0.5);
   \draw (0.0, 0.5) to [bend left=30] node [above,midway] {$\timestamp_0 \in (0,1)$} (1.0, 0.5);
   \draw (1.0, 0.5) to [bend left=30] node [above,midway] {$\timestamp_1 \in (0,1)$} (2.0, 0.5);
   \draw (2.0, 0.5) to [bend left=30] node [above,midway] {$\timestamp_2 \in (0,1)$} (3.0, 0.5);
   \draw (0.0, -0.2) to [bend right=30] node [below,midway] {$\timestamp_0 + \timestamp_1 =1$} (2.0, -0.2);
   \draw (1.0, -0.2) to [bend right=30] node [below,midway] {$\timestamp_1 + \timestamp_2 \in (0,1)$} (3.0, -0.2);
  \end{tikzpicture}
  \begin{tikzpicture}[xscale=2.0]
   %============== prefix' =========================
   \node at (-0.6,0) {$\successor[t]{\prefix'}: $};
   \draw[->] (-0.3, 0.0) -- (3.3, 0.0);
   \draw (0.0, -0.1) -- (0.0, 0.1);
   \draw (1.0, -0.1) -- (1.0, 0.1);
   \draw (2.0, -0.1) -- (2.0, 0.1);
   \draw (3.0, -0.1) -- (3.0, 0.1);
   \node at (0.0, -0.3) {0};
   \node at (1.0, 0.3) {$\styleact{a}$};
   \node at (2.0, 0.3) {$\styleact{b}$};
   \node at (3.0, 0.3) {};
   \node at (3.3, -0.3) {$t$};
   \draw (0.0, 0.5) to [bend left=80] node [above,midway] {$\timestamp_0 + \timestamp_1 + \timestamp_2 \in (1,2)$} (3.0, 0.5);
   \draw (0.0, 0.5) to [bend left=30] node [above,midway] {$\timestamp_0 \in (0,1)$} (1.0, 0.5);
   \draw (1.0, 0.5) to [bend left=30] node [above,midway] {$\timestamp_1 \in (0,1)$} (2.0, 0.5);
   \draw (2.0, 0.5) to [bend left=30] node [above,midway] {$\timestamp_2 \in (0,1)$} (3.0, 0.5);
   \draw (0.0, -0.2) to [bend right=30] node [below,midway] {$\timestamp_0 + \timestamp_1 =1$} (2.0, -0.2);
   \draw (1.0, -0.2) to [bend right=30] node [below,midway] {$\timestamp_1 + \timestamp_2 = 1$} (3.0, -0.2);
  \end{tikzpicture}
  \caption{Illustration of $\Lg$, $\prefix$, $\prefix'$, and $\successor[t]{\prefix'}$ in \cref{example:continuous-consistency}}%
  \label{figure:continuous-consistency:languages}
 \end{figure}
 %-----------------------------------------------------------
\end{example}

In contrast, any timed observation table becomes \emph{discretely} consistent in the limit.

% The proof of \cref{lemma:consistency_in_limit} is as follows.

 %-----------------------------------------------------------
 \newcommand{\consistencyInLimitStatement}{%
 Let $O = \TimedObsTableInside$ be a timed observation table.
 If $\CellSim^{\SuffixSet}_{\targetLg}$ is equal to $\CellSim^{\Elementary}_{\targetLg}$,
 $O$ is consistent.
 }
\begin{proposition}%
 \label{lemma:consistency_in_limit}
 \consistencyInLimitStatement{}
 \ShortVersion{\qed{}}
 \end{proposition}
 %-----------------------------------------------------------

 % \recallResult{lemma:consistency_in_limit}{\consistencyInLimitStatement{}}

 %-----------------------------------------------------------
 \begin{proof}
 We prove the contraposition, \ie{}
 $\successor[\action]{\prefix} \NCellSim^{\Elementary}_{\targetLg} \successor[\action]{\prefix'}$ implies
 $\prefix \NCellSim^{\Elementary}_{\targetLg} \prefix'$.
 When
 $\successor[\action]{\prefix} \NCellSim^{\Elementary}_{\targetLg} \successor[\action]{\prefix'}$ holds,
 since $\targetLg$ is a recognizable timed language, by \cref{theorem:finite_suffix},
 there is a finite set $\SuffixSet \subseteq \Elementary$ satisfying
 $\successor[\action]{\prefix} \NCellSim^{\SuffixSet}_{\targetLg} \successor[\action]{\prefix'}$.
 Moreover, by definition of $\CellSim^{\SuffixSet}_{\targetLg}$, for any renaming equation $\Rename$, we have
 $\successor[\action]{\prefix} \NCellSim^{\SuffixSet, \Rename}_{\targetLg} \successor[\action]{\prefix'}$.
 For such $\SuffixSet$ and $\Rename$, let
 $\SuffixSet' = \{\action \cdot \suffix \mid \suffix \in \SuffixSet\}$ and let
 $\Rename'$ be the renaming equation such that $\timestamp_{|u|+1}$ and $\timestamp'_{|u'|+1}$ in $\Rename$ are replaced with $\timestamp_{|u|}$ and $\timestamp'_{|u'|}$, respectively,
 where
 $\action \cdot \suffix = \{\action \cdot \word \mid \word \in \suffix\}$, and
 $u$ and $u'$ are such that
 $\prefix = \fractionalElementary$ and $\prefix' = \fractionalElementary[']$.

 % We show $\prefix \NCellSim^{\SuffixSet', \Rename'}_{\targetLg} \prefix'$, which contradicts the assumption.
 Let $\fractionalElementary[_\action] = \successor[\action]{\prefix}$ and $\fractionalElementary['_\action] = \successor[\action]{\prefix'}$.
 When we have $\successor[\action]{\prefix} \NCellSim^{\SuffixSet, \Rename}_{\targetLg} \successor[\action]{\prefix'}$ at least one of the following holds.
 \begin{enumerate}
  \item There is $\word_\action \in \elementary[_\action]$ such that for any $\word'_\action \in \elementary['_\action]$, we have $\timeValuation{\word_\action},\timeValuation{\word'_\action}\not\models \Rename$.\label{list:row_inequivalence:left_to_right}
  \item There is $\word'_\action \in \elementary['_\action]$ such that for any $\word_\action \in \elementary[_\action]$, we have $\timeValuation{\word_\action},\timeValuation{\word'_\action}\not\models \Rename$.\label{list:row_inequivalence:right_to_left}
  \item There is $\elementary[''] \in \SuffixSet$ such that $\rename{\symbolicMemQ{\targetLg}(\elementary[_\action] \cdot \elementary[''])}{\Rename} \land \timedCondition'_\action$ is strictly stronger than $\symbolicMemQ{\targetLg}(\elementary['_\action] \cdot \elementary['']) \land \Rename \land \timedCondition_\action$.\label{list:row_inequivalence:left_stronger}
  \item There is $\elementary[''] \in \SuffixSet$ such that $\symbolicMemQ{\targetLg}(\elementary['_\action] \cdot \elementary['']) \land \Rename \land \timedCondition_\action$ is strictly stronger than $\rename{\symbolicMemQ{\targetLg}(\elementary[_\action] \cdot \elementary[''])}{\Rename} \land \timedCondition'_\action$.\label{list:row_inequivalence:right_stronger}
 \end{enumerate}

 When \cref{list:row_inequivalence:left_to_right} holds,
 by definition of $\elementary[_\action]$ and $\elementary['_\action]$,
 there is $\word \in \elementary$ such that $\word \cdot \action = \word_\action$ and for any $\word' \in \elementary[']$,
 we have $\timeValuation{\word \cdot \action},\timeValuation{\word' \cdot \action}\not\models\Rename$.
 For such $\word$ and $\word'$, we have
 $\timestamp_{|u_a| - 1} = \timestamp_{|u_a|}$ and $\timestamp'_{|u'_a| - 1} = \timestamp_{|u'_a|}$
 in $\timeValuation{\word \cdot \action}$ and $\timeValuation{\word' \cdot \action}$.
 Therefore,
 when we have \cref{list:row_inequivalence:left_to_right},
 there is $\word \in \elementary$ such that for any $\word' \in \elementary[']$, we have
 $\timeValuation{\word}, \timeValuation{\word'} \models \Rename'$,
 and we have $\prefix \NCellSim^{\SuffixSet', \Rename'}_{\targetLg} \prefix'$.
 Similarly,
 \cref{list:row_inequivalence:right_to_left} also implies $\prefix \NCellSim^{\SuffixSet', \Rename'}_{\targetLg} \prefix'$.

 Since $\elementary[_\action] \cdot \elementary[''] = \elementary \cdot (\action \cdot \elementary[''])$ holds,
 $\rename{\symbolicMemQ{\targetLg}(\elementary[_\action] \cdot \elementary[''])}{\Rename} \land \timedCondition'_\action$ is equivalent to
 $\rename{\symbolicMemQ{\targetLg}(\elementary \cdot (\action \cdot \elementary['']))}{\Rename'} \land \timedCondition'$ with a variable renaming.
 Similarly,
 $\symbolicMemQ{\targetLg}(\elementary['] \cdot (\action \cdot \elementary[''])) \land \Rename' \land \timedCondition$ is equivalent to
 $\symbolicMemQ{\targetLg}(\elementary['_\action] \cdot \elementary['']) \land \Rename \land \timedCondition_\action$ with a variable renaming.
 Therefore,
 when we have \cref{list:row_inequivalence:left_stronger},
 $\rename{\symbolicMemQ{\targetLg}(\elementary \cdot (\action \cdot \elementary['']))}{\Rename} \land \timedCondition'$ is also strictly stronger than $\symbolicMemQ{\targetLg}(\elementary['] \cdot (\action \cdot \elementary[''])) \land \Rename \land \timedCondition$, and we have
 $\prefix \NCellSim^{\SuffixSet', \Rename'}_{\targetLg} \prefix'$.
 Similarly,
 \cref{list:row_inequivalence:right_stronger} also implies $\prefix \NCellSim^{\SuffixSet', \Rename'}_{\targetLg} \prefix'$.
 \end{proof}
 %-----------------------------------------------------------

%%%%%%%%%%%%%%%%%%%%%%%%%%%%%%%%%%%%%%%%%%%%%%%%%%%%%%%%%%%%
%%%%%%%%%%%%%%%%%%%%%%%%%%%%%%%%%%%%%%%%%%%%%%%%%%%%%%%%%%%%
\section{Details of the implementation}\label{appendix:implementation}
%%%%%%%%%%%%%%%%%%%%%%%%%%%%%%%%%%%%%%%%%%%%%%%%%%%%%%%%%%%%
%%%%%%%%%%%%%%%%%%%%%%%%%%%%%%%%%%%%%%%%%%%%%%%%%%%%%%%%%%%%

\mw{In a journal version (if we make), we can elaborate this section}
Here, we describe some details of our prototype library \ourTool{}.
\ourTool{} is implemented in C++17 using Boost (as a general-purpose library) and Eigen.
We implemented \ourTool{} focusing on DTAs without unobservable transitions, \ie{} each edge is labeled with an event $\action \in \Alphabet$.
To construct a DTA without unobservable transitions, we require the following condition to a timed observation table in addition to the cohesion.
This can make $\PrefixSet$ larger but does not harm the termination of the learning algorithm intuitively because of the finiteness of the region graph.

\todo{Try a variant without time-saturation and compare the efficiency}
%-----------------------------------------------------------
\begin{definition}
 [time-saturation]
 A timed observation table
 $\ObsTableInside$ is \emph{time-saturated} if
 for each $\prefix \in \PrefixSet$
 $\successor[t]{\prefix} \notin \PrefixSet$ implies
 $\prefix \CellSim^{\SuffixSet, \top}_{\targetLg} \successor[t]{\prefix}$.
\end{definition}
%-----------------------------------------------------------

\paragraph{On timed conditions}
We use \emph{difference bounded matrices (DBM)}~\cite{Dill89} to represent timed conditions.
We use Eigen to implement DBMs.
A DBM represents a finite conjunction of inequalities constraints of the form
$x - x' \prec \dConstant$, $x \prec \dConstant$, or $-x \prec \dConstant$, where
 $x$ and $x'$ are variables, $\dConstant$ is a constant, and ${\prec} \in \{<, \leq\}$.
We use a DBM over $\{\sumTimestamp{0}{n}, \sumTimestamp{1}{n}, \dots, \sumTimestamp{n}{n}\}$ to encode a timed condition over $\timeVariables = \{\timestampSequence\}$.
By this encoding, we can represent any timed condition 
because $\sumTimestamp{i}{j} = \sumTimestamp{i}{n} - \sumTimestamp{j + 1}{n}$, where $0 \leq i \leq j < n$.

\paragraph{On the generation of candidate renaming equations}
As we show in \cref{appendix:row_equivalence}, 
when we check if $\elementary \CellSim^{\SuffixSet}_{\Lg} \elementary[']$,
we generate candidate renaming equations $\Rename$ and check if $\elementary \CellSim^{\SuffixSet,\Rename}_{\Lg} \elementary[']$.
As we have already mentioned, such a candidate renaming equation is generated from a bipartite graph constructed from $\timedCondition$ and $\timedCondition'$.
In our implementation, to make the DTA construction simpler, we focus on a reaming equation such that the morphism defined by $(u, \timedCondition, u',\timedCondition', \Rename)$ is a function.
More concretely, we generate renaming equations such that $\Rename = \top$ (\ie{} $\Rename$ does not constrain the value after the morphism) or $\Rename$ constrains $\sumTimestamp[']{i}{n}$ if
the value of $\sumTimestamp[']{i}{n}$ is not uniquely defined in $\timedCondition'$.
Such a restriction does not affect the termination of the learning algorithm.

\paragraph{On ``continuous consistency''}
As we observe in \cref{example:continuous-consistency},
we cannot enforce ``continuous consistency'' to timed observation tables.
Nevertheless, we can often refine $\SuffixSet$ by adding some dwell time before $\suffix \in \SuffixSet$.
In our implementation, to reduce the number of equivalence queries, when the timed observation table is continuously inconsistent, we try to add some dwell time before $\suffix \in \SuffixSet$ and refine $\SuffixSet$ if it indeed refines $\SuffixSet$.

\paragraph{On DTA construction}
% On the class of constructed DTAs
To simplify the DTA construction, we use clock renaming (\ie{} $\clock \coloneqq \clock'$ for $\clock, \clock' \in \Clock$) and constant assignments (\ie{} $\clock \coloneqq \dConstant$ for $\clock \in \Clock$ and $\dConstant \in \N$).
It is known that these extension does not change the expressive power~\cite{BDFP04} but can reduce the number of states.
% On the transition merging
We also simplify the candidate DTAs by merging transitions with the same source/target locations, the same resets, and juxtaposed guards, \ie{} the union of two guards is still a guard.
Moreover, we remove the ``dead'' locations, \ie{} locations unreachable to any accepting locations, by a zone-based analysis.
Such simplification makes the constructed DTAs more interpretable by a human.

As we observe, \eg{} in \cref{figure:DTA_learning:final_hypothesis},
our DTA learning algorithm generates a DTA with redundant clock variables due to the DTA construction in~\cite{MP04}.
It is a future work to apply clock reduction techniques~\cite{DBLP:conf/hybrid/SaeedloeiK18,DBLP:conf/concur/GuhaNA14} to generate a DTA with fewer clock variables.

%%%%%%%%%%%%%%%%%%%%%%%%%%%%%%%%%%%%%%%%%%%%%%%%%%%%%%%%%%%%
%%%%%%%%%%%%%%%%%%%%%%%%%%%%%%%%%%%%%%%%%%%%%%%%%%%%%%%%%%%%
\section{Detail of the benchmarks}
%%%%%%%%%%%%%%%%%%%%%%%%%%%%%%%%%%%%%%%%%%%%%%%%%%%%%%%%%%%%
%%%%%%%%%%%%%%%%%%%%%%%%%%%%%%%%%%%%%%%%%%%%%%%%%%%%%%%%%%%%
Here, we present the detail of the benchmarks.

%%%%%%%%%%%%%%%%%%%%%%%%%%%%%%%%%%%%%%%%%%%%%%%%%%%%%%%%%%%%
\subsection{\Random{}}
%%%%%%%%%%%%%%%%%%%%%%%%%%%%%%%%%%%%%%%%%%%%%%%%%%%%%%%%%%%%
%-----------------------------------------------------------
\begin{table}[tbp]
 \caption{Summary of the complexity of each DTA in \Random{}. $|\Loc|$ is the number of locations, $|\Alphabet|$ is the alphabet size, and $K_{\Clock}$ is the maximum constant in the guards.}%
 \label{table:benchmark:random:complexity}
 \scriptsize
 \begin{minipage}[t]{.45\textwidth}
  \centering
  \begin{tabular}{lrrr}
\toprule
 & $|\Loc|$ & $|\Alphabet|$ & $K_{\Clock}$ \\
\midrule
3\_2\_10/3\_2\_10-1 & 3 & 2 & 7 \\
3\_2\_10/3\_2\_10-10 & 3 & 2 & 9 \\
3\_2\_10/3\_2\_10-2 & 3 & 2 & 8 \\
3\_2\_10/3\_2\_10-3 & 3 & 2 & 10 \\
3\_2\_10/3\_2\_10-4 & 3 & 2 & 7 \\
3\_2\_10/3\_2\_10-5 & 3 & 2 & 5 \\
3\_2\_10/3\_2\_10-6 & 3 & 2 & 10 \\
3\_2\_10/3\_2\_10-7 & 3 & 2 & 6 \\
3\_2\_10/3\_2\_10-8 & 3 & 2 & 10 \\
3\_2\_10/3\_2\_10-9 & 3 & 2 & 9 \\
4\_2\_10/4\_2\_10-1 & 4 & 2 & 9 \\
4\_2\_10/4\_2\_10-10 & 4 & 2 & 10 \\
4\_2\_10/4\_2\_10-2 & 4 & 2 & 9 \\
4\_2\_10/4\_2\_10-3 & 4 & 2 & 10 \\
4\_2\_10/4\_2\_10-4 & 4 & 2 & 9 \\
4\_2\_10/4\_2\_10-5 & 4 & 2 & 10 \\
4\_2\_10/4\_2\_10-6 & 4 & 2 & 8 \\
4\_2\_10/4\_2\_10-7 & 4 & 2 & 7 \\
4\_2\_10/4\_2\_10-8 & 4 & 2 & 6 \\
4\_2\_10/4\_2\_10-9 & 4 & 2 & 10 \\
4\_4\_20/4\_4\_20-1 & 4 & 4 & 20 \\
4\_4\_20/4\_4\_20-10 & 4 & 4 & 20 \\
4\_4\_20/4\_4\_20-2 & 4 & 4 & 20 \\
4\_4\_20/4\_4\_20-3 & 4 & 4 & 20 \\
4\_4\_20/4\_4\_20-4 & 4 & 4 & 20 \\
4\_4\_20/4\_4\_20-5 & 4 & 4 & 18 \\
4\_4\_20/4\_4\_20-6 & 4 & 4 & 19 \\
4\_4\_20/4\_4\_20-7 & 4 & 4 & 20 \\
4\_4\_20/4\_4\_20-8 & 4 & 4 & 20 \\
4\_4\_20/4\_4\_20-9 & 4 & 4 & 19 \\
\bottomrule
\end{tabular}
  
 \end{minipage}
 \hfill
 \begin{minipage}[t]{.45\textwidth}
  \centering
  \begin{tabular}{lrrr}
\toprule
 & $|\Loc|$ & $|\Alphabet|$ & $K_{\Clock}$ \\
\midrule
5\_2\_10/5\_2\_10-1 & 5 & 2 & 9 \\
5\_2\_10/5\_2\_10-10 & 5 & 2 & 10 \\
5\_2\_10/5\_2\_10-2 & 5 & 2 & 10 \\
5\_2\_10/5\_2\_10-3 & 5 & 2 & 10 \\
5\_2\_10/5\_2\_10-4 & 5 & 2 & 10 \\
5\_2\_10/5\_2\_10-5 & 5 & 2 & 10 \\
5\_2\_10/5\_2\_10-6 & 5 & 2 & 9 \\
5\_2\_10/5\_2\_10-7 & 5 & 2 & 8 \\
5\_2\_10/5\_2\_10-8 & 5 & 2 & 10 \\
5\_2\_10/5\_2\_10-9 & 5 & 2 & 9 \\
6\_2\_10/6\_2\_10-1 & 5 & 2 & 8 \\
6\_2\_10/6\_2\_10-10 & 5 & 2 & 10 \\
6\_2\_10/6\_2\_10-2 & 5 & 2 & 7 \\
6\_2\_10/6\_2\_10-3 & 5 & 2 & 10 \\
6\_2\_10/6\_2\_10-4 & 5 & 2 & 10 \\
6\_2\_10/6\_2\_10-5 & 5 & 2 & 10 \\
6\_2\_10/6\_2\_10-6 & 5 & 2 & 10 \\
6\_2\_10/6\_2\_10-7 & 5 & 2 & 9 \\
6\_2\_10/6\_2\_10-8 & 5 & 2 & 10 \\
6\_2\_10/6\_2\_10-9 & 5 & 2 & 10 \\
\bottomrule
\end{tabular}
  
 \end{minipage}
\end{table}
%-----------------------------------------------------------

\Random{} is a benchmark taken from~\cite{ACZZZ20}.
\Random{} consists of five classes: $3\_2\_10$,  $4\_2\_10$,  $4\_4\_20$, $5\_2\_10$, and $6\_2\_10$,
where each value of $|\Loc|\_|\Alphabet|\_M$ is the number of locations, the alphabet size, and the upper bound of the maximum constant in the guards in the DTAs, respectively.
Each class consists of 10 randomly generated DTAs.
\cref{table:benchmark:random:complexity} shows the summary of the complexity of each DTA.\@
The maximum constant $K_{\Clock}$ in the guards can be smaller than its upper bound due to its random construction.
%\todo{Make figures of the target and the learned DTAs}

%%%%%%%%%%%%%%%%%%%%%%%%%%%%%%%%%%%%%%%%%%%%%%%%%%%%%%%%%%%%
\subsection{\unbalanced{}}
%%%%%%%%%%%%%%%%%%%%%%%%%%%%%%%%%%%%%%%%%%%%%%%%%%%%%%%%%%%%

%-----------------------------------------------------------
\begin{figure}[tbp]
  \centering
  \scalebox{.8}{\begin{tikzpicture}[node distance=4cm]
   \node[initial,accepting,state] (loc0) [align=center]{$\loc_0$};
   \node[accepting,state,node distance=10cm] (loc1) [right=of loc0, align=center]{$\loc_1$};
   \node[accepting,state] (loc2) [below right=of loc1, align=center]{$\loc_2$};
   \node[accepting,state,node distance=6cm] (loc3) [left=of loc2, align=center]{$\loc_3$};
   \node[accepting,state,node distance=6cm] (loc4) [left=of loc3, align=center]{$\loc_4$};
   \path[->]
   (loc0) edge node[above] {$\styleact{a}, \clock_0 > 1, \clock_0 < 2 / \clock \coloneqq 0$} (loc1)
   (loc1) edge node[below left] {$\styleact{a}, \clock > 1, \clock < 2 / \clock \coloneqq 0$} (loc2)
   (loc2) edge node[above] {$\styleact{a}, \clock > 1, \clock < 2 / \clock \coloneqq 0$} (loc3)
   (loc3) edge node[above] {$\styleact{a}, \clock > 1, \clock < 2 / \clock \coloneqq 0$} (loc4)
   (loc4) edge node[right] {$\styleact{a}, \clock > 1, \clock < 2 / \clock \coloneqq 0$} (loc0)
   ;
  \end{tikzpicture}}
  \caption{The target DTA of \unbalanced{}:1}%
  \label{figure:benchmark:unbalanced:1:target}
\end{figure}
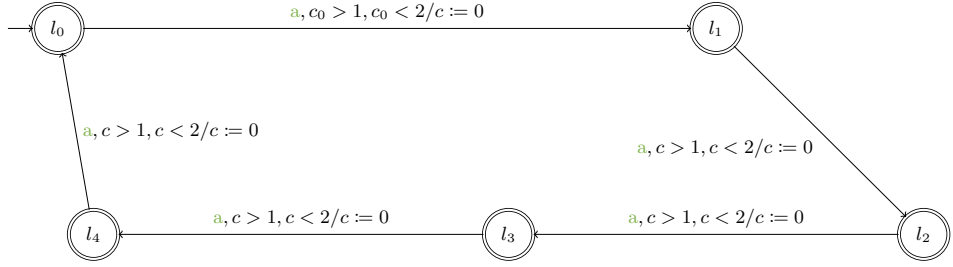
%-----------------------------------------------------------
%-----------------------------------------------------------
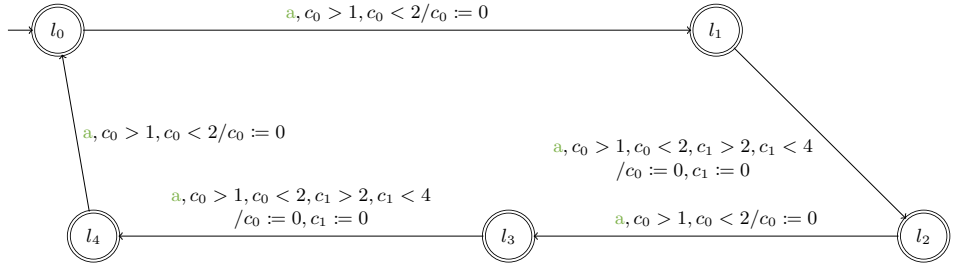
\begin{figure}[tbp]
  \centering
  \scalebox{.8}{\begin{tikzpicture}[node distance=4cm]
   \node[initial,accepting,state] (loc0) [align=center]{$\loc_0$};
   \node[accepting,state,node distance=10cm] (loc1) [right=of loc0, align=center]{$\loc_1$};
   \node[accepting,state] (loc2) [below right=of loc1, align=center]{$\loc_2$};
   \node[accepting,state,node distance=6cm] (loc3) [left=of loc2, align=center]{$\loc_3$};
   \node[accepting,state,node distance=6cm] (loc4) [left=of loc3, align=center]{$\loc_4$};
   \path[->]
   (loc0) edge node[above] {$\styleact{a}, \clock_0 > 1, \clock_0 < 2 / \clock_0 \coloneqq 0$} (loc1)
   (loc1) edge node[below left,align=center] {$\styleact{a}, \clock_0 > 1, \clock_0 < 2, \clock_1 > 2, \clock_1 < 4$\\ $/ \clock_0 \coloneqq 0, \clock_1 \coloneqq 0$} (loc2)
   (loc2) edge node[above] {$\styleact{a}, \clock_0 > 1, \clock_0 < 2 / \clock_0 \coloneqq 0$} (loc3)
   (loc3) edge node[above,align=center] {$\styleact{a}, \clock_0 > 1, \clock_0 < 2, \clock_1 > 2, \clock_1 < 4$\\ $/ \clock_0 \coloneqq 0, \clock_1 \coloneqq 0$} (loc4)
   (loc4) edge node[right] {$\styleact{a}, \clock_0 > 1, \clock_0 < 2 / \clock_0 \coloneqq 0$} (loc0)
   ;
  \end{tikzpicture}}
  \caption{The target DTA of \unbalanced{}:2}%
  \label{figure:benchmark:unbalanced:2:target}
\end{figure}
%-----------------------------------------------------------
%-----------------------------------------------------------
\begin{figure}[tbp]
  \centering
  \scalebox{.8}{\begin{tikzpicture}[node distance=4cm]
   \node[initial,accepting,state] (loc0) [align=center]{$\loc_0$};
   \node[accepting,state,node distance=10cm] (loc1) [right=of loc0, align=center]{$\loc_1$};
   \node[accepting,state] (loc2) [below right=of loc1, align=center]{$\loc_2$};
   \node[accepting,state,node distance=6cm] (loc3) [left=of loc2, align=center]{$\loc_3$};
   \node[accepting,state,node distance=6cm] (loc4) [left=of loc3, align=center]{$\loc_4$};
   \path[->]
   (loc0) edge node[above,align=center] {$\styleact{a}, \clock_0 > 1, \clock_0 < 2 / \clock_0 \coloneqq 0$} (loc1)
   (loc1) edge node[below left,align=center] {$\styleact{a}, \clock_0 > 1, \clock_0 < 2, \clock_1 > 2, \clock_1 < 4$\\ $/ \clock_0 \coloneqq 0, \clock_1 \coloneqq 0$} (loc2)
   (loc2) edge node[above,align=center] {$\styleact{a}, \clock_0 > 1, \clock_0 < 2, \clock_2 > 2, \clock_2 < 4$\\ $/ \clock_0 \coloneqq 0, \clock_2 \coloneqq 0$} (loc3)
   (loc3) edge node[above,align=center] {$\styleact{a}, \clock_0 > 1, \clock_0 < 2, \clock_1 > 2, \clock_1 < 4$\\ $/ \clock_0 \coloneqq 0, \clock_1 \coloneqq 0$} (loc4)
   (loc4) edge node[right,align=center] {$\styleact{a}, \clock_0 > 1, \clock_0 < 2, \clock_2 > 2, \clock_2 < 4$\\ $/ \clock_0 \coloneqq 0, \clock_2 \coloneqq 0$} (loc0)
   ;
  \end{tikzpicture}}
  \caption{The target DTA of \unbalanced{}:3}%
  \label{figure:benchmark:unbalanced:3:target}
\end{figure}
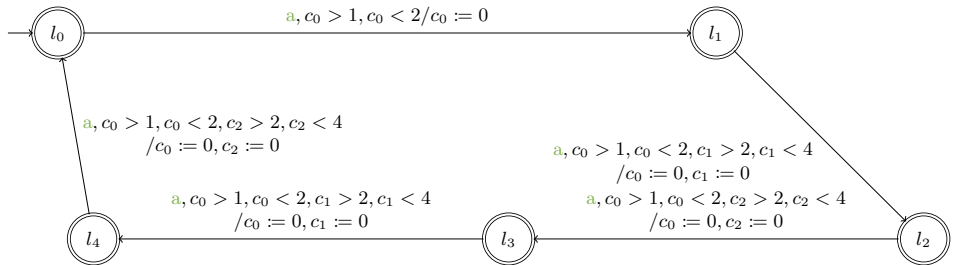
%-----------------------------------------------------------
%-----------------------------------------------------------
\begin{figure}[tbp]
  \centering
  \scalebox{.8}{\begin{tikzpicture}[node distance=4cm]
   \node[initial,accepting,state] (loc0) [align=center]{$\loc_0$};
   \node[accepting,state,node distance=10cm] (loc1) [right=of loc0, align=center]{$\loc_1$};
   \node[accepting,state] (loc2) [below right=of loc1, align=center]{$\loc_2$};
   \node[accepting,state,node distance=6cm] (loc3) [left=of loc2, align=center]{$\loc_3$};
   \node[accepting,state,node distance=6cm] (loc4) [left=of loc3, align=center]{$\loc_4$};
   \path[->]
   (loc0) edge node[above, align=center] {$\styleact{a}, \clock_0 > 1, \clock_0 < 2 / \clock_0 \coloneqq 0$} (loc1)
   (loc1) edge node[below left,align=center] {$\styleact{a}, \clock_0 > 1, \clock_0 < 2, \clock_1 > 2, \clock_1 < 4$\\ $/ \clock_0 \coloneqq 0, \clock_1 \coloneqq 0$} (loc2)
   (loc2) edge node[below=0.5cm,align=center] {$\styleact{a}, \clock_0 > 1, \clock_0 < 2, \clock_2 > 2, \clock_2 < 4, \clock_3 > 3, \clock_3 < 6$\\ $/ \clock_0 \coloneqq 0, \clock_2 \coloneqq 0, \clock_3 \coloneqq 0$} (loc3)
   (loc3) edge node[below=0.5cm,align=center] {$\styleact{a}, \clock_0 > 1, \clock_0 < 2, \clock_1 > 2, \clock_1 < 4$\\ $/ \clock_0 \coloneqq 0, \clock_1 \coloneqq 0$} (loc4)
   (loc4) edge node[right,align=center] {$\styleact{a}, \clock_0 > 1, \clock_0 < 2, \clock_2 > 2, \clock_2 < 4$\\ $/ \clock_0 \coloneqq 0, \clock_2 \coloneqq 0$} (loc0)
   ;
  \end{tikzpicture}}
  \caption{The target DTA of \unbalanced{}:4}%
  \label{figure:benchmark:unbalanced:4:target}
\end{figure}
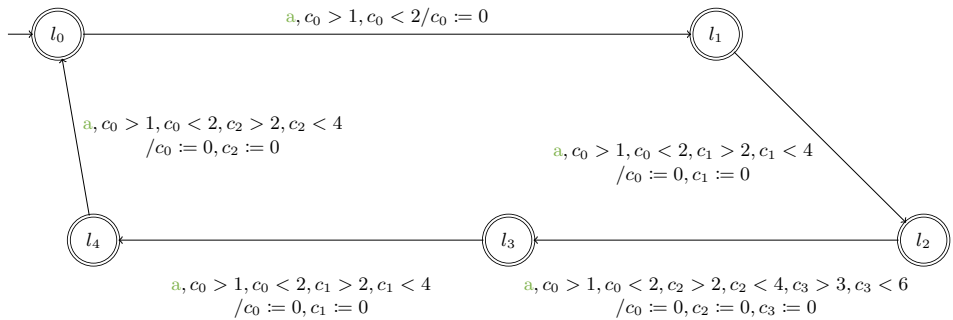
%-----------------------------------------------------------
%-----------------------------------------------------------
\begin{figure}[tbp]
  \centering
  \scalebox{.8}{\begin{tikzpicture}[node distance=4cm]
   \node[initial,accepting,state] (loc0) [align=center]{$\loc_0$};
   \node[accepting,state,node distance=10cm] (loc1) [right=of loc0, align=center]{$\loc_1$};
   \node[accepting,state] (loc2) [below right=of loc1, align=center]{$\loc_2$};
   \node[accepting,state,node distance=6cm] (loc3) [left=of loc2, align=center]{$\loc_3$};
   \node[accepting,state,node distance=6cm] (loc4) [left=of loc3, align=center]{$\loc_4$};
   \path[->]
   (loc0) edge node[above, align=center] {$\styleact{a}, \clock_0 > 1, \clock_0 < 2 / \clock_0 \coloneqq 0$} (loc1)
   (loc1) edge node[below left,align=center] {$\styleact{a}, \clock_0 > 1, \clock_0 < 2, \clock_1 > 2, \clock_1 < 4$\\ $/ \clock_0 \coloneqq 0, \clock_1 \coloneqq 0$} (loc2)
   (loc2) edge node[below=0.5cm,align=center] {$\styleact{a}, \clock_0 > 1, \clock_0 < 2, \clock_2 > 2, \clock_2 < 4, \clock_3 > 3, \clock_3 < 6$\\ $/ \clock_0 \coloneqq 0, \clock_2 \coloneqq 0, \clock_3 \coloneqq 0$} (loc3)
   (loc3) edge node[below=0.5cm,align=center] {$\styleact{a}, \clock_0 > 1, \clock_0 < 2, \clock_1 > 2, \clock_1 < 4, \clock_4 > 3, \clock_4 < 6$\\ $/ \clock_0 \coloneqq 0, \clock_1 \coloneqq 0, \clock_4 \coloneqq 0$} (loc4)
   (loc4) edge node[right,align=center] {$\styleact{a}, \clock_0 > 1, \clock_0 < 2, \clock_2 > 2, \clock_2 < 4$\\ $/ \clock_0 \coloneqq 0, \clock_2 \coloneqq 0$} (loc0)
   ;
  \end{tikzpicture}}
  \caption{The target DTA of \unbalanced{}:5}%
  \label{figure:benchmark:unbalanced:5:target}
\end{figure}
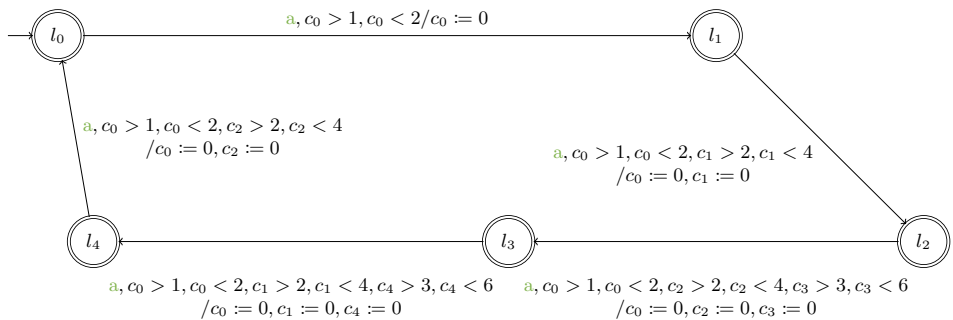
%-----------------------------------------------------------
\unbalanced{} is our original benchmark consisting of five DTAs.
\crefrange{figure:benchmark:unbalanced:1:target}{figure:benchmark:unbalanced:5:target} illustrate the target DTAs.
As shown in \crefrange{figure:benchmark:unbalanced:1:target}{figure:benchmark:unbalanced:5:target}, 
the target DTAs have the same structure with a different number of clock variables and timing constraints.

%%%%%%%%%%%%%%%%%%%%%%%%%%%%%%%%%%%%%%%%%%%%%%%%%%%%%%%%%%%%
\subsection{\AKM{}}
%%%%%%%%%%%%%%%%%%%%%%%%%%%%%%%%%%%%%%%%%%%%%%%%%%%%%%%%%%%%
\AKM{} is a benchmark on an Authentication and Key management service of the WiFi.
\AKM{} is initially used in~\cite{DBLP:conf/lata/VaandragerB021}.
We took the model from~\cite{Leslieaj/DOTALearningSMT}.

\subsection{\CAS{}}
%%%%%%%%%%%%%%%%%%%%%%%%%%%%%%%%%%%%%%%%%%%%%%%%%%%%%%%%%%%%
%-----------------------------------------------------------
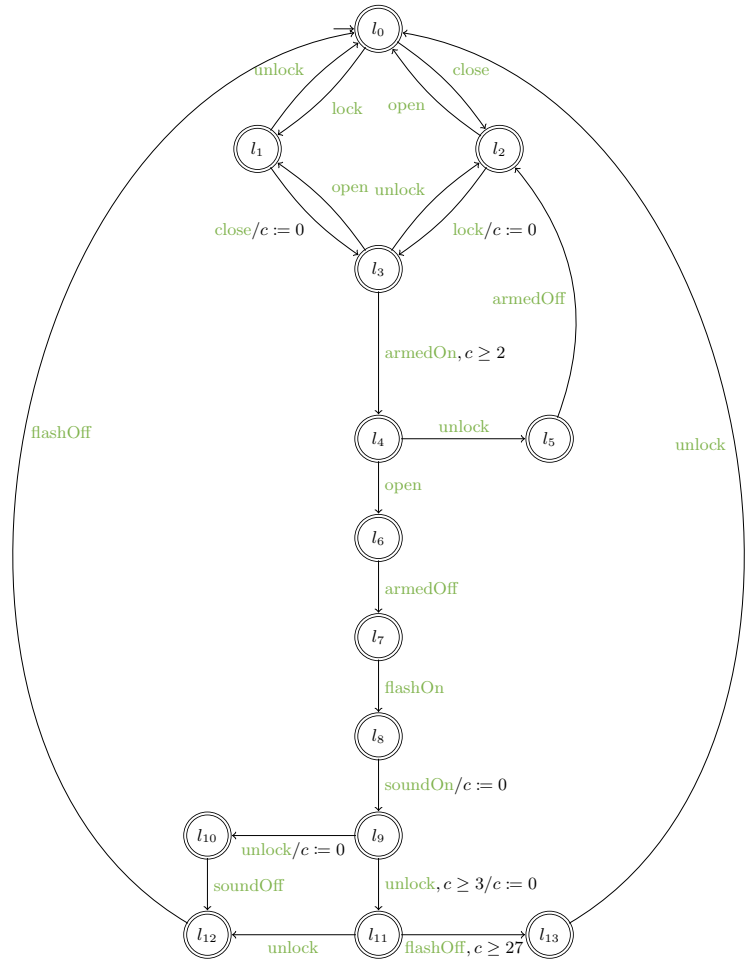
\begin{figure}[tbp]
 \centering
 \scalebox{.8}{\begin{tikzpicture}[shorten >=1pt,auto,node distance=1.0cm,scale=0.9,every node/.style={transform shape},every initial by arrow/.style={initial text={}}]
  %% locations
  \node[initial,state,accepting] (l0) [align=center]{$\loc_0$};
  \node[state,accepting,node distance=2.3cm] (l1) [below left=of l0] [align=center]{$\loc_1$};
  \node[state,accepting,node distance=2.3cm] (l2) [below right=of l0] [align=center]{$\loc_2$};
  \node[state,accepting,node distance=2.3cm] (l3) [below right=of l1] [align=center]{$\loc_3$};
  \node[state,accepting,node distance=2.3cm] (l4) [below=of l3] [align=center]{$\loc_4$};
  \node[state,accepting,node distance=2.3cm] (l5) [right=of l4] [align=center]{$\loc_5$};
  \node[state,accepting] (l6) [below=of l4] [align=center]{$\loc_6$};
  \node[state,accepting] (l7) [below=of l6] [align=center]{$\loc_7$};
  \node[state,accepting] (l8) [below=of l7] [align=center]{$\loc_8$};
  \node[state,accepting] (l9) [below=of l8] [align=center]{$\loc_9$};
  \node[state,accepting,node distance=2.3cm] (l10) [left=of l9] [align=center]{$\loc_{10}$};
  \node[state,accepting] (l11) [below=of l9] [align=center]{$\loc_{11}$};
  \node[state,accepting,node distance=2.3cm] (l12) [left=of l11] [align=center]{$\loc_{12}$};
  \node[state,accepting,node distance=2.3cm] (l13) [right=of l11] [align=center]{$\loc_{13}$};

  %% edges
  \path[->] 
  (l0) edge [bend left=10] node {$\styleact{lock}$} (l1)
  (l1) edge [bend left=10] node {$\styleact{unlock}$} (l0)
  (l0) edge [bend left=10] node {$\styleact{close}$} (l2)
  (l2) edge [bend left=10] node {$\styleact{open}$} (l0)
  (l1) edge [bend right=10] node[below left] {$\styleact{close}/ \clock \coloneqq 0$} (l3)
  (l3) edge [bend right=10] node[above right] {$\styleact{open}$} (l1)
  (l2) edge [bend left=10] node {$\styleact{lock}/ \clock \coloneqq 0$} (l3)
  (l3) edge [bend left=10] node[above left] {$\styleact{unlock}$} (l2)
  (l3) edge node {$\styleact{armedOn}, \clock \geq 2$} (l4)
  (l4) edge node {$\styleact{unlock}$} (l5)
  (l5) edge [bend right] node {$\styleact{armedOff}$} (l2)
  (l4) edge node {$\styleact{open}$} (l6)
  (l6) edge node {$\styleact{armedOff}$} (l7)
  (l7) edge node {$\styleact{flashOn}$} (l8)
  (l8) edge node {$\styleact{soundOn} / \clock \coloneqq 0$} (l9)
  (l9) edge node {$\styleact{unlock} / \clock \coloneqq 0$} (l10)
  (l9) edge node {$\styleact{unlock}, \clock \geq 3 / \clock \coloneqq 0$} (l11)
  (l10) edge node {$\styleact{soundOff}$} (l12)
  (l11) edge node {$\styleact{unlock}$} (l12)
  (l11) edge node[below] {$\styleact{flashOff}, \clock \geq 27$} (l13)
  (l12) edge [bend left=70] node[right] {$\styleact{flashOff}$} (l0)
  (l13) edge [bend right=70] node {$\styleact{unlock}$} (l0)
  ;
 \end{tikzpicture}}
 \caption{The target DTA in \CAS{}}%
 \label{figure:benchmark:CAS}
\end{figure}
%-----------------------------------------------------------
\CAS{} is a benchmark on a car alarm system.
\cref{figure:benchmark:CAS} illustrates the target DTA.\@
\CAS{} is initially used in~\cite{DBLP:journals/sigsoft/AichernigBJK11}.
We took the model from~\cite{Leslieaj/DOTALearningSMT}.

%%%%%%%%%%%%%%%%%%%%%%%%%%%%%%%%%%%%%%%%%%%%%%%%%%%%%%%%%%%%
\subsection{\light{}}
%%%%%%%%%%%%%%%%%%%%%%%%%%%%%%%%%%%%%%%%%%%%%%%%%%%%%%%%%%%%
%-----------------------------------------------------------
\begin{figure}[tbp]
 \begin{subfigure}[b]{1.0\textwidth}
 \centering
 \scalebox{1.0}{\begin{tikzpicture}[shorten >=1pt, auto,node distance=2.5cm,every initial by arrow/.style={initial text={}}]
  %% locations
  \node[initial,state,accepting] (l0) at (0,0) [align=center]{$\loc_0$};
  \node[state,accepting] (l1) [below=of l0] [align=center]{$\loc_1$};
  \node[state,accepting] (l2) [right=of l1] [align=center]{$\loc_2$};
  \node[state,accepting] (l3) [right=of l0] [align=center]{$\loc_3$};
  \node[state,accepting] (l4) [left=of l1] [align=center]{$\loc_4$};

  %% edges
  \path[->] 
  (l0) edge [bend left=10] node {$\styleact{press} / \clock \coloneqq 0$} (l1)
  (l1) edge [bend left=10] node {$\styleact{release}, \clock < 5$} (l0)
  (l1) edge node {$\styleact{starthold}, \clock \geq 10$} (l2)
  (l2) edge node {$\styleact{release}$} (l3)
  (l1) edge node[above] {$\styleact{release}, \clock \geq 5$} (l4)
  (l3) edge node {$\styleact{endhold}$} (l0)
  (l4) edge [bend left] node {$\styleact{touch}$} (l0)
  ;
 \end{tikzpicture}}
 \caption{The target DTA}%
 \label{figure:benchmark:light:target}
\end{subfigure}
%-----------------------------------------------------------
\begin{subfigure}[b]{1.0\textwidth}
 \centering
 \scalebox{.6}{\begin{tikzpicture}[shorten >=1pt, auto,node distance=9cm,every initial by arrow/.style={initial text={}}]
  %% locations
  \node[initial,state,accepting] (l0) at (0,0) [align=center]{$\loc_0$};
  \node[state,accepting,node distance=5cm] (l1) [below=of l0] [align=center]{$\loc_1$};
  \node[state,accepting] (l2) [right=of l1] [align=center]{$\loc_2$};
  \node[state,accepting] (l3) [right=of l0] [align=center]{$\loc_3$};
  \node[state,accepting] (l4) [left=of l1] [align=center]{$\loc_4$};

  %% edges
  \path[->] 
  (l0) edge [bend left=10] node[align=left] {
                $\styleact{press}, \clock_0 \leq 0 / \clock_1 \coloneqq 0$\\
                $\styleact{press}, \clock_0 > 0 / \clock_0 \coloneqq 0, \clock_1 \coloneqq 0$} (l1)
  (l1) edge node[align=center] {
                $\styleact{starthold}, \clock_0 \geq 10, \clock_0 \leq 10, \clock_1 \geq 10, \clock_1 \leq 10 / \clock_2 \coloneqq 0$\\
                $\styleact{starthold}, \clock_0 > 10, \clock_1 > 10, / \clock_0 \coloneqq 10, \clock_1 \coloneqq 10, \clock_2 \coloneqq 0$} (l2) % l2 is loc3 in the result
  (l1) edge node[above, align=center] {
                $\styleact{release}, \clock_0 \geq 5, \clock_0 \leq 5, \clock_1 \geq 5, \clock_1 \leq 5 / \clock_2 \coloneqq 0$\\
                $\styleact{release}, \clock_0 > 5, \clock_1 > 5  / \clock_0 \coloneqq 5, \clock_1 \coloneqq 5, \clock_2 \coloneqq 0$} (l4) % l4 is loc2 in the result
  (l1) edge [bend left=10] node {$\styleact{release}, \clock_0 < 5, \clock_1 < 5 / \clock_0 \coloneqq 0$} (l0) % l0 is loc0 in the result
  (l4) edge [bend left] node[above=2.5cm,align=center] {
                $\styleact{touch}, \clock_0 \geq 5, \clock_0 \leq 5, \clock_1 \geq 5, \clock_1 \leq 5, \clock_2 \leq 0 / \clock_0 \coloneqq 0$\\
                $\styleact{touch}, \clock_0 > 5, \clock_1 > 5 \clock_2 > 0 / \clock_0 \coloneqq 0, \clock_1 \coloneqq 5,  \clock_2 \coloneqq 0$} (l0)
  (l2) edge node[above left=1.5cm and 0cm,align=right] {
                $\styleact{release}, \clock_0 \geq 10, \clock_0 \leq 10, \clock_1 \geq 10, \clock_1 \leq 10, \clock_2 \leq 0 / \clock_3 \coloneqq 0$\\
                $\styleact{release}, \clock_0 > 10, \clock_1 > 10, \clock_2 > 0 / \clock_0 \coloneqq 10, \clock_1 \coloneqq 10, \clock_2 \coloneqq 0, \clock_3 \coloneqq 0$} (l3)
  (l3) edge node[above=.5cm,align=center] {
                $\styleact{endhold}, \clock_0 \geq 10, \clock_0 \leq 10, \clock_1 \geq 10, \clock_1 \leq 10, \clock_2 \leq 0, \clock_3 \leq 0 / \clock_0 \coloneqq 0$\\
                $\styleact{endhold}, \clock_0 > 10, \clock_1 > 10, \clock_2 > 0, \clock_3 > 0 / \clock_0 \coloneqq 0, \clock_1 \coloneqq 10,  \clock_2 \coloneqq 0, \clock_3 \coloneqq 0$} (l0)
  ;
 \end{tikzpicture}}
 \caption{The learned DTA after the simplification in \cref{appendix:implementation}, \eg{} removing the ``dead'' locations. For simplicity, we displace the edges with the same source and target locations on the same edge.}%
 \label{figure:benchmark:light:result}
\end{subfigure}
 \caption{\light{}}%
 \label{figure:benchmark:light}
\end{figure}
%-----------------------------------------------------------
\light{} is a benchmark on a smart light switch.
\light{} is initially used in~\cite{APT20} inspired by~\cite{DBLP:conf/fates/HesselLNPS03}.
\cref{figure:benchmark:light:target} illustrates the target DTA.\@
We took the model from~\cite{Leslieaj/DOTALearningSMT}.

%%%%%%%%%%%%%%%%%%%%%%%%%%%%%%%%%%%%%%%%%%%%%%%%%%%%%%%%%%%%
\subsection{\PC{}}
\PC{} is a benchmark on a particle counter.
\PC{} is initially used in~\cite{DBLP:conf/tap/AichernigAJKKSS14}.
% \cref{figure:benchmark:light:target} illustrates the target DTAs.
We took the model from~\cite{Leslieaj/DOTALearningSMT}.

%%%%%%%%%%%%%%%%%%%%%%%%%%%%%%%%%%%%%%%%%%%%%%%%%%%%%%%%%%%%
\subsection{\TCP{}}
%%%%%%%%%%%%%%%%%%%%%%%%%%%%%%%%%%%%%%%%%%%%%%%%%%%%%%%%%%%%
\TCP{} is a benchmark on a functional specification of TCP protocol.
\TCP{} is initially used in~\cite{ACZZZ20}.
We took the model from~\cite{Leslieaj/DOTALearningSMT}.

%%%%%%%%%%%%%%%%%%%%%%%%%%%%%%%%%%%%%%%%%%%%%%%%%%%%%%%%%%%%
\subsection{\Train{}}
%%%%%%%%%%%%%%%%%%%%%%%%%%%%%%%%%%%%%%%%%%%%%%%%%%%%%%%%%%%%

%-----------------------------------------------------------
\begin{figure}[tbp]
 \begin{subfigure}{1.0\textwidth}
  \centering
  \scalebox{.9}{\begin{tikzpicture}[node distance=2.7cm]
   %% locations
   \node[initial,accepting,state] (loc2) [align=center]{$\loc_{2}$};
   \node[accepting,state] (loc5) [right=of loc2, align=center]{$\loc_{5}$};
   \node[accepting,state] (loc1) [right=of loc5, align=center]{$\loc_{1}$};
   \node[accepting,state] (loc0) [above=of loc1, align=center]{$\loc_{0}$};
   \node[accepting,state] (loc3) [right=of loc1, align=center]{$\loc_{3}$};
   \node[accepting,state] (loc4) [above=of loc3, align=center]{$\loc_{4}$};
   
   %% edges
   \path[->]
   (loc2) edge node[above] {$\styleact{start} / \clock \coloneqq 0$} (loc5)
   (loc5) edge node[above] {$\styleact{appr}, \clock \geq 5 / \clock \coloneqq 0$} (loc1)
   (loc1) edge node[left] {$\styleact{enter}, \clock \geq 10 / \clock \coloneqq 0$} (loc0)
   (loc1) edge node[above] {$\styleact{stop} / \clock \coloneqq 0$} (loc3)
   (loc0) edge node[above left] {$\styleact{leave}, \clock \geq 3$} (loc2)
   (loc3) edge node[right] {$\styleact{go} / \clock \coloneqq 0$} (loc4)
   (loc4) edge node[above] {$\styleact{enter}, \clock \geq 7 / \clock \coloneqq 0$} (loc0)
   ;
  \end{tikzpicture}}%
  \caption{The target DTA}%
  \label{figure:benchmark:train:target}
 \end{subfigure}
 %-----------------------------------------------------------
 \begin{subfigure}{1.0\textwidth}
  \centering
  \scalebox{0.55}{\begin{tikzpicture}[node distance=6cm]
   \node[initial,accepting,state] (loc0) [align=center]{$\loc_{0}$};
   \node[accepting,state] (loc1) [right=of loc0, align=center]{$\loc_{1}$};
   \node[accepting,state] (loc2) [right=of loc1, align=center]{$\loc_{2}$};
   \node[accepting,state] (loc3) [right=of loc2, align=center]{$\loc_{3}$};
   \node[accepting,state] (loc4) [above=of loc3, align=center]{$\loc_{4}$};
   \node[accepting,state] (loc5) [above=of loc2, align=center]{$\loc_{5}$};
   \path[->]
   (loc0) edge node[align=center,below] {
   $\styleact{start}, \clock_{0} \leq 0 / \clock_{1} \coloneqq 0$\\
   $\styleact{start}, \clock_{0} > 0 / \clock_{0} \coloneqq 0,  \clock_{1} \coloneqq 0$} (loc1)
   (loc1) edge node [align=center,below] {
   $\styleact{appr}, \clock_{0} \geq 5, \clock_{0} \leq 5, \clock_{1} \geq 5, \clock_{1} \leq 5$\\ $/ \clock_{2} \coloneqq 0$\\
   $\styleact{appr}, \clock_{0} > 5, \clock_{1} > 5$\\ $/ \clock_{0} \coloneqq 5,  \clock_{1} \coloneqq 5,  \clock_{2} \coloneqq 0$} (loc2)
   (loc2) edge node[align=center,left] {$\styleact{enter}, \clock_{2} \geq 10, \clock_{1} \geq 15, \clock_{0} \geq 15$\\ $/ \clock_{0} \coloneqq 12,  \clock_{1} \coloneqq 12,  \clock_{2} \coloneqq 7,  \clock_{3} \coloneqq 7,  \clock_{4} \coloneqq 7,  \clock_{5} \coloneqq 0$} (loc5)
   (loc2) edge node[align=center,below] {
   $\styleact{stop}, \clock_{0} \geq 5, \clock_{0} \leq 5, \clock_{1} \geq 5, \clock_{1} \leq 5, \clock_{2} \leq 0$\\ $/ \clock_{3} \coloneqq 0$\\
   $\styleact{stop}, \clock_{2} > 0, \clock_{1} > 5, \clock_{0} > 5$\\ $/ \clock_{0} \coloneqq 5,  \clock_{1} \coloneqq 5,  \clock_{2} \coloneqq 0,  \clock_{3} \coloneqq 0$} (loc3)
   (loc3) edge node[align=center,left] {
   $\styleact{go}, \clock_{0} \geq 5, \clock_{0} \leq 5, \clock_{1} \geq 5, \clock_{1} \leq 5, \clock_{2} \leq 0, \clock_{3} \leq 0$\\ $/ \clock_{4} \coloneqq 0$\\
   $\styleact{go}, \clock_{0} > 5, \clock_{1} > 5, \clock_{2} > 0, \clock_{3} > 0$\\ $/ \clock_{0} \coloneqq 5,  \clock_{1} \coloneqq 5,  \clock_{2} \coloneqq 0,  \clock_{3} \coloneqq 0,  \clock_{4} \coloneqq 0$} (loc4)
   (loc4) edge node[align=center,above=.5cm] {
   $\styleact{enter}, \clock_{0} \geq 12, \clock_{0} \leq 12, \clock_{1} \geq 12, \clock_{1} \leq 12, \clock_{2} \geq 7, \clock_{2} \leq 7,$\\ $\clock_{3} \geq 7, \clock_{3} \leq 7, \clock_{4} \geq 7, \clock_{4} \leq 7$\\ $/ \clock_{5} \coloneqq 0$\\
   $\styleact{enter}, \clock_{0} > 12, \clock_{1} > 12, \clock_{2} > 7, \clock_{3} > 7, \clock_{4} > 7$\\ $/ \clock_{0} \coloneqq 12,  \clock_{1} \coloneqq 12,  \clock_{2} \coloneqq 7,  \clock_{3} \coloneqq 7,  \clock_{4} \coloneqq 7,  \clock_{5} \coloneqq 0$} (loc5)
   (loc5) edge node[align=center,above left] {
   $\styleact{leave}, \clock_{0} \geq 15, \clock_{0} \leq 15, \clock_{1} \geq 15, \clock_{1} \leq 15,$\\ $\clock_{2} \geq 10, \clock_{2} \leq 10, \clock_{3} \geq 10, \clock_{3} \leq 10, \clock_{4} \geq 10, \clock_{4} \leq 10,$\\ $\clock_{5} \geq 3, \clock_{5} \leq 3/ \clock_{0} \coloneqq 0$\\
   $\styleact{leave}, \clock_{0} > 15, \clock_{1} > 15, \clock_{2} > 10, \clock_{3} > 10, \clock_{4} > 10, \clock_{5} > 3$\\ $/ \clock_{0} \coloneqq 0,  \clock_{1} \coloneqq 15,  \clock_{2} \coloneqq 10,  \clock_{3} \coloneqq 10,  \clock_{4} \coloneqq 10,  \clock_{5} \coloneqq 3$} (loc0)
   ;
  \end{tikzpicture}}
 \caption{The learned DTA after the simplification in \cref{appendix:implementation}, \eg{} removing the ``dead'' locations. For simplicity, we displace the edges with the same source and target locations on the same edge.}%
 \label{figure:benchmark:train:result}
\end{subfigure}
 \caption{\Train{}}%
 \label{figure:benchmark:train}
\end{figure}
%-----------------------------------------------------------
\Train{} is an abstract model of a train.
\cref{figure:benchmark:train:target} illustrates the target DTA.\@
\Train{} is initially used in~\cite{TALL19}.
% \cref{figure:benchmark:light:target} illustrates the target DTAs.
We took the model from~\cite{Leslieaj/DOTALearningSMT}.

%%%%%%%%%%%%%%%%%%%%%%%%%%%%%%%%%%%%%%%%%%%%%%%%%%%%%%%%%%%%
\subsection{\FDDI{}}
%%%%%%%%%%%%%%%%%%%%%%%%%%%%%%%%%%%%%%%%%%%%%%%%%%%%%%%%%%%%
\FDDI{} is a benchmark of an abstract model of the FDDI protocol.
The model is initially used in~\cite{DBLP:conf/hybrid/DawsOTY95}.
We took a model from TChecker~\cite{ticktac-project/tchecker} with two processes.
The TChecker model of \FDDI{} can be obtained by executing \verb|fddi.sh 2 5 1 0|, where
 \verb|fddi.sh| is distributed in
\url{https://github.com/ticktac-project/tchecker/blob/master/examples/fddi.sh}.

%%%%%%%%%%%%%%%%%%%%%%%%%%%%%%%%%%%%%%%%%%%%%%%%%%%%%%%%%%%%
%%%%%%%%%%%%%%%%%%%%%%%%%%%%%%%%%%%%%%%%%%%%%%%%%%%%%%%%%%%%
\section{Detailed experiment results}
%%%%%%%%%%%%%%%%%%%%%%%%%%%%%%%%%%%%%%%%%%%%%%%%%%%%%%%%%%%%
%%%%%%%%%%%%%%%%%%%%%%%%%%%%%%%%%%%%%%%%%%%%%%%%%%%%%%%%%%%%

%-----------------------------------------------------------
\begin{table}[tbp]
 \caption{Detailed experiment results}%
 \label{table:experiment_results:full:1}
 \scriptsize
 \centering
 \begin{tabular}{llrrrrrr}
\toprule
 &  & \# of Mem.\ queries & \# of Sym.\ Mem.\ queries & \# of Eq.\ queries & Exec.\ time [sec.] & $|\Loc|$ & $|\Clock|$ \\
\midrule
\multirow[c]{2}{*}{3\_2\_10/3\_2\_10-1} & \ourTool{} & 3897 & 2144 & 5 & 1.01e-01 & 4 & 3 \\
 & \DOTA{} & 59 & N/A & 6 & 8.23e-01 & 3 & 1 \\
\multirow[c]{2}{*}{3\_2\_10/3\_2\_10-10} & \ourTool{} & 22402 & 5259 & 6 & 6.27e-01 & 3 & 3 \\
 & \DOTA{} & 468 & N/A & 11 & 9.58e-01 & 3 & 1 \\
\multirow[c]{2}{*}{3\_2\_10/3\_2\_10-2} & \ourTool{} & 4192 & 3196 & 7 & 7.70e-02 & 4 & 3 \\
 & \DOTA{} & 216 & N/A & 8 & 1.18e-01 & 3 & 1 \\
\multirow[c]{2}{*}{3\_2\_10/3\_2\_10-3} & \ourTool{} & 35268 & 15341 & 11 & 2.32e+00 & 16 & 4 \\
 & \DOTA{} & 263 & N/A & 12 & 2.39e-01 & 3 & 1 \\
\multirow[c]{2}{*}{3\_2\_10/3\_2\_10-4} & \ourTool{} & 4148 & 2516 & 5 & 6.90e-02 & 6 & 3 \\
 & \DOTA{} & 171 & N/A & 8 & 9.71e-02 & 3 & 1 \\
\multirow[c]{2}{*}{3\_2\_10/3\_2\_10-5} & \ourTool{} & 2830 & 1532 & 5 & 9.20e-02 & 7 & 3 \\
 & \DOTA{} & 32 & N/A & 5 & 6.58e-02 & 3 & 1 \\
\multirow[c]{2}{*}{3\_2\_10/3\_2\_10-6} & \ourTool{} & 28643 & 9315 & 9 & 1.25e+00 & 4 & 4 \\
 & \DOTA{} & 175 & N/A & 9 & 1.15e-01 & 3 & 1 \\
\multirow[c]{2}{*}{3\_2\_10/3\_2\_10-7} & \ourTool{} & 31236 & 12704 & 9 & 1.93e+00 & 11 & 4 \\
 & \DOTA{} & 387 & N/A & 13 & 2.80e-01 & 3 & 1 \\
\multirow[c]{2}{*}{3\_2\_10/3\_2\_10-8} & \ourTool{} & 6914 & 3181 & 5 & 1.71e-01 & 3 & 2 \\
 & \DOTA{} & 188 & N/A & 9 & 1.03e-01 & 3 & 1 \\
\multirow[c]{2}{*}{3\_2\_10/3\_2\_10-9} & \ourTool{} & 2887 & 1850 & 4 & 4.50e-02 & 3 & 2 \\
 & \DOTA{} & 91 & N/A & 7 & 9.23e-02 & 3 & 1 \\
\multirow[c]{2}{*}{4\_2\_10/4\_2\_10-1} & \ourTool{} & 47914 & 23601 & 9 & 7.71e+00 & 13 & 3 \\
 & \DOTA{} & 271 & N/A & 14 & 1.50e-01 & 4 & 1 \\
\multirow[c]{2}{*}{4\_2\_10/4\_2\_10-10} & \ourTool{} & 10619 & 6186 & 4 & 4.88e-01 & 12 & 4 \\
 & \DOTA{} & 396 & N/A & 10 & 3.53e-01 & 4 & 1 \\
\multirow[c]{2}{*}{4\_2\_10/4\_2\_10-2} & \ourTool{} & 14258 & 9310 & 9 & 2.17e+00 & 14 & 3 \\
 & \DOTA{} & 257 & N/A & 7 & 1.27e-01 & 4 & 1 \\
\multirow[c]{2}{*}{4\_2\_10/4\_2\_10-3} & \ourTool{} & 29445 & 8279 & 8 & 9.55e-01 & 9 & 3 \\
 & \DOTA{} & 262 & N/A & 13 & 1.44e-01 & 4 & 1 \\
\multirow[c]{2}{*}{4\_2\_10/4\_2\_10-4} & \ourTool{} & 194442 & 69285 & 14 & 2.48e+01 & 26 & 4 \\
 & \DOTA{} & 776 & N/A & 14 & 2.45e-01 & 4 & 1 \\
\multirow[c]{2}{*}{4\_2\_10/4\_2\_10-5} & \ourTool{} & 22111 & 6815 & 7 & 1.06e+00 & 5 & 4 \\
 & \DOTA{} & 366 & N/A & 14 & 1.77e-01 & 4 & 1 \\
\multirow[c]{2}{*}{4\_2\_10/4\_2\_10-6} & \ourTool{} & 77864 & 30636 & 8 & 2.65e+01 & 14 & 4 \\
 & \DOTA{} & 985 & N/A & 10 & 3.03e-01 & 4 & 1 \\
\multirow[c]{2}{*}{4\_2\_10/4\_2\_10-7} & \ourTool{} & 98232 & 26569 & 6 & 9.18e+00 & 15 & 4 \\
 & \DOTA{} & 272 & N/A & 13 & 1.65e-01 & 4 & 1 \\
\multirow[c]{2}{*}{4\_2\_10/4\_2\_10-8} & \ourTool{} & 28427 & 10235 & 8 & 1.78e+00 & 14 & 3 \\
 & \DOTA{} & 255 & N/A & 10 & 1.50e-01 & 4 & 1 \\
\multirow[c]{2}{*}{4\_2\_10/4\_2\_10-9} & \ourTool{} & 36656 & 16403 & 6 & 5.15e+00 & 14 & 3 \\
 & \DOTA{} & 670 & N/A & 16 & 2.76e-01 & 4 & 1 \\
\multirow[c]{2}{*}{4\_4\_20/4\_4\_20-1} & \ourTool{} & 613498 & 209071 & 21 & 5.90e+01 & 14 & 4 \\
 & \DOTA{} & 5329 & N/A & 35 & 1.54e+00 & 4 & 1 \\
\multirow[c]{2}{*}{4\_4\_20/4\_4\_20-10} & \ourTool{} & 670040 & 304827 & 15 & 2.59e+02 & 26 & 4 \\
 & \DOTA{} & 4140 & N/A & 42 & 2.19e+00 & 4 & 1 \\
\multirow[c]{2}{*}{4\_4\_20/4\_4\_20-2} & \ourTool{} & T/O & T/O & T/O & T/O & T/O & T/O \\
 & \DOTA{} & 4239 & N/A & 37 & 1.65e+00 & 4 & 1 \\
\multirow[c]{2}{*}{4\_4\_20/4\_4\_20-3} & \ourTool{} & T/O & T/O & T/O & T/O & T/O & T/O \\
 & \DOTA{} & 3961 & N/A & 29 & 1.60e+00 & 4 & 1 \\
\multirow[c]{2}{*}{4\_4\_20/4\_4\_20-4} & \ourTool{} & 933187 & 331561 & 20 & 2.42e+02 & 22 & 3 \\
 & \DOTA{} & 1740 & N/A & 26 & 1.18e+00 & 4 & 1 \\
\multirow[c]{2}{*}{4\_4\_20/4\_4\_20-5} & \ourTool{} & 1055910 & 373123 & 14 & 4.94e+02 & 35 & 3 \\
 & \DOTA{} & 1966 & N/A & 26 & 8.27e-01 & 4 & 1 \\
\multirow[c]{2}{*}{4\_4\_20/4\_4\_20-6} & \ourTool{} & 1048528 & 492591 & 18 & 6.70e+02 & 34 & 3 \\
 & \DOTA{} & 3013 & N/A & 32 & 1.62e+00 & 4 & 1 \\
\multirow[c]{2}{*}{4\_4\_20/4\_4\_20-7} & \ourTool{} & 248399 & 125653 & 10 & 3.23e+01 & 18 & 3 \\
 & \DOTA{} & 2652 & N/A & 30 & 1.19e+00 & 4 & 1 \\
\multirow[c]{2}{*}{4\_4\_20/4\_4\_20-8} & \ourTool{} & 618743 & 186476 & 14 & 1.15e+03 & 11 & 3 \\
 & \DOTA{} & 3863 & N/A & 36 & 1.40e+00 & 4 & 1 \\
\multirow[c]{2}{*}{4\_4\_20/4\_4\_20-9} & \ourTool{} & 1681769 & 520310 & 11 & 8.34e+03 & 34 & 5 \\
 & \DOTA{} & 4074 & N/A & 35 & 1.01e+00 & 4 & 1 \\
\bottomrule
\end{tabular}

\end{table}
%-----------------------------------------------------------

%-----------------------------------------------------------
\begin{table}[tbp]
 \caption{Detailed experiment results}%
 \label{table:experiment_results:full:2}
 \scriptsize
 \centering
 \begin{tabular}{llrrrrrr}
\toprule
 &  & \# of Mem.\ queries & \# of Sym.\ Mem.\ queries & \# of Eq.\ queries & Exec.\ time [sec.] & $|\Loc|$ & $|\Clock|$ \\
\midrule
\multirow[c]{2}{*}{5\_2\_10/5\_2\_10-1} & \ourTool{} & 25038 & 10299 & 7 & 7.85e-01 & 8 & 3 \\
 & \DOTA{} & 1332 & N/A & 15 & 5.20e-01 & 5 & 1 \\
\multirow[c]{2}{*}{5\_2\_10/5\_2\_10-10} & \ourTool{} & 21615 & 10292 & 10 & 7.84e-01 & 7 & 3 \\
 & \DOTA{} & 1093 & N/A & 22 & 4.04e-01 & 5 & 1 \\
\multirow[c]{2}{*}{5\_2\_10/5\_2\_10-2} & \ourTool{} & 32562 & 13085 & 6 & 2.41e+00 & 13 & 5 \\
 & \DOTA{} & 359 & N/A & 12 & 2.58e-01 & 4 & 1 \\
\multirow[c]{2}{*}{5\_2\_10/5\_2\_10-3} & \ourTool{} & 37725 & 18962 & 8 & 6.62e+00 & 12 & 6 \\
 & \DOTA{} & 652 & N/A & 16 & 2.84e-01 & 5 & 1 \\
\multirow[c]{2}{*}{5\_2\_10/5\_2\_10-4} & \ourTool{} & 8121 & 4597 & 5 & 1.96e-01 & 5 & 3 \\
 & \DOTA{} & 778 & N/A & 12 & 4.30e-01 & 4 & 1 \\
\multirow[c]{2}{*}{5\_2\_10/5\_2\_10-5} & \ourTool{} & 334000 & 63946 & 12 & 1.67e+02 & 23 & 4 \\
 & \DOTA{} & 752 & N/A & 15 & 3.01e-01 & 5 & 1 \\
\multirow[c]{2}{*}{5\_2\_10/5\_2\_10-6} & \ourTool{} & 627980 & 69007 & 19 & 4.36e+01 & 8 & 4 \\
 & \DOTA{} & 1159 & N/A & 16 & 4.55e-01 & 5 & 1 \\
\multirow[c]{2}{*}{5\_2\_10/5\_2\_10-7} & \ourTool{} & 33157 & 14270 & 8 & 2.57e+00 & 11 & 5 \\
 & \DOTA{} & 639 & N/A & 16 & 3.31e-01 & 5 & 1 \\
\multirow[c]{2}{*}{5\_2\_10/5\_2\_10-8} & \ourTool{} & 56059 & 15114 & 6 & 2.17e+00 & 8 & 5 \\
 & \DOTA{} & 896 & N/A & 20 & 3.41e-01 & 5 & 1 \\
\multirow[c]{2}{*}{5\_2\_10/5\_2\_10-9} & \ourTool{} & 22807 & 12003 & 8 & 1.97e+00 & 17 & 6 \\
 & \DOTA{} & 1100 & N/A & 18 & 3.35e-01 & 5 & 1 \\
\multirow[c]{2}{*}{6\_2\_10/6\_2\_10-1} & \ourTool{} & 2912 & 1976 & 7 & 4.40e-02 & 4 & 2 \\
 & \DOTA{} & 104 & N/A & 11 & 1.73e-01 & 4 & 1 \\
\multirow[c]{2}{*}{6\_2\_10/6\_2\_10-10} & \ourTool{} & 100795 & 23336 & 9 & 4.67e+00 & 8 & 5 \\
 & \DOTA{} & 2122 & N/A & 16 & 8.46e-01 & 6 & 1 \\
\multirow[c]{2}{*}{6\_2\_10/6\_2\_10-2} & \ourTool{} & 3859 & 3060 & 6 & 1.08e-01 & 5 & 5 \\
 & \DOTA{} & 910 & N/A & 14 & 6.89e-01 & 5 & 1 \\
\multirow[c]{2}{*}{6\_2\_10/6\_2\_10-3} & \ourTool{} & 36219 & 21419 & 8 & 3.62e+00 & 12 & 7 \\
 & \DOTA{} & 1445 & N/A & 28 & 7.68e-01 & 6 & 1 \\
\multirow[c]{2}{*}{6\_2\_10/6\_2\_10-4} & \ourTool{} & 26780 & 9860 & 11 & 9.83e-01 & 6 & 4 \\
 & \DOTA{} & 3124 & N/A & 27 & 8.77e-01 & 6 & 1 \\
\multirow[c]{2}{*}{6\_2\_10/6\_2\_10-5} & \ourTool{} & 555939 & 70427 & 12 & 2.44e+02 & 16 & 5 \\
 & \DOTA{} & 2593 & N/A & 16 & 8.27e-01 & 6 & 1 \\
\multirow[c]{2}{*}{6\_2\_10/6\_2\_10-6} & \ourTool{} & 73012 & 28542 & 12 & 1.13e+01 & 7 & 5 \\
 & \DOTA{} & 900 & N/A & 19 & 5.05e-01 & 5 & 1 \\
\multirow[c]{2}{*}{6\_2\_10/6\_2\_10-7} & \ourTool{} & 208647 & 26577 & 10 & 1.37e+01 & 11 & 5 \\
 & \DOTA{} & 1748 & N/A & 14 & 6.21e-01 & 6 & 1 \\
\multirow[c]{2}{*}{6\_2\_10/6\_2\_10-8} & \ourTool{} & 33971 & 11466 & 10 & 1.19e+00 & 6 & 4 \\
 & \DOTA{} & 2073 & N/A & 28 & 9.83e-01 & 6 & 1 \\
\multirow[c]{2}{*}{6\_2\_10/6\_2\_10-9} & \ourTool{} & 22651 & 12691 & 14 & 1.95e+00 & 6 & 3 \\
 & \DOTA{} & 3929 & N/A & 35 & 1.72e+00 & 6 & 1 \\
\multirow[c]{2}{*}{\AKM{}} & \ourTool{} & 12263 & 10881 & 11 & 5.85e-01 & 12 & 7 \\
 & \DOTA{} & 3453 & N/A & 49 & 7.97e+00 & 12 & 1 \\
\multirow[c]{2}{*}{\CAS{}} & \ourTool{} & 66067 & 42611 & 17 & 4.65e+00 & 14 & 10 \\
 & \DOTA{} & 4769 & N/A & 18 & 9.58e+01 & 14 & 1 \\
\FDDI{} & \ourTool{} & 9986271 & 2408549 & 43 & 3.00e+03 & 348 & 10 \\
\multirow[c]{2}{*}{\PC{}} & \ourTool{} & 245134 & 162997 & 23 & 6.49e+01 & 25 & 10 \\
 & \DOTA{} & 10390 & N/A & 29 & 1.24e+02 & 25 & 1 \\
\multirow[c]{2}{*}{\TCP{}} & \ourTool{} & 11300 & 11337 & 15 & 3.82e-01 & 20 & 9 \\
 & \DOTA{} & 4713 & N/A & 32 & 2.20e+01 & 20 & 1 \\
\multirow[c]{2}{*}{\Train{}} & \ourTool{} & 13487 & 8424 & 8 & 1.72e-01 & 6 & 6 \\
 & \DOTA{} & 838 & N/A & 13 & 1.13e+00 & 6 & 1 \\
\unbalanced{}:1 & \ourTool{} & 51 & 52 & 2 & 2.00e-03 & 1 & 1 \\
\unbalanced{}:2 & \ourTool{} & 576142 & 14439 & 3 & 3.64e+01 & 25 & 7 \\
\unbalanced{}:3 & \ourTool{} & 403336 & 13001 & 4 & 2.24e+01 & 25 & 7 \\
\unbalanced{}:4 & \ourTool{} & 4142835 & 46204 & 5 & 2.40e+02 & 93 & 8 \\
\unbalanced{}:5 & \ourTool{} & 10691400 & 86184 & 5 & 8.68e+02 & 140 & 9 \\
\bottomrule
\end{tabular}

\end{table}
%-----------------------------------------------------------

\cref{table:experiment_results:full:1,table:experiment_results:full:2} show the detailed experiment results, where the columns $|\Loc|$ and $|\Clock|$ show the number of locations and the clock variables of the learned DTA.\@

% %-----------------------------------------------------------
% \begin{figure}[tbp]
%  \centering
%  \begin{tikzpicture}[shorten >=1pt,scale=0.8,every node/.style={transform shape},every initial by arrow/.style={initial text={}}, node distance = 2cm]
%   %% locations
%   \node[initial,state] (l0) {$\loc_0$};
%   \node[state] (l1) [right=of l0]{$\loc_1$};
%   \node[state] (l2) [right=of l1]{$\loc_2$};
%   \node[state,accepting] (l3) [right=of l2]{$\loc_3$};

%   %% edges
%   \path[->]
%   (l0) edge [above] node {$\styleact{a}$, $x_1 \Coloneqq 0$} (l1)
%   (l1) edge [above] node {$\styleact{b}$, $x_0 = 1$} (l2)
%   (l2) edge [above] node {$\styleact{c}$, $x_1 = 1$} (l3)
%   ;
%  \end{tikzpicture}
%  \caption{\unbalanced{}}
% \end{figure}
%  %-----------------------------------------------------------

\ifdefined\VersionWithComments%
 \setcounter{tocdepth}{1} 
\listoftodos{}
\fi
\fi
\end{document}